\documentclass[11pt]{article}

\usepackage{lineno}
\usepackage{amssymb}
\usepackage{amsmath}
\usepackage{amsthm}
\usepackage{epsfig}
\usepackage{graphicx}
\usepackage{graphics}
\usepackage{float}
\usepackage{subfigure}
\usepackage{multirow}
\usepackage{color}
\usepackage{lineno}
\usepackage{fullpage}
\usepackage[normalem]{ulem} 
\usepackage{makeidx}
\usepackage{xspace}
\usepackage{wrapfig}
\usepackage{blkarray, bigstrut}
\usepackage{pdfpages}
\usepackage{imakeidx}
\usepackage{afterpage}
\usepackage{url}
\usepackage{nomencl}
\makenomenclature
\usepackage[normalem]{ulem} 
\usepackage[title]{appendix}
\usepackage{hyperref} 
\usepackage{etoolbox}
\renewcommand\nomgroup[1]{%
  \item[\bfseries
  \ifstrequal{#1}{M}{Matrices and Vectors}{%
  \ifstrequal{#1}{T}{Tensors}{%
  \ifstrequal{#1}{R}{Real Numbers}{}}}%
]}

\usepackage{xcolor}
\usepackage{smartdiagram}
\usepackage{verbatim}
\usepackage{tikz}
\usepackage{array}
\usepackage{booktabs}
\usepackage[round]{natbib}
\usepackage{authblk}
\newcommand{\ten}[1]{\mathcal{#1}}
\newcommand{\mat}[1]{\boldsymbol{#1}}
\DeclareMathOperator{\vect}{vec}
\DeclareMathOperator{\ic}{IC}

\DeclareMathOperator{\spn}{span}
\DeclareMathOperator{\rank}{rank}
\DeclareMathOperator{\tr}{tr}

\newcommand{\as}{\ensuremath{\,\text{a.s.}}}

\newcolumntype{P}[1]{>{\centering\arraybackslash}p{#1}}
\DeclareMathOperator*{\argmin}{argmin}
\newtheorem{definition}{Definition}
\newtheorem{assumption}{Assumption}
\newtheorem{proposition}{Proposition}
\newtheorem{lemma}{Lemma}
\newtheorem{corollary}{Corollary}
\newtheorem{theorem}{Theorem}
\newtheorem{remark}{Remark}

\newcommand{\h}[1]{\boldsymbol{#1}}
\newcommand{\ha}{\h{a}}
\newcommand{\hu}{\h{u}}
\newcommand{\hv}{\h{v}}
\newcommand{\hA}{\h{A}}

\newcommand{\hX}{\h{X}}

\newcommand{\cX}{{\mathcal X}}
\newcommand{\cY}{{\mathcal Y}}
\newcommand{\cE}{{\mathcal E}}

\title{\textbf{Multi-linear Tensor Autoregressive Models}}
\author{Zebang Li and Han Xiao}
\affil{Rutgers University}
\date{}

\begin{document}

\maketitle

\begin{abstract}
Contemporary time series analysis has seen more and more tensor type data, from many fields. For example, stocks can be grouped according to Size, Book-to-Market ratio, and Operating Profitability, leading to a 3-way tensor observation at each month. We propose an autoregressive model for the tensor-valued time series, with autoregressive terms depending on multi-linear coefficient matrices. Comparing with the traditional approach of vectoring the tensor observations and then applying the vector autoregressive model, the tensor autoregressive model preserves the tensor structure and admits corresponding interpretations. We introduce three estimators based on projection, least squares, and maximum likelihood. Our analysis considers both fixed dimensional and high dimensional settings. For the former we establish the central limit theorems of the estimators, and for the latter we focus on the convergence rates and the model selection. The performance of the model is demonstrated by simulated and real examples.

\end{abstract}
\newpage

\section{Introduction}\label{introduction}

In many fields and applications, multiple observations are generated and recorded with respect to the time. Traditional approach for dealing with multivariate time series data typically stacks all the observations at one time into a vector, and models the temporal dynamics of the vector data. On the other hand, very often these observations have a finer structure, and can be represented as a matrix or tensor. For example, stocks can be grouped according to the Fama-French factors Size, Book-to-Market Ratio, and Operating Profitability, leading to a 3-way tensor observation at each month. The import/export among different countries at each quarter can also be conveniently represented as a matrix. The columns/rows/tubes of the matrix/tensor actually correspond to different ways of grouping the observations. By treating these observations as a vector, one loses the grouping information and can miss a better understanding of the dynamics of the data. It is therefore interesting and important to explore the possibility and advantage of preserving the matrix/tensor form of the data. In response to such an imperative, we propose a tensor autoregressive model, which maintains the tensor form through multi-linear autoregressive terms, and admits corresponding interpretations. We also consider the model selection and estimation when the tensors are themselves of large dimensions.


There has been a surge of interest on high dimensional time series analysis. Most of them are based on sparse VAR models \citep{basu2015, davis2016, han15, kock2015, lin2017, loh2012, melnyk2016, nicholson2017}. \cite{Guo2016} studied a class of VAR with banded coefficient matrices. \cite{Basu2019} and \cite{lin2020} considered a VAR whose coefficient matrix is the sum of a sparse and a low rank matrices. \cite{hall2019} introduced the generalized VAR model.  and \cite{Ghosh2019} studied a VAR model from the Bayesian perspective.  

\cite{hoff2015multilinear} first introduced the multilinear form of the regression model for longitudinal data, \cite{ding2016matrix} studied matrix-variate regression models. For matrix and tensor time series, \cite{wang2019factor} proposed the matrix factor models, \cite{Chen2019FactorMF} and \cite{han2020tensor} considered the tensor factor models. \cite{chen2020autoregressive} introduced the autoregressive model for matrix time series, and \cite{xiao2021} studied the matrix autoregressive model (MAR) with low rank coefficient matrices. \cite{wang2021high} proposed an autoregressive model for tensor time series, with a coefficient tensor whose order is twice of that of the observations. Although this is equivalent to the VAR, they further imposed low rank conditions on the coefficient tensor.

In this paper, we generalize the MAR model proposed by \cite{chen2020autoregressive} to a tensor autoregressive model (TenAR), for the applications where the observed time series at each time point is a tensor. We also consider the TenAR model with multiple terms and multiple lagged terms, which offers a more comprehensive modeling framework for various applications. Besides, while \cite{chen2020autoregressive} and \cite{xiao2021} were only concerned with the fixed-dimensional setting, our analysis is also carried out under the high dimensional paradigm, allowing the tensor dimensions to grow with the sample size.

For the estimation of the coefficient matrices, we introduce three estimators, based on projection, least squares and likelihood respectively. Both the least squares estimator (LSE) and the maximum likelihood estimator (MLE) are obtained by iterative algorithms, alternating over the involved parameter matrices. Our empirical analysis reveals that the LSE and MLE require a good initial value for the alternating algorithm. Although the projection estimator is usually less efficient, it serves a good initializer for the other two estimation procedures. We establish the asymptotic normality for all the estimators when the tensor dimension is fixed, and show the convergence rates under the high dimensional setting.


The general TenAR model involves multiple lagged terms (referred to as the {\it order} of the model), and for each lagged term, it can have multiple multi-linear terms (referred to as the {\it K-rank}). To select the order and the K-rank, we propose an extended Bayesian information criterion, and establish the model selection consistency under both fixed and high dimensional settings.

The rest of the paper is organized as follows. 
We introduce the tensor autoregressive model in Section~\ref{Autoregressive}, discuss its basic properties, and provide some model interpretations. The estimation procedures are presented in Section~\ref{estimation}. Asymptotic properties of the estimators will be considered in Section~\ref{asymptotics}. Section~\ref{sec:selection} considers the model selection using the extended BIC. In Section~\ref{numerical}, we carry out extensive numerical studies to demonstrate the performance of the model and compare different estimators. We also apply the TenAR model to a tensor time series of Fama-French portfolios. All the proofs and some additional figures are collected in Appendix. 


\section{Autoregressive Models for Tensor-Valued Time Series}\label{Autoregressive}

\subsection{Basics about Tensor}\label{basic}

To fix notations, we briefly introduce some basic concepts and operations about tensors. For a more thorough account on various aspects of tensors, see \cite{kolda2009tensor} and \cite{ND2017}.

In this paper we use script capital letter for tensors, capital letters in boldface for matrices, and lower-case letters in boldface for vectors. A tensor $\mathcal{X}=:\{\mathcal X_{i_1 \ldots i_K}\}$ is a multidimensional array, where $1\leq i_k\leq d_K$. The number of dimensions $K$ is called the {\it order} of $\mathcal X$, also known as the number of modes. In particular, a matrix ${\hX}$ is a tensor of order 2. 
The vector $(\mathcal X_{i_1 \ldots i_K})_{1\leq i_k\leq d_k}$ with all indices except $i_k$ fixed is called a {\it mode-$k$ tube} of $\mathcal X$. In particular, for a matrix $\hX$, a column is a mode-1 fiber and a row is a mode-2 fiber. {\it Slices} are two-dimensional sections of a tensor, defined by fixing all but two indices. For example, the frontal slices of a third order tensor $\mathcal{X}$, fixing the first index, are denoted by $\mathcal{X}_{::i_3}$, $i_3=1,\ldots,d_3$.

The {\it tensor mode product} is the product of a tensor and a matrix along a mode. Specifically, suppose $\hA \in \mathbb{R}^{\tilde{d}_k \times d_k}$, the {\it mode-$k$ product} of $\mathcal X$ and $\hA$, denoted by $\mathcal{X} \times_k \hA$, is an order-$K$ tensor of dimensions $d_1 \times \cdots \times d_{k-1} \times \tilde{d}_k \times d_{k+1} \times \cdots \times d_K$, defined by $$(\mathcal{X} \times_k \hA)_{i_1 \ldots i_{k-1} j i_{k+1} \ldots i_K} = \sum_{i_k = 1}^{d_k} \cX_{i_1 \ldots i_k \ldots i_K} \hA_{j i_k}.$$

The {\it tensor generalized inner product} of two tensors $\cX \in \mathbb{R}^{d_1 \times d_2 \times \cdots \times d_K}$ and  $\cY \in \mathbb{R}^{d_1 \times d_2 \times \cdots \times d_L}$ with $K \ge L$, denoted by $\langle \cX, \cY \rangle$, is an order-$(K-L)$ tensor in $\mathbb{R}^{d_{L+1} \times \cdots \times d_K}$ defined by $$\langle \cX, \cY \rangle_{i_{L+1}\ldots i_{K}} = \sum_{i_1=1}^{d_1}\sum_{i_2=1}^{d_2}\cdots\sum_{i_L=1}^{d_L} \cX_{i_1 i_2 \ldots i_L i_{L+1} \ldots i_K} \cY_{i_1 i_2 \ldots i_L},$$
where $1 \le i_{L+1} \le d_{L+1}$, $\cdots$, $1 \le i_{K} \le d_{K}$. In particular, when $m=K$, $\langle \cX, \cY \rangle$ is called the tensor inner product. The Frobenius norm of any tensor $\cX$ is defined as $\|\cX\|_F=\sqrt{\langle \cX, \cX \rangle}$. 

Let $\ha_k=(a_{k,i_k})_{1\leq i_k\leq d_k}$ be $d_k$ dimensional vectors for $1\leq k\leq K$, the outer product of $\ha_1,\ldots,\ha_K$, denoted by $\ha_1 \circ \ha_2 \cdots \circ \ha_K$, is a $d_1\times\cdots\times d_K$ tensor whose $(i_1\ldots i_K)$-th element equals $a_{1,i_1}\cdot\cdots\cdot a_{K,i_K}$. Such a tensor, if nonzero, is called a rank-one tensor. 
The CP rank of $\mathcal{X}$ is the minimum number of rank-one tensors needed to produce $\mathcal{X}$ as their sum \citep{carroll1970analysis,harshman1970foundations}. 

\textbf{Notations}. The discussion of tensors and tensor models involves many notations. For easy references, we provide a list of notations in Appendix. Here we highlight some of them that are frequently used. Throughout the paper, $\circ$ denotes the outer product, and $\otimes$ denotes the Kronecker product.
The Frobenius norm, denoted by $\|\cdot\|_F$, can be extended from matrices to tensors, as the square root of the sum of squared entries. We use $\|\cdot\|_s$ to denote the matrix spectral norm, $\lambda_i(\cdot)$ the eigenvalues, $s_i(\cdot)$ the singular values and $\rho(\cdot)$ the spectral radius. We use $\hA_{i_1 i_2}$ and $\cX_{i_1 \ldots i_K}$ to denote the entries of the matrix/tensor. But when other indicies also appear in the subscript, we will use $\hA[i_1,i_2]$ and $\cX[i_1,\ldots,i_K]$ instead.
The notations $\propto$ denotes that two vectors/matrices/tensors are proportional to each other. For any integer $m>0$, $[m]:=\{1,2,\ldots,m\}$.

\subsection{Tensor Autogressive Models}\label{sec:tensor}

Consider a tensor time series $\{\cX_t\}$, where at each time $t$, an order-$K$ tensor $\mathcal{X}_t \in  \mathbb{R}^{d_1 \times d_2 \times \cdots \times d_K}$ is observed. We first introduce the tensor autoregressive model of the form
\begin{equation}\label{TAR1}
\mathcal{X}_t = \sum_{r=1}^R \mathcal{X}_{t-1} \times_{1}  \hA_1^{(r)} \times_{2}  \cdots \times_{K} \hA_K^{(r)} + \mathcal{E}_t,
\end{equation}
where $\hA_k^{(r)} \in \mathbb{R}^{d_k \times d_k}$ are coefficient matrices, and $\mathcal{E}_t \in  \mathbb{R}^{d_1 \times d_2 \times \cdots \times d_K}$ is a tensor white noise satisfying $\mathrm{Cov}(\cE_t,\cE_s)=\boldsymbol{0}$ whenever $s\neq t$. Note that only the lag-1 term $\cX_{t-1}$ appears on the right hand side, so we refer to \eqref{TAR1} as an order-1 model, abbreviated as TenAR(1), following the terminology of time series analysis. On the other hand, allowing multiple terms (all involving $\cX_{t-1}$) can provide more flexibility for capturing the interactions among fibers of the tensor. We refer to $R$ as the rank of the model for the reason to be discussed (see \eqref{eq:Phi_CP}). For the rank-one model with $R=1$, we will drop the superscript and denote the coefficient matrices by $\hA_k$ for simplicity. When $K=2$ and $R=1$, the TenAR(1) reduces to the matrix autoregressive model (MAR) introduced by \cite{chen2020autoregressive}.



The TenAR(1) offers a parsimonious representation of the vector autoregressive models (VAR). 
\begin{equation}
    \label{eq:var1}
      \vect(\mathcal{X}_t) = \Phi\vect(\mathcal{X}_{t-1}) +   \vect(\mathcal{E}_t).
\end{equation}
After vectorization, the model \eqref{TAR1} becomes
\begin{equation}
  \label{eq:multi_tenAR_vec}
  \vect(\mathcal{X}_t) = \left[\sum_{r=1}^{R} \hA^{(r)}_K \otimes \hA^{(r)}_{K-1} \otimes \cdots \otimes \hA^{(r)}_1 \right] \vect({\mathcal{X}_{t-1}}) + \vect(\mathcal{E}_t).
\end{equation}
In other words, the TenAR(1) model corresponds to a VAR(1) whose coefficient matrix $\Phi$ takes the form $\Phi=\sum_{r=1}^{R} \hA^{(r)}_K \otimes \hA^{(r)}_{K-1} \otimes \cdots \otimes \hA^{(r)}_1$. Note that the set of all entries in
$\hA_K \otimes \hA_{K-1} \otimes \cdots \otimes
\hA_1$ is the same as those in
$\vect(\hA_1) \circ \vect(\hA_2) \circ
\cdots \circ \vect(\hA_K)$, thus we can define a re-arrangement
operator
$\mathcal{R}: \mathbb{R}^{d\times d} \to \mathbb{R}^{d_1^2 \times
  d_2^2 \cdots \times d_K^2}$ such that
$$\mathcal{R}(\hA_{K} \otimes \hA_{K-1} \otimes \cdots \otimes \hA_1)= \vect(\hA_1) \circ \vect(\hA_2) \circ \cdots \circ \vect(\hA_K).$$
Therefore, the representation \eqref{eq:multi_tenAR_vec} indicates that the TenAR(1)
model \eqref{TAR1} can be viewed as a VAR(1) whose
coefficient matrix $\Phi$, after the rearrangement, is an order-$K$
tensor of rank $R$, i.e.
\begin{equation}
  \label{eq:Phi_CP}
  \mathcal R(\Phi) = \sum_{r=1}^R \vect\left(\hA_1^{(r)}\right) \circ
  \vect\left(\hA_{2}^{(r)}\right) \circ \cdots \circ \vect\left(\hA_K^{(r)}\right).
\end{equation}

The VAR(1) model \eqref{eq:var1} for $\vect(\cX_t)$ can be written equivalently in the tensor form
\begin{equation}
\label{eq:tenAR}
    \cX_t = \langle \mathcal{A}, \cX_{t-1} \rangle + \cE_t.
\end{equation}
where $\mathcal{A} \in \mathbb{R}^{d_1 \times \cdots \times d_K \times d_1 \times \cdots \times d_K}$ is an order-$2K$ tensor. The TenAR(1) model \eqref{TAR1} can also be represented in the form \eqref{eq:tenAR} with 
\begin{equation}
\label{eq:tenAR_A}
    \mathcal{A} = \sum_{r=1}^R \hA_1^{(r)}\circ\cdots\circ\hA_K^{(r)},
\end{equation}
where $\hA_1^{(r)}\circ\cdots\circ\hA_K^{(r)}$ is a $(d_1 \times \cdots \times d_K \times d_1 \times \cdots \times d_K)$ tensor whose $(i_1,\ldots,i_K,j_1,\ldots,j_K)$-th element is $\prod_{k=1}^K\hA_k^{(r)}[i_k,j_k]$.
Recently \cite{wang2021high} considered the tensor autoregressive model based on \eqref{eq:tenAR}, and their method hinges upon the low multi-linear rank assumption on the transition tensor $\mathcal A$. Our approach is quite different. For the TenAR(1) model \eqref{TAR1}, we do not impose any low rank conditions on the matrices $\hA_k^{(r)}$. As a result, the tensor $\mathcal A$ in \eqref{eq:tenAR_A} is not of low multi-linear ranks. Instead, the low dimensional structure of the TenAR(1) model is manifested through \eqref{eq:Phi_CP}. Furthermore, as will be illustrated in Section~\ref{interpretations}, $\hA_k^{(r)}$ capture the interactions along different mode of $\cX_t$ and admit corresponding interpretations. Therefore, our focus is on the estimation of $\hA_k^{(r)}$. \cite{wang2021high} considered the estimation of $\mathcal A$ instead.

The innovation process $\{\cE_t\}$ is assumed to be a tensor white noise, i.e. $\mathrm{Cov}(\cE_t,\cE_s)=\boldsymbol{0}$ whenever $s\neq t$. On the other hand, we allow the elements of $\cE_t$ to have concurrent dependence. Let $\Sigma_e: = \textup{Cov}(\vect(\mathcal{E}_t))$. For the least squares estimator introduced in Section~\ref{sec:als}, the only condition we require on $\Sigma_e$ is that it is nonsingular. We also consider a special form of $\Sigma_e$, 
\begin{equation}\label{kronecker cov}
    \textup{Cov}(\vect(\mathcal{E}_t)) = {\Sigma}_K \otimes {\Sigma}_{K-1} \otimes \cdots \otimes {\Sigma}_1,
\end{equation}
which will allow us to introduce the MLE in Section~\ref{sec:mle} under normality. In \eqref{kronecker cov}, each $\Sigma_i$ is a $d_i \times d_i$ symmetric positive definite matrix, $i=1,\cdots, K$. It is equivalent to assuming $\mathcal{E}_t = \mathcal{Z}_t \times_1 {\Sigma}_1^{1/2} \cdots \times_K {\Sigma}_K^{1/2}$, where all elements of $\mathcal{Z}_t$ are uncorrelated with unit variances. Intuitively, $\Sigma_i$ corresponds to mode $i$ interactions, $i=1,\cdots, K$. We will provide more background and discussion on \eqref{kronecker cov} in Section~\ref{sec:mle}.

The TenAR(1) model can be extended directly to include $p$ previous observations such as
\begin{equation}\label{multiarp}
\mathcal{X}_t = \sum_{i=1}^{p} \sum_{r=1}^{R_i} \mathcal{X}_{t-i} \times_{1} \hA_{1}^{(ir)} \times_{2}  \cdots \times_{K} \hA_{K}^{(ir)} + \mathcal{E}_t.
\end{equation}
Note that for different $\cX_{t-i}$, the number of terms $R_i$ can be different, and we use $\hA_{k}^{(ir)}$ to denote the coefficient matrix corresponding to lag $i$, term $r$ and mode $k$, $1 \le i \le p$, $1 \le r \le R_{i}$, $1 \le k \le K$. We refer to \eqref{multiarp} as the TenAR($p$) model, and $p$ the autoregressive order of the model.

From the VAR(1) representation \eqref{eq:multi_tenAR_vec}, it is immediately seen that the TenAR(1) model is causal if  $\rho\left[\sum_{r=1}^{R} \hA^{(r)}_K \otimes \hA^{(r)}_{K-1} \otimes \cdots \otimes \hA^{(r)}_1 \right]<1$. The TenAR($p$) model \eqref{multiarp} also becomes a VAR($p$) after the vectorization, through which the causality condition can be similarly given. Specifically, let $$\Phi(z)=\h{I}-\sum_{i=1}^p \left[\sum_{r=1}^{R_i} \hA^{(ir)}_K \otimes \hA^{(ir)}_{K-1} \otimes \cdots \otimes \hA^{(ir)}_1 \right]z^i.$$
Then the TenAR($p$) model is causal if $\det \Phi(z) \neq 0$ for all $|z|\leq 1$.

\subsection{Identifiability}

The multi-linear form of \eqref{TAR1} suggests that model has indeterminacy due to rescaling of coefficient matrices and permutation of terms, as seen from \eqref{eq:multi_tenAR_vec}. For this reason, we define the identifiability of the model as follows. 

\begin{definition}
    We say the model \eqref{TAR1} is {\it identified}, if 
    \begin{equation*}
    \sum_{r=1}^{R} \hA^{(r)}_K \otimes \hA^{(r)}_{K-1} \otimes \cdots \otimes \hA^{(r)}_1 = \sum_{r=1}^{\tilde R} \tilde{\hA}^{(r)}_K \otimes \tilde{\hA}^{(r)}_{K-1} \otimes \cdots \otimes \tilde\hA^{(r)}_1, \quad \tilde R\leq R 
    \end{equation*}
    implies that $\tilde R=R$ and there is permutation $\pi:[R]\rightarrow[R]$ such that
    \begin{align*}
        &\hA_k^{(r)} \propto \tilde\hA_k^{(\pi_r)}, && 1\leq k\leq K,\;1\leq r\leq R \\
        &\hA^{(r)}_K \otimes \hA^{(r)}_{K-1} \otimes \cdots \otimes \hA^{(r)}_1 = \tilde{\hA}^{(\pi_r)}_{K} \otimes \tilde{\hA}^{(\pi_r)}_{K-1} \otimes \cdots \otimes \tilde\hA^{(\pi_r)}_1, && 1\leq r\leq R.
    \end{align*}
\end{definition}

It is clear that if the model \eqref{TAR1} is identified according to this definition, and if we also require that $\|\hA_k^{(r)}\|_F=1$ for $1\leq k\leq K-1,\;1\leq r\leq R$, then each coefficient matrix $\hA_k^{(r)}$ is further identified up to a sign change.

The representations \eqref{eq:multi_tenAR_vec} and \eqref{eq:Phi_CP} help to introduce the identifiability conditions for model
\eqref{TAR1}. When $K=2$, \eqref{TAR1} becomes a multi-term MAR model,
and \eqref{eq:Phi_CP} corresponds to the singular value decomposition
of $\mathcal R(\Phi)$. To guarantee the identifiability of the
matrices $\hA_k^{(r)}$, we require that
$\tr\left[\hA_k^{(r)}\left(\hA_k^{(l)}\right)^{\prime}\right]=0$ whenever
$r\neq l$, and $\|\hA_1^{(r)}\|_F=1$ for $1\leq r\leq R$. As a result,
all the matrices $\hA_k^{(r)}$ are unique up to sign changes if the nonzero singular values of $\mathcal R(\Phi)$ are distinct. On the
other hand, for the higher order TenAR(1) model with $K\geq 3$, the classical results on the uniqueness of the tensor CP decomposition
suggest that the identifiability of $\hA_k^{(r)}$ 
is granted under the Kruskal's condition
\citep{kruskal1977three, kruskal1989rank}. 
We summarize the identifiability condition of TenAR(1) model in Propositon~\ref{kruskal}, which relies on the generalized Kruskal condition for order-$K$ tensors (\cite{sidiropoulos2000uniqueness}). 

The Kruskal condition is given through the Kruskal rank $\kappa(\hA)$ of a matrix $\hA$, which is defined as the maximum value $\kappa$ such that any $\kappa$ columns of $\hA$ are linearly independent. Let $\mathbb{A}_k:=\left[\vect\left(\hA_k^{(1)}\right),
  \vect\left(\hA_k^{(2)}\right),\cdots,\vect\left(\hA_k^{(R)}\right)
\right]$. 

\begin{proposition}\label{kruskal}
The TenAR(1) model \eqref{TAR1} is identified if any of the following holds
\begin{itemize}
    \item [(i)] $K=2$, $\tr\left[\hA_k^{(r)}\left(\hA_k^{(l)}\right)^{\prime}\right]=0$ for all $r\neq l$, $1\leq k\leq K$, and $\mathcal R(\Phi)$ has $R$ distinct nonzero singular values.
    \item [(ii)] $K\geq 3$, $\sum_{k=1}^{K} \kappa(\mathbb{A}_k) \ge 2R + K -1$.
\end{itemize}
\end{proposition}

In particular, when $K\geq 3$, if we assume that for each $k$, the matrix $\mathbb{A}_k$ is of rank $R$, then the Kruskal's condition is fulfilled. 

We have given the identifiability conditions for the TenAR(1) model in Proposition~\ref{kruskal}. For the TenAR($p$) model \eqref{multiarp}, the identifiability conditions should be imposed for each lag $i$, $1\leq i\leq p$. These identifiability conditions will be assumed for the rest of this paper.

\subsection{Model interpretations}\label{interpretations}

The autoregressive term in the TenAR model involves mode products of $\cX_{t-i}$ with the coefficient matrices. It is helpful to picture what type of temporal dependence the mode product is introducing. We use the Fama-French portfolio as an example. The stocks are allocated to two Size (according to market equity) groups (Small and Big), four B/M (Book-to-Market ratio) groups (low B/M, mid1 B/M, mid2 B/M, high B/M), and four OP (Operating Profitability) groups (low OP, mid1 OP, mid2 OP, high OP). This cross allocation leads to $2\times 4\times 4=32$ groups. A portfolio is constructed for each group, and at each month, the returns of these 32 portfolios are recorded in a $2\times 4\times 4$ tensor $\cX_t$. Consider the TenAR(1) model
and assume $\hA_1=\hA_2=\boldsymbol{I}$, leading to the simplified model
\begin{equation*}
\mathcal{X}_t = \mathcal{X}_{t-1} \times_{1} \hA_3 + \mathcal{E}_t.
\end{equation*}
We use the 4-th frontal slices $\cX_{::4}$ as an example, its conditional expectation is given by a linear combination of 4 frontal slices of $\cX_{t-1}$, as illustrated in Figure~\ref{Horizontal}.
\begin{figure}[!ht]
    \centering
  
\tikzset {_pl31428i3/.code = {\pgfsetadditionalshadetransform{ \pgftransformshift{\pgfpoint{0 bp } { 0 bp }  }  \pgftransformrotate{0 }  \pgftransformscale{2 }  }}}
\pgfdeclarehorizontalshading{_em3jn3dzr}{150bp}{rgb(0bp)=(1,1,1);
rgb(37.5bp)=(1,1,1);
rgb(37.5bp)=(0.9,0.9,0.9);
rgb(100bp)=(0.9,0.9,0.9)}

  
\tikzset {_2za1f41r5/.code = {\pgfsetadditionalshadetransform{ \pgftransformshift{\pgfpoint{0 bp } { 0 bp }  }  \pgftransformrotate{0 }  \pgftransformscale{2 }  }}}
\pgfdeclarehorizontalshading{_djc0tfuow}{150bp}{rgb(0bp)=(1,1,1);
rgb(37.5bp)=(1,1,1);
rgb(37.5bp)=(0.9,0.9,0.9);
rgb(100bp)=(0.9,0.9,0.9)}

  
\tikzset {_6sipefcep/.code = {\pgfsetadditionalshadetransform{ \pgftransformshift{\pgfpoint{0 bp } { 0 bp }  }  \pgftransformrotate{0 }  \pgftransformscale{2 }  }}}
\pgfdeclarehorizontalshading{_g5pk77cnr}{150bp}{rgb(0bp)=(1,1,1);
rgb(37.5bp)=(1,1,1);
rgb(37.5bp)=(0.9,0.9,0.9);
rgb(100bp)=(0.9,0.9,0.9)}

  
\tikzset {_9q8acnjiu/.code = {\pgfsetadditionalshadetransform{ \pgftransformshift{\pgfpoint{0 bp } { 0 bp }  }  \pgftransformrotate{0 }  \pgftransformscale{2 }  }}}
\pgfdeclarehorizontalshading{_b2v10zbtf}{150bp}{rgb(0bp)=(1,1,1);
rgb(37.5bp)=(1,1,1);
rgb(37.5bp)=(0.9,0.9,0.9);
rgb(100bp)=(0.9,0.9,0.9)}

\tikzset{every picture/.style={line width=0.75pt}} 

\begin{tikzpicture}[x=0.75pt,y=0.75pt,yscale=-0.8,xscale=0.8]

\begin{scope}
\path  [shading=_g5pk77cnr,_6sipefcep] (470,70) -- (620,70) -- (620,220) -- (470,220) -- cycle ; 
\draw   (470,70) -- (620,70) -- (620,220) -- (470,220) -- cycle ;
\end{scope}

\begin{scope}
\path  [shading=_djc0tfuow,_2za1f41r5] (410,100) -- (560,100) -- (560,250) -- (410,250) -- cycle ; 
\draw   (410,100) -- (560,100) -- (560,250) -- (410,250) -- cycle ;
\end{scope}

\begin{scope}
\path  [shading=_em3jn3dzr,_pl31428i3] (350,130) -- (500,130) -- (500,280) -- (350,280) -- cycle ; 
\draw   (350,130) -- (500,130) -- (500,280) -- (350,280) -- cycle ;
\end{scope}

\begin{scope}
\path   [shading=_b2v10zbtf,_9q8acnjiu] (290,160) -- (440,160) -- (440,310) -- (290,310) -- cycle ; 
\draw   (290,160) -- (440,160) -- (440,310) -- (290,310) -- cycle ;
\end{scope}

\begin{scope}
\path  [shading=_g5pk77cnr,_6sipefcep] (50,120) -- (200,120) -- (200,270) -- (50,270) -- cycle ;
\draw   (50,120) -- (200,120) -- (200,270) -- (50,270) -- cycle ;
\end{scope}

\draw (445,290) node [anchor=north west][inner sep=0.75pt]   [align=left] {$\times \hA_3[1,1] +$};
\draw (500,260) node [anchor=north west][inner sep=0.75pt]   [align=left] {$\times \hA_3[1,2] +$};
\draw (560,230) node [anchor=north west][inner sep=0.75pt]   [align=left] {$\times \hA_3[1,3] +$};
\draw (620,200) node [anchor=north west][inner sep=0.75pt]   [align=left] {$\times \hA_3[1,4]$};
\draw (240,140) node [anchor=north west][inner sep=0.75pt]   [align=left] {low OP};
\draw (280,115) node [anchor=north west][inner sep=0.75pt]   [align=left] {mid1 OP};
\draw (330,90) node [anchor=north west][inner sep=0.75pt]   [align=left] {mid2 OP};
\draw (380,70) node [anchor=north west][inner sep=0.75pt]   [align=left] {high OP};
\draw (50,100) node [anchor=north west][inner sep=0.75pt]   [align=left] {low OP};
\draw (120,280) node [anchor=north west][inner sep=0.75pt]   [align=left] {$t$};
\draw (350,320) node [anchor=north west][inner sep=0.75pt]   [align=left] {$t-1$};
\draw (225,200) node [anchor=north west][inner sep=0.75pt]   [align=left] {$=$};

\end{tikzpicture}
    \caption{Linear combinations of frontal slices}\label{Horizontal}
    \label{fig4}
\end{figure}
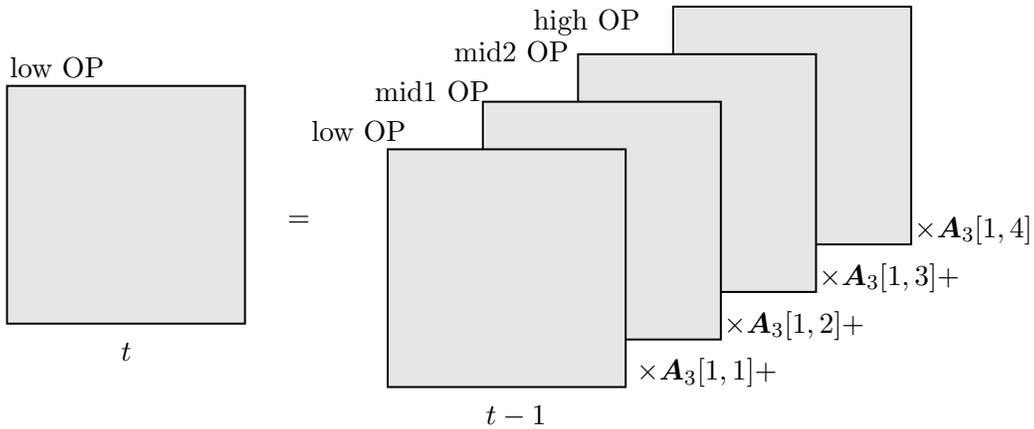

If we take $\hA_1=\hA_3=\boldsymbol{I}$ or $\hA_2=\hA_3=\boldsymbol{I}$, there are similar interpretations. In general, the mode-$k$ product gives linear combinations of mode-$k$ slices of $\cX_{t-1}$. In the full TenAR(1) model, the interactions along all modes are mixed up, and the multiple terms allow interactions along different directions, giving a more comprehensive modeling capacity.

\section{Estimation}\label{estimation}

\subsection{Alternating least squares}
\label{sec:als}

For the estimation, we first consider the least squares method. To fix ideas, we begin with the one-term ($R=1$) TenAR(1) model. The extension to multi-term TenAR(1) and to TenAR($p$) is relatively straightforward, and will be given at the end of this section. The least squares estimator (LSE), denoted by $\hat{\hA}_1, \cdots, \hat{\hA}_K$,  is the solution of the minimization problem
\begin{equation}\label{lse}
\min_{\hA_1,\cdots,\hA_K}\sum_{t=2}^T\|\mathcal{X}_t - \mathcal{X}_{t-1} \times_{1} \hA_1 \cdots \times_{K} \hA_K\|^2_F.
\end{equation}

The optimization problem (\ref{lse}) is not convex, due to its multi-linear form. We propose to use the alternating least squares to solve it: updating one $\hA_k$, while holding others fixed. To give details of the algorithm, we need to perform the tensor matricization operation, which, also known as {\it unfolding} or {\it flattening}, transforms a tensor into a matrix along a mode. Suppose $\cX\in\mathbb{R}^{d_1\times\cdots\times d_K}$. Denote $d=\prod_{k=1}^K d_k$ and $d_{-k}=d/d_k$. The {\it mode-$k$ matricization}, denoted by $\boldsymbol{X}_{(k)} \in \mathbb{R}^{d_k \times d_{-k}}$, is obtained by assembling all mode-$k$ fibers of $\mathcal{X}$ as columns of the matrix $\boldsymbol{X}_{(k)}$. Specifically, the tensor element $(i_1, i_2, \cdots, i_K)$ of $\mathcal{X}$ maps to the matrix element $(i_k,j)$ of $\boldsymbol{X}_{(k)}$ where
$$j = 1 + \sum_{\substack{s = 1 \\ s\neq k}}^{K} (i_s - 1)J_s \ \ \text{with} \ \ J_s = \prod_{\substack{l=1 \\ l\neq k}}^{s-1} d_l,\;\hbox{ and } J_1=1.$$
After the mode-$k$ matricization, the one term TenAR(1) model becomes
$$
\boldsymbol{X}_{t(k)} = \hA_k \underbrace{\boldsymbol{X}_{t-1,(k)} (\hA_K \otimes \cdots \otimes \hA_{k+1} \otimes \hA_{k-1} \cdots \otimes \hA_1 )^{\prime}}_{\h{W}_{t-1,(k)}}+\h{E}_{t(k)}.$$
If the matrices $\hA_1,\ldots,\hA_{k-1},\hA_{k+1},\ldots,\hA_K$ are given, the optimization over $\hA_k$ corresponds to a multivariate linear regression, in view of the preceding equation. Therefore, $\hA_k$ can be updated as
\begin{equation*}
    \hA_k \gets \left(\sum_{t} \boldsymbol{X}_{t(k)} \boldsymbol{W}^{\prime}_{t-1,(k)} \right) \left(\sum_{t} \boldsymbol{W}_{t-1,(k)} \boldsymbol{W}^{\prime}_{t-1,(k)} \right)^{-1}.
\end{equation*}

The alternating least squares algorithm update $\hA_k$ one by one iteratively until convergence. Since for each iteration, the sum of squared errors is reduced, so the algorithm is guaranteed to converge. However, the iterative algorithm often converges to a local minima. In practice, we suggest to use the projection estimator to be discussed in \ref{sec:proj} as the initial values of the alternating least squares. Our simulation experiment in Section~\ref{local} conforms that this initialization often leads to satisfactory performance.

Next we introduce the LSE estimator of the multi-term TenAR($p$) model \eqref{multiarp}, given by
\begin{equation}\label{lsearp}
    \left(\hat{\hA}^{(11)}_1, \cdots,  \hat{\hA}^{(pR_p)}_K\right) = \argmin_{\hA^{(11)}_1,\cdots,\hA^{(pR_p)}_K} \sum_t
    \left\| \mathcal{X}_t -\sum_{i=1}^{p} \sum_{r=1}^{R_i} \mathcal{X}_{t-i} \times_{1} \hA_{1}^{(ir)} \times_{2} \cdots \times_{K} \hA_{K}^{(ir)} \right\|_F^{2}.
\end{equation}
Let $\Phi_{k}^{(ir)} := \hA^{(ir)}_{K} \otimes \cdots \otimes \hA^{(ir)}_{k+1} \otimes \hA^{(ir)}_{k-1} \otimes \cdots \otimes \hA^{(ir)}_{1}$, and $\h{W}_{t(k)}^{(ir)}:=\hX_{t(k)}\left( \Phi_{k}^{(ir)}\right)^\prime$. Given all other coefficient matrices, $\hA_k^{(ir)}$ is updated by
\begin{equation*}
    \hA^{(ir)}_{k} \gets \sum_{t}\left[ \left( \boldsymbol{X}_{t(k)}  - \sum_{(j,l) \neq (i,r)} \hA_{k}^{(jl)} \h{W}_{t-1,(k)}^{(jl)} \right) \left(  \h{W}_{t-1,(k)}^{(ir)}\right)^\prime\right] \left[\sum_{t}  \h{W}_{t-1,(k)}^{(ir)}\left(\h{W}_{t-1,(k)}^{(ir)}\right)^\prime \right]^{-1}
\end{equation*}
Then \eqref{lsearp} is solved by updating $\hA_{k}^{(ir)}$ iteratively until convergence.

\subsection{MLE when \texorpdfstring{$\mathrm{Cov}(\vect(\cE_t))$}{cov} is separable}
\label{sec:mle}

For VAR models, the least squares estimator is also the conditional MLE under normality. Although the TenAR model can be represented in the VAR form (e.g. \eqref{eq:multi_tenAR_vec}), the coefficient matrix $\Phi$ is not a free parameter, but a sum of a few Kronecker products. As a result, the LSE is no longer the MLE for TenAR models, and the MLE is very difficult to compute. However, if we assume in addition that $\mathrm{Cov}(\vect(\cE_t))$ has the form \eqref{kronecker cov}, then the MLE can be obtained through an alternating algorithm. The covariance structure \eqref{kronecker cov} has been proposed and studied from various aspects in the literature \citep{allen2010transposable,hoff2011separable,tsiligkaridis2013covariance,zhou2014gemini,hafner2020estimation,linton2019estimation}. Following the terminology in spatial statistics \citep{cressie2015statistics}, we say the covariance matrix of the form \eqref{kronecker cov} is {\it separable}.  \cite{chen2020autoregressive} also considered the MAR model under this covariance structure.


We denote the MLE of the TenAR($p$) model under normality by $\tilde \hA_{k}^{(ir)}$ and $\tilde\Sigma_k$. To simplify many long equations involved in the discussion, we introduce some notations 
\begin{align*}
\mathcal{R}_t &= \mathcal{X}_{t} - \sum_{i=1}^{p} \sum_{r=1}^{R_i} \mathcal{X}_{t-1} \times_{1} \hA_{1}^{(ir)} \times_{2} \cdots \times_{K} \hA_{K}^{(ir)}, \\
     \Phi_{k}^{(ir)} &= \hA_K^{(ir)} \otimes \cdots \otimes \hA_{k+1}^{(ir)} \otimes \hA_{k-1}^{(ir)} \otimes \cdots \otimes \hA_{1}^{(ir)}, \\
     \boldsymbol{S}_k &= {\Sigma}_K \otimes \cdots  \otimes {\Sigma}_{k+1} \otimes {\Sigma}_{k-1} \otimes  \cdots \otimes {\Sigma}_1. 
\end{align*}
The log likelihood under normality can be written as, for any $1 \le k \le K$,
\begin{equation}
    -\frac{1}{2} (T-1)\log 2 \pi - \frac{1}{2}\sum_{k=1}^K (T-1) (\prod_{i \neq k}d_i) \log|\Sigma_k| - \frac{1}{2} \sum_{t} \text{tr} [{\Sigma}_k^{-1} \mathcal{R}_{t(k)} \boldsymbol{S}_k^{-1} \mathcal{R}_{t(k)}^{\prime}].
\end{equation}
The gradient conditions for ${\Sigma}_k$ and $\hA^{(ir)}_{k}$ are given by
\begin{align*}
    \sum_{t} \boldsymbol{R}_{t(k)} \left( \boldsymbol{S}_{k}^{-1} {\Phi^{(ir)}_k} {\boldsymbol{X}_{t-i,(k)}}^{\prime}  \right) &= 0, \\
    (T-1)\prod_{l \neq k} d_l{\Sigma}_k - \sum_t \boldsymbol{R}_{t(k)} \boldsymbol{S}_k^{-1} \boldsymbol{R}_{t(k)}^{\prime} &= 0.
\end{align*}
Therefore, when all other parameters are given, $\hA^{(ir)}_{k}$ and ${\Sigma}_k$ are updated by,
\begin{align*}
        \hA^{(ir)}_{k} \gets &
         \sum_{t}  \left[ \left( \boldsymbol{X}_{t(k)}  - \sum_{l=1}^{p} \sum_{(l,j) \neq (i,r)} \hA_{k}^{(lj)} \boldsymbol{X}_{t-l,(k)} {\Phi^{(lj)}_k}^{\prime}  \right) \boldsymbol{S}_{k}^{-1}{\Phi^{(ir)}_k} {\boldsymbol{X}_{t-i,(k)}}^{\prime} \right] \\ &\left[\sum_{t} \boldsymbol{X}_{t-i,(k)} {\Phi^{(ir)}_k}^{\prime}\boldsymbol{S}_{k}^{-1} {\Phi^{(ir)}_k} {\boldsymbol{X}_{t-i,(k)}}^{\prime}  \right]^{-1}, \\ 
    {\Sigma}_k \gets & \frac{\sum_t  \mathcal{R}_{t(k)} \boldsymbol{S}_k^{-1} \mathcal{R}_{t(k)}^{\prime}}{(T-1)\prod_{l \neq k} d_l}.
\end{align*}
The MLE is then obtained by updating $\hA^{(ir)}_{k}$ and ${\Sigma}_k$ iteratively until convergence. Similar to the alternating least squares algorithm, the algorithm for MLE also requires a good initialization, which we discuss next.

\subsection{Initialization of the algorithm}
\label{sec:proj}

Alternating algorithms for LSE and MLE are convenient and easy to implement, but they are not guaranteed to converge to a global minimum, only to a solution where the objective function ceases to decrease. As a result, either the algorithm can take many iterations to converge, or the final solution can be heavily dependent on the initial values. For example, our simulations show that some randomly chosen initial values can lead to very poor estimates. Therefore, it is crucial to choose initial values properly for the alternating algorithms.

Our approach is to use the projection estimators as initial values. We shall discuss the initialization of the TenAR(1) model, while the extension to TenAR($p$) is relatively straightforward. We first fit the VAR(1) model \eqref{eq:var1} to $\vect(\cX_t)$ and obtain the LSE $\check\Phi$ of the coefficnet matrix $\Phi$. Since the TenAR(1) model in (\ref{TAR1}) can be viewed as a structured VAR(1) model in (\ref{eq:multi_tenAR_vec}), the projection estimators are obtained by projecting $\hat{\Phi}$ onto the space of Kronecker products under the Frobenius norm: 
\begin{equation}\label{multipro}
    (\bar{\hA}^{(1)}_1, \cdots,\bar{\hA}^{(R)}_K) = \argmin_{\hA^{(1)}_1,\cdots,\hA^{(R)}_K} \left\| \hat{\Phi} - \sum_{r=1}^{R} \hA^{(r)}_K \otimes  \cdots \otimes \hA^{(r)}_1 \right\|_F^2.
\end{equation}
When $R=1$, $K=2$ this minimization problem is called \textit{the nearest Kronecker product} (NKP) problem in matrix computation \citep{van2000ubiquitous, van1993approximation}, which can be solved by rearrangement and SVD decomposition. More generally, when $K\geq 2$, 
there exist a rearrangement operator $\mathcal{R} : \mathbb{R}^{d_1 d_2\cdots d_K \times d_1 d_2 \cdots d_K} \to \mathbb{R}^{d_1^2 \times d_2^2 \cdots \times d_K^2}$ such that
\begin{equation*}
    \mathcal{R}\left(\sum_{r=1}^{R} \h{A}^{(r)}_K \otimes  \cdots \otimes \h{A}^{(r)}_1\right)=\sum_{r=1}^{R} \h{a}^{(r)}_1 \circ  \cdots \circ \h{a}^{(r)}_K,
\end{equation*}
where $\h{a}^{(r)}_k=\vect(\h{A}^{(r)}_k)$, $1 \le r \le R$, $1 \le k \le K$. The explicit formula of the rearrangement operator is given in Appendix. After the rearrangement, the optimization \eqref{multipro} becomes a problem of finding the best rank-$R$ approximation of the $K$-way tensor $\mathcal{R}(\hat{\Phi})$:
\begin{equation}\label{multiprovec}
\left\|\check{\Phi} - \sum_{r=1}^{R} \h{A}^{(r)}_K \otimes  \cdots \otimes \h{A}^{(r)}_1 \right\|_F^2=\left\|\mathcal{R}(\check{\Phi}) - \sum_{r=1}^{R} \h{a}^{(r)}_1 \circ  \cdots \circ \h{a}^{(r)}_K\right\|_F^2.
\end{equation}
It is well known that the best low rank approximation may not exist for tensors of orders higher than or equal to 3 \citep{de2008tensor, krijnen2008non, Stegeman}. 
Furthermore, the alternating least squares algorithm is not guaranteed to converge to the global minimum \citep{kolda2009tensor}, even when it does exist. 
On the other hand, \cite{anandkumar2014guaranteed} and \cite{sun2017provable} provided local and global convergence guarantees for recovering CP tensor decomposition when the tensor components are incoherent, which can be viewed as a soft-orthogonality constraint. We suggest to use these methods to find the best low rank approximation. 
In Theorem~\ref{proj_arp_multi} in Appendix, we establish the central limit theorem for the estimators $\bar\hA_k$ based on the one-term TenAR(1) model. For the general multi-term TenAR($p$) model, the initialization is done similarly, and the corresponding central limit theorems for $\bar\hA_k^{(ir)}$ can be similarly developed.


For MLE, we use a hierarchical SVD procedure to initialize $\Sigma_1, \cdots, \Sigma_K$. Without loss of generality, assume $K=3$. First, we obtain $\check{\Sigma}$, which is the estimated covariance matrix of $\vect(\cE_t)$ based on the VAR(1) model (\ref{eq:var1}). Second, we get $\mathcal{R}_1(\check{\Sigma})$ by rearrangement operator $\mathcal{R}_1: \mathbb{R}^{d_1 d_2 d_3 \times d_1 d_2 d_3} \to \mathbb{R}^{d_1^2 \times d_2^2 d_3^2}$ such that $\mathcal{R}_1(\Sigma_3 \otimes \Sigma_2 \otimes \Sigma_1) = \vect(\Sigma_1) \vect(\Sigma_3 \otimes \Sigma_2)^{\prime}$. Denote $\h{S}_1 = \mathcal{R}_1(\check{\Sigma})$ which has the SVD decomposition $\h{S}_1 = \sum_{i=1}^{m_1} s_i \hu_i \hv_i^{\prime}$ where $\hu_i \in \mathbb{R}^{d_1^2}$ and $\hv_i \in \mathbb{R}^{d_2^2 d_3^2}$, $s_1 \ge \cdots \ge s_m$. Let $\bar{\Sigma}_1 := \vect^{-1}(\hu_1) \in \mathbb{R}^{d_1 \times d_1}$. Next, we have $\vect^{-1}(s_1 \hv_1) \in \mathbb{R}^{d_2 d_3 \times d_2 d_3}$ and denote $\h{S}_2 = \mathcal{R}_2(\vect^{-1}(s_1 \hv_1)) \in \mathbb{R}^{d_2^2 \times d_3^2}$ where rearrangement operator $\mathcal{R}_2: \mathbb{R}^{d_2 d_3 \times d_2 d_3} \to \mathbb{R}^{d_2^2 \times d_3^2}$ such that $\mathcal{R}_2(\Sigma_3 \otimes \Sigma_2 ) = \vect(\Sigma_2) \vect(\Sigma_3)^{\prime}$. Similarly, $\h{S}_2$ has the SVD decomposition $\h{S}_2 = \sum_{i=1}^{m_2} \sigma_i \ha_i \h{b}_i^{\prime}$ where $\ha_i \in \mathbb{R}^{d_2^2}$ and $\h{b}_i \in \mathbb{R}^{d_3^2}$. Let $\bar{\Sigma}_2 := \vect^{-1}(\ha_1) \in \mathbb{R}^{d_2 \times d_2}$ and $\bar{\Sigma}_3 := \vect^{-1}(\sigma_1 \h{b}_1) \in \mathbb{R}^{d_3 \times d_3}$. It turns out the hierarchical SVD procedure guarantees $\bar{\Sigma}_k$, $k=1,\cdots,K$ to be symmetric positive semi-definite. The following proposition also asserts their consistency under suitable conditions.


\begin{proposition}\label{symmetric_prosig}
Assume $\Sigma = \mathrm{Cov}(\vect(\cE_t))$ has the form \eqref{kronecker cov}.
\begin{enumerate}
    \item [(i)] Each $\bar{\Sigma}_k$ is symmetric and positive semi-definite, $1\leq k\leq K$.
    \item [(ii)] If $\|\tilde{\Sigma} - \Sigma\|_F = o_p(1)$, then $\|\bar{\Sigma}_k - \Sigma_k\|_F = o_p(1)$, $1\leq k\leq K$.
\end{enumerate}
\end{proposition}

The proof is based on \cite{van1993approximation} and the matrix perturbation theory \citep{davis1970rotation, wedin1972perturbation}, and is given in Appendix. 

\section{Asymptotics}\label{asymptotics}


In this section, we establish the central limit theorem for the LSE and MLE under the fixed-dimensional setup, assuming $d_1,\ldots,d_K$ are fixed. We also investigate the convergence rates of the LSE under the high dimensional paradigm, allowing $d=d_1\cdots d_N$ to grow with the sample size $T$. Recall that the conditions of Proposition~\ref{kruskal} are assumed to hold for identifiability. Furthermore, we make the convention that $\|\hA_k^{(ir)}\|_F = 1$ and the estimators $\hat{\hA}_{k}^{(ir)}$ are also rescaled so that $\|\hat{\hA}_{k}^{(ir)}\|_F = 1$ for $1 \le k \le K-1$, $1 \le i \le p$, $1 \le r \le R_i$.

\subsection{Asymptotics for LSE estimators in multi-term TenAR(\texorpdfstring{$p$}{p}) Model}

Recall that $\hat\hA_{k}^{(ir)}$ denote the LSE. We first introduce some notations for Theorem~\ref{lse_arp_multi}. Let $\ha_k^{(ir)} := \vect(\hA_k^{(ir)})$, $d = d_1^2 + \cdots + d_K^2$ and $\h{\gamma}_{k}^{(ir)} := (\h{0}^{\prime},\ha_k^{(ir)\prime}, \h{0}^{\prime})^{\prime}$ be a vector in $\mathbb{R}^{\sum_{i=1}^{p} R_i d}$, where the first $\h{0} \in \mathbb{R}^{(r-1)d + (d_1^{2} + \cdots + d_{k-1}^{2}) + \sum_{j=1}^{i-1} R_j d}$ and the second $\h{0} \in \mathbb{R}^{(R_i-r)d + (d_{k+1}^{2} + \cdots + d_{K}^{2}) + \sum_{j={i+1}}^{p} R_j d}$, for $1 \le i \le p$, $1 \le r \le R_i$ and $1 \le k \le K$. Define $\boldsymbol{H} := \mathbb{E}(\boldsymbol{W}_t \boldsymbol{W}_t^{\prime}) + \sum_{i=1}^{p} \sum_{r=1}^{R_i} \sum_{k=1}^{K-1} \h{\gamma}_{k}^{(ir)} \h{\gamma}_{k}^{(ir) \prime}$,  $\Xi_2 =: \boldsymbol{H}^{-1} \mathbb{E}(\boldsymbol{W}_t {\Sigma} \boldsymbol{W}_t^{\prime}) \boldsymbol{H}^{-1}$,
\begin{align*}
\boldsymbol{W}_t = \begin{pmatrix} 
\boldsymbol{W}^{(11)}_t \\
\cdots\\
\boldsymbol{W}^{(pR_p)}_t\\
\end{pmatrix}, \ 
\boldsymbol{W}^{(ir)}_{t} = \begin{pmatrix} 
((\boldsymbol{X}_{t+1-i,(1)}{\Phi^{(iR_i)}_1}^{\prime}) \otimes \boldsymbol{I}_{d_1}) \boldsymbol{Q}_1 \\
\cdots\\
((\boldsymbol{X}_{t+1-i,(K)}{\Phi^{(iR_i)}_K}^{\prime}) \otimes \boldsymbol{I}_{d_K})  \boldsymbol{Q}_K\\
\end{pmatrix},
\end{align*}
where $\Phi^{(ir)}_k= \hA^{(ir)}_K \otimes  \cdots \hA^{(ir)}_{k+1} \otimes \hA^{(ir)}_{k-1} \cdots \otimes \hA^{(ir)}_1$, $1 \le r \le R_i$, $1 \le i \le p$, $1 \le k \le K$, and $\boldsymbol{Q}_k$ are permutation matrices (\ref{Q}) defined in Appendix~\ref{appendix:basic}. 


\begin{theorem}\label{lse_arp_multi}
Assume that the TenAR($p$) model \eqref{eq:tenAR} is causal, and the error tensors $\mathcal{E}_1, \cdots, \mathcal{E}_{T}$ are IID with mean zero and finite second moments. Also assume that the coefficient matrices $\hA_k^{(ir)}$, and $\Sigma$ are nonsingular. It holds that
\begin{equation*}
 \sqrt{T}\begin{pmatrix} 
vec(\hat{\hA}^{(11)}_{1} - \hA^{(11)}_1) \\
\cdots \\
vec(\hat{\hA}^{(pR_p)}_{K} - \hA^{(pR_p)}_K)
\end{pmatrix} \to \mathcal{N}(0,\Xi_2).
\end{equation*}
\end{theorem}

The proof is given in Appendix~\ref{appendix:proofs}. The central limit theorem for the TenAR(1) model is a special case of Theorem~\ref{lse_arp_multi}, for which the asymptotic covariance matrix is much simplified. Since TenAR(1) may arguably be the most popular TenAR model in practice, we provide the explicit formulas for the one- and multi-term TenAR(1) models in Appendix \ref{appendix:theorems}. Theorem~\ref{lse_arp_multi} includes the MAR(1) model considered in \cite{chen2020autoregressive} as a special case with $p=1$, $R=1$ and $K=2$.

\subsection{Asymptotics for MLE in multi-term TenAR(\texorpdfstring{$p$}{p}) Model}

With the additional assumption (\ref{kronecker cov}) on the covariance structure of $\mathcal{E}_t$, we present the central limit theorem for the MLE estimators $\tilde{\hA}_k^{(ir)}$, $1 \le i \le p$, $1 \le r \le R_i$, $1 \le k \le K$. Define $\boldsymbol{H} := \mathbb{E}(\boldsymbol{W}_t {\Sigma}^{-1} \boldsymbol{W}_t^{\prime}) + \sum_{i=1}^{p} \sum_{r=1}^{R_i} \sum_{k=1}^{K-1} \h{\gamma}_{k}^{(ir)} \h{\gamma}_{k}^{(ir) \prime}$, where $\boldsymbol{W}_t$ and $\h{\gamma}_{k}^{(ir)}$ are defined before Theorem~\ref{lse_arp_multi}. Let $\Xi_3 =: \boldsymbol{H}^{-1} \mathbb{E}(\boldsymbol{W}_t {\Sigma}^{-1} \boldsymbol{W}_t^{\prime}) \boldsymbol{H}^{-1}$.

\begin{theorem}\label{mle_arp_multi}
Assume the same conditions as Theorem~\ref{lse_arp_multi}. In addition, assume the error tensors $\cE_t$ are IID normal, with covariance tensor of the form \eqref{kronecker cov}. It holds that
\begin{equation*}
 \sqrt{T}\begin{pmatrix} 
vec(\tilde{\hA}_{1}^{(11)} - \hA_{1}^{(11)}) \\
\cdots \\
vec(\tilde{\hA}^{(pR_p)}_{K} - \hA^{(pR_p)}_K)
\end{pmatrix} \to \mathcal{N}(0,\Xi_3).
\end{equation*}
\end{theorem}
The proof is given in Appendix. Similar as Corollary \ref{lse_ar1_one} and \ref{lse_ar1_multi}, it includes the asymptotics for the multi-term and one-term TenAR(1) MLE estimators as special cases and we omit the details. The explicit formulas for the one- and multi-term TenAR(1) models are put in Appendix~\ref{appendix:theorems}. Theorem~\ref{mle_arp_multi} includes the MAR(1) model considered in \cite{chen2020autoregressive} as a special case with $p=1$, $R=1$ and $K=2$.


\subsection{Convergence rates under high dimensionality}\label{sec:phi_convergence}

In this section we consider the convergence rates of the estimators under high dimensional paradigm, allowing $d=d_1\cdots d_K$ to grow with the sample size $T$. To avoid the complication involved in the covariance matrix estimation, we focus on the LSE. For each $1\leq i\leq p$, let ${\Phi}^{(i)} = \sum_{r=1}^{R_i} {\hA}_K^{(ir)} \otimes \cdots \otimes {\hA}_1^{(ir)}$, and $\hat{\Phi}^{(i)} = \sum_{r=1}^{R_i} \hat{\hA}_K^{(ir)} \otimes \cdots \otimes \hat{\hA}_1^{(ir)}$ be its corresponding estimator, constructed using the LSE $\hat\hA_{k}^{(ir)}$ introduced in Section~\ref{sec:als}. We use $\|\cdot\|$ to denote the spectral norm of a matrix.

\begin{theorem}\label{phi_convergence}
Assume the conditions of Theorem~\ref{lse_arp_multi}, the error tensors $\mathcal{E}_t$ are IID sub-Gaussian, and assume $d\log d/T\rightarrow 0$. It holds that
$$\|\hat{\Phi}^{(i)} - \Phi^{(i)} \| = O_p(\sqrt{{d}/{T}}), \quad 1\leq i\leq p.$$
\end{theorem}
The proof is given in Appendix~\ref{appendix:proofs}. 

\section{Determining Autoregressive Orders and Terms}
\label{sec:selection}

The general TenAR($p$) model (\ref{multiarp}) involves $p$ previous tensor observations (referred to as the order of the model), and for each lag $i$, it can have $R_i$ multi-linear terms (referred to as the K-rank). We collect the K-ranks in the vector $\h{R}_p=(R_1,\ldots,R_p)'$.
While the general TenAR($p$) model provides more flexibility and capability to capture different interactions among fibers of the tensor, it also poses the challenge of finding the order and the suitable number of terms for each lag. We propose an information criterion based procedure, which achieves selection consistency under both fixed and high dimensional setup.

For any given order $\tilde p$ and K-ranks $\tilde{\h{R}}_{\tilde p}=(\tilde R_1,\ldots,\tilde R_{\tilde p})'$, define the information criterion as
\begin{equation}
\label{ic}
\ic (\tilde{\h{R}}_{\tilde p}) := \frac{1}{2}\log\left(\frac{1}{dT}\sum_{t}\left\|\mathcal{X}_t - \sum_{i=1}^{\tilde p} \sum_{r=1}^{\tilde R_i} \mathcal{X}_{t-i} \times_{1} \hat\hA_{1}^{(ir)} \times_{2}  \cdots \times_{K} \hat\hA_{K}^{(ir)}\right\|_F^2\right) + g(d,T) \sum_{i=1}^{\tilde p} \tilde R_i,
\end{equation}
where $\hat{\hA}_{k}^{(ir)}$ are the estimates obtained under given order and K-ranks $\tilde{\h{R}}_{\tilde p}$. The function $g(d,T)$ controls the penalty on the complexity of the model. We assume it satisfies the following condition.
\begin{assumption}\label{assump:penalty}
$g(d,T) \to 0$ and $\frac{T}{d}  g(d,T) \to \infty$ as $T \to \infty$.
\end{assumption}
We propose two specific choices of $g(d,T)$ which satisfy the preceding assumption.
\begin{align}
    \ic_1: & \;\hbox{ with }\; g_1(d,T)=\log T/T, \label{ic1}\\
    \ic_2: & \;\hbox{ with }\; g_2(d,T)=\frac{(d_1^2 + \cdots + d_K^2 - K + 1) \log T}{dT}. \label{ic2}
\end{align}



\begin{remark}
The choice of $g_1(d,T)$ uses the total number of terms $\sum_{i=1}^{\tilde p} \tilde R_i$ as the complexity of the model, with the weight $\log T/T$. For $g_2(d,T)$, we use the total number of parameters as the complexity of the model for tensor mode $K\ge3$. If they are matrices, strictly speaking we should adjust the total number of parameters since we have additional orthogonality conditions in Proposition \ref{kruskal}. Nevertheless, for simplicity we continue to use the form since $p$ is fixed and $R_i$ is not so large in our case.
\end{remark}

 These criteria can be viewed as an extended Bayesian information criterion. Such criteria were first introduced by \cite{chen2008extended} and \cite{foygel2010extended} under different context. Similar forms were also used for selecting the number of factors \citep{bai2002determining}, the configuration of the Kronecker product \citep{cai2019hybrid,cai2019kopa}. 

In practice we typically cap the maximum order and K-rank at some given $P_{\max}$ and $R_{\max}$, so the estimated $\hat p$ and $\hat{\h{R}}$ is given by
\begin{equation}
\label{eq:joint}
    \hat{\h{R}}_{\hat p}:=(\hat R_1,\ldots,\hat R_{\hat p})' = \arg\min_{\tilde p\leq P_{\max},\tilde R_1\leq R_{\max},\ldots,\tilde R_{\tilde p}\leq R_{\max}} \ic(\tilde{\h{R}}).
\end{equation}
The joint selection of $R_i$ can be costly when $d$ and $R_{\max}$ are large. We also consider the separate selection procedure. Specifically, to select $R_i$, we fix $p$ at $P_{\max}$, and all $R_{i'}$ except $R_i$ at $R_{\max}$,
\begin{equation}
\label{eq:sep}
    \hat R_i = \arg\min_{\tilde R_i\leq R_{\max}} \ic(R_{\max},\ldots,\tilde R_i,\ldots,R_{\max}), \quad1\leq i\leq P_{\max}.
\end{equation}
The largest $i$, for which the selected $\tilde R_i>0$, is the estimated autoregressive order $\hat p$.

For the selection consistency, we need an additional assumption on the ``size'' of $\hA_k^{(ir)}$.  
\begin{assumption}\label{fullrank}
Assume there exists some constant $\eta > 0$ such that $\| \hA_{K}^{(ir)} \otimes \hA_{K-1}^{(ir)} \otimes \cdots \otimes \hA_{1}^{(ir)} \|_{F}^2 \geq \eta d$ for all $1\leq r\leq R_i,\;1 \le i \le p$.
\end{assumption}

\begin{theorem}\label{select}
Assume the conditions of Theorem~\ref{phi_convergence}, and Assumptions \ref{assump:penalty} and \ref{fullrank}.
Then for both the joint selection procedure \eqref{eq:joint} and the separate one \eqref{eq:sep},
$$\lim_{T \to \infty} P(\hat{\h{R}}_{\hat p} = \h{R}_{p}) = 1.$$
\end{theorem}

The proof of Theorem~\ref{select} is in Appendix~\ref{appendix:proofs}.

\section{Numerical Results}\label{numerical}

\subsection{Simulations}\label{sec:simulations}

In this section, we study the empirical performances of the proposed estimators and the order and K-rank selection procedures. The simulation studies are grouped into three part: first on the estimation errors, second on the empirical coverage probabilities of the confidence intervals, and the last on the selection of the autoregressive order $p$ and K-ranks $R_i$. Throughout this section, we focus on order-3 tensors ($K=3$). Various combinations of tensor dimensions $d_1$, $d_2$, $d_3$, K-ranks $R_i$ and autoregressive order $p$ are considered. When $p>1$, we set $R_1=\cdots=R_p=R$ in the true model for simplicity.

For all experiments, the data $\cX_t$ are generated from the model \eqref{multiarp}, where the coefficient matrices $\hA_{k}^{(ir)}$ are generated randomly and rescaled so that $\rho(\Phi)=0.8$. For $p=2$, since any VAR($p$) can be rewritten as a VAR($1$), which is also known as the companion form of the VAR($p$) \citep{brockwell2009time}, we set the spectral radius of $2d \times 2d$ coefficient matrix in the companion form VAR($1$) to be $0.8$. Since we focus on order-3 tensors in the simulations, no further identifiability constraints on $\hA_{k}^{(ir)}$ are required except that $\|\hA_{k}^{(ir)}\|_F=1$ for $k=1,2$. The error tensors $\cE_t$ are IID normal with covariance matrix $\Sigma_e:= \mathrm{Cov}[\vect(\mathcal{E}_t)]$, for which we consider three choices, following the discussion at the beginning of Section~\ref{sec:mle}:



\begin{itemize}
    \item Setting I: $\Sigma_e = \boldsymbol{I}$.
    \item Setting II: $\Sigma_e = \boldsymbol{Q} \Lambda \boldsymbol{Q^{\prime}}$, where the elements of the diagonal matrix $\Lambda$ are IID absolute standard normal random variables, and $\boldsymbol{Q}$ is a random orthogonal matrix generated from the Haar measure.
    \item Setting III: $\Sigma_e$ takes the Kronecker product form (\ref{kronecker cov}), where ${\Sigma}_k$, $1 \le k \le K$ are generated similarly as the $\Sigma_e$ in Setting II.
\end{itemize}

\noindent {\bf Simulation I: Estimation Error.} We first consider the model \eqref{multiarp} with $p=1$ and $R=2$, and plot the estimation errors of the LSE and MLE in the log scale. 


\begin{equation*}
    \log \left\| \sum_{r=1}^{2} \left( \hat{\hA}^{(r)}_3 \otimes \hat{\hA}^{(r)}_2 \otimes \hat{\hA}^{(r)}_1 - \hA^{(r)}_3 \otimes \hA^{(r)}_2 \otimes \hA^{(r)}_1 \right)\right\|_F^2.
\end{equation*}
The projection estimator \eqref{multiprovec} (abbreviated as PROJ) and the VAR estimator (obtained by fitting VAR(1) to the vectorized tensors) are also included for comparison. Figures \ref{R2_TenAR(1)_IID} to \ref{R2_TenAR(1)_MLE} are for the three aforementioned settings of $\Sigma_e$ respectively. It is clear from the plots that the VAR estimator, which does not take advantage of the structure of \eqref{multiarp}, is out performed by all other estimators based on the TenAR(1) model, in all cases. Although the PROJ estimator is not as good as LSE and MLE, especially in Setting II and III (i.e. Figures \ref{R2_TenAR(1)_SVD} and \ref{R2_TenAR(1)_MLE}), it still improves significantly from the VAR, and can well serve as the initializer of the LSE and MLE. The performance of LSE and MLE are very similar in Setting I and II, though the LSE is slightly better. On the other hand, in Setting III when $\Sigma_e$ does take the form \eqref{kronecker cov}, the MLE estimator has a much higher estimation accuracy.

\begin{figure}[!ht]
    \centering
    \includegraphics[width = 16cm]{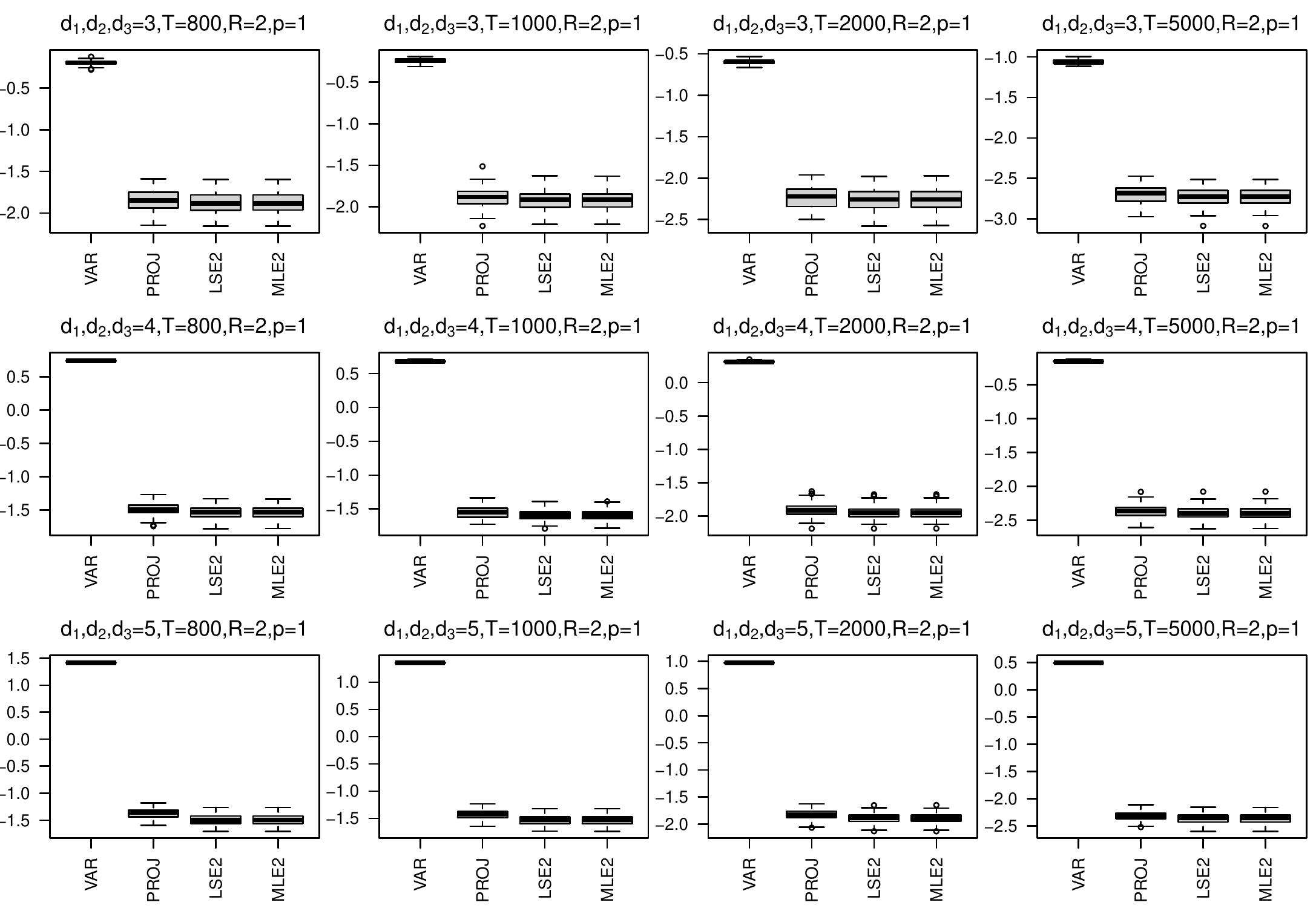}
    \caption{Estimation errors in the log scale. True model is two-term TenAR(1) under setting~I. Comparison of VAR, PROJ, LSE2 (two-term LSE), MLE2 (two-term MLE). We repeat the simulation 100 times. For each row, we fixed the dimension while let $T=800, 1000, 2000, 5000$. For each column, $T$ is fixed while $(d_1, d_2, d_3) = (3,3,3), (4,4,4), (5,5,5)$.}
    \label{R2_TenAR(1)_IID}
\end{figure}

\begin{figure}[!ht]
    \centering
    \includegraphics[width = 16cm]{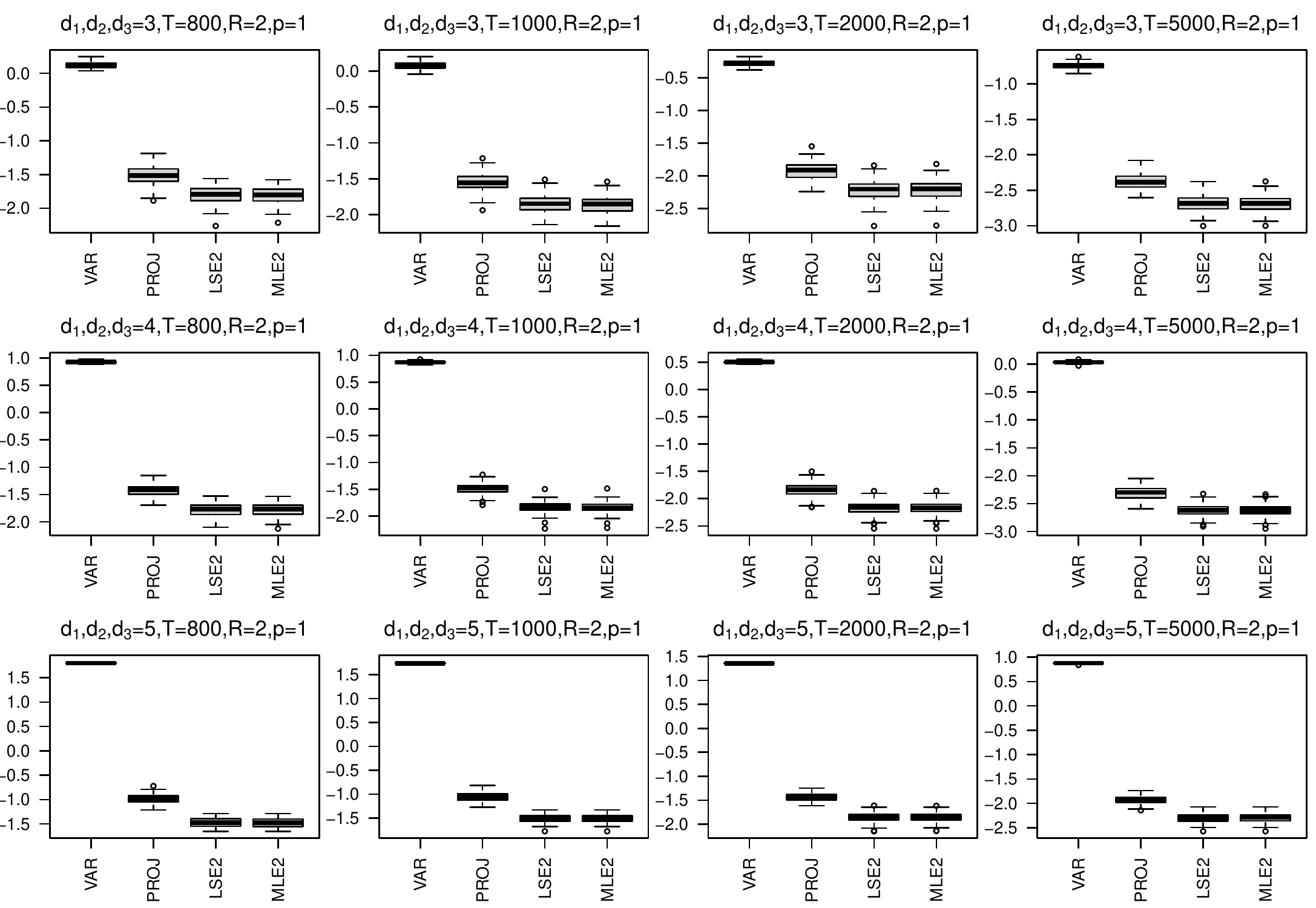}
    \caption{Estimation errors in the log scale. True model is two-term TenAR(1) under setting~II. Comparison of VAR, PROJ, LSE2 (two-term LSE), MLE2 (two-term MLE). We repeat the simulation 100 times. For each row, we fixed the dimension while let $T=800, 1000, 2000, 5000$. For each column, $T$ is fixed while $(d_1, d_2, d_3) = (3,3,3), (4,4,4), (5,5,5)$.}
    \label{R2_TenAR(1)_SVD}
\end{figure}

\begin{figure}[!ht]
    \centering
    \includegraphics[width = 16cm]{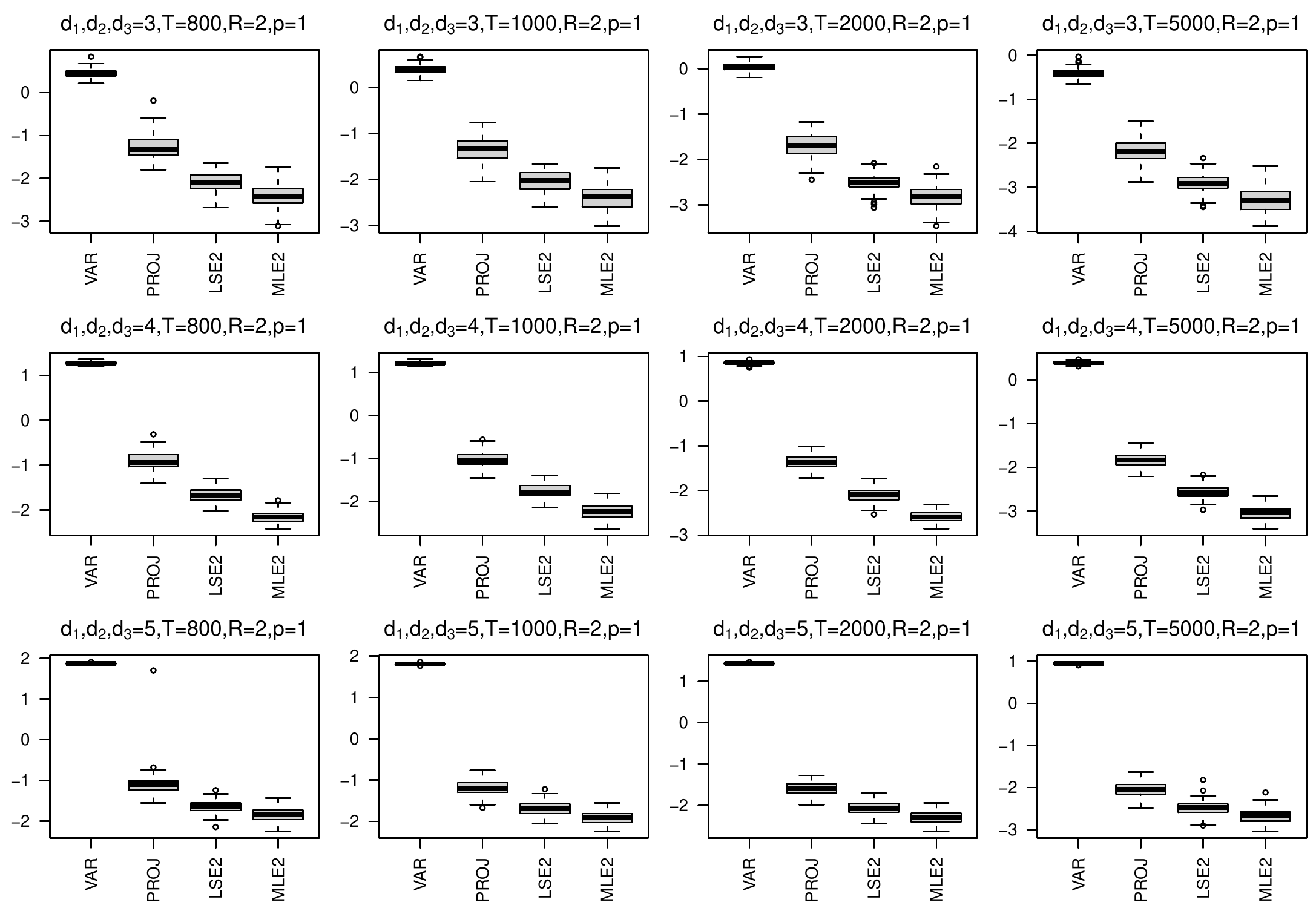}
    \caption{Estimation errors in the log scale. True model is two-term TenAR(1) under setting~III. Comparison of VAR, PROJ, LSE2 (two-term LSE), MLE2 (two-term MLE). We repeat the simulation 100 times. For each row, we fixed the dimension while let $T=900, 1000, 2000, 5000$. For each column, $T$ is fixed while $(d_1, d_2, d_3) = (3,3,3), (4,4,4), (5,5,5)$.}
    \label{R2_TenAR(1)_MLE}
\end{figure}

We also consider the model \eqref{multiarp} with $p=1,\,R=1$ (Figures \ref{R1_TenAR(1)_IID} to \ref{R1_TenAR(1)_MLE}) and $p=2,\,R_1=R_2=1$ (Figures \ref{R1_TenAR(2)_IID} to \ref{R1_TenAR(2)_MLE}). The simulation results confirm the comparisons we have made on different estimators. These additional figures are put in Appendix~\ref{Appendix:sim} for the sake of space.


\noindent {\bf Simulation II: Empirical Coverage of the Confidence Intervals.} In this experiment we look into the empirical coverage of the confidence intervals constructed based on LSE and MLE, and the asymptotic covariance matrices given in Theorems~\ref{lse_arp_multi} and \ref{mle_ar1_multi}. We cosider the TenAR(1) models with $p=1,\,R=1$ and $p=1,\,R=2$, and fix the dimensions at $d_1=d_2=d_3=2$.
Individual $95\%$ confidence intervals are constructed for each entry of $\hA_k^{(r)}$, and empirical coverage probabilities aggregated over all entries and 1000 repetitions are reported in Table~\ref{coverage}. We see the empirical coverage is close to the norminal level $95\%$ for all cases except the MLE in Setting II of $\Sigma_e$. This is not surprising because Theorem~\ref{mle_ar1_multi} is based on the assumption \eqref{kronecker cov}, which does not hold under Setting II. We also note that empirical coverage improves as the sample size increases, confirming the validity of the asymptotic normality.


\begin{table}[ht]
{
\centering
\begin{tabular}{@{}cccccccc@{}}
\toprule
                        & Setting   & \multicolumn{2}{c}{I} & \multicolumn{2}{c}{II} & \multicolumn{2}{c}{III} \\ \cmidrule(l){2-8}
                        & Estimator   & LSE    & MLE     & LSE   & MLE       & LSE   & MLE   \\ \cmidrule(l){1-8}
\multirow{3}{*}{$R=1$}  
                        & T=100            & 0.945  &   0.941   &  0.940 & 0.771  &  0.937 &   0.944\\
                        & T=200            & 0.951  &   0.950   &  0.941 & 0.774  &  0.946 &   0.948\\
                        & T=1000           & 0.953  &   0.952   &  0.951 & 0.776  &  0.955 &   0.956     \\ \cmidrule(l){2-8}
\multirow{3}{*}{$R=2$}  
                        & T=500            & 0.937  &   0.937   &  0.934 & 0.706  & 0.907 & 0.933\\
                        & T=1000           & 0.943  &   0.942   &  0.941 & 0.726 & 0.906 & 0.935\\
                        & T=2000           & 0.950  &   0.950   &  0.943 & 0.724  & 0.920 & 0.936\\ \bottomrule
\end{tabular}
\caption{Percentage of coverages of 95\% confidence intervals.}\label{coverage}
}
\end{table}

\noindent {\bf Simulation III: Model Selection.} 
In the third experiment, we examine the performance of the separate model selection procedure (\ref{eq:sep}), using the information criteria (\ref{ic1}) and (\ref{ic2}). The performance of the joint procedure \eqref{eq:joint} is slightly better, and will not be included here. The data generating model is TenAR(2) with $R_1=R_2=2$. Two choices of dimensions $d_1=d_2=d_3=3$ and $d_1=d_2=d_3=5$ are considered, and both $p_{\max}$ and $R_{\max}$ are capped at 3. We consider Setting II of $\Sigma_e$, which makes the model selection most challenging among the three settings. The empirical frequencies of the correct selection out of 1000 repetitions are reported in Table~\ref{choose}, signifying a satisfactory performance. We note that when $d_1=d_2=d_3=5$, there are in total $d=125$ individual time series under consideration, but the selection is already good enough even when the sample size is merely $T=500$. We also note that $\ic_1$ and $\ic_2$ have very similar performance.




\begin{table}[ht]
\centering
\begin{tabular}{@{}ccccc@{}}
\toprule
                       &       & \multicolumn{3}{c}{$(R_1,R_2)=(2,2)$}  \\ \cmidrule(l){3-5} 
                       &       &  T=500  &  T=800  & T=1000 \\ \midrule
\multirow{3}{*}{$\ic_1$}  
                       &(3,3,3)&(.94, .91, .99)  &(.95, .97, .99)   &(.98, .98, 1)  \\
                       &(5,5,5)&(.98, .98, .98)   &(.98, .98, .99)      &(1, 1, 1)     \\ 
                       \cmidrule(l){2-5}
\multirow{3}{*}{$\ic_2$}  
                    &(3,3,3)&(.94, .89, .99)  &(.95, .96, .99)   &(.98, .98, 1)  \\
                       &(5,5,5)&(.98, .98, .97)   &(.98, .98, .99)      &(1, 1, 1)     \\ \bottomrule
\end{tabular}
\caption{The empirical frequencies that the correct K-ranks is selected by the information criteria $\ic_1$ (\ref{ic1}) and $\ic_2$ (\ref{ic2}) under Setting II, out of 100 repetitions. The signal strength $\rho = 0.8$, $R_{\max} = 3$ and $p_{\max} = 3$. The three numbers in parenthesis refer to the corresponding frequencies of $R_1$, $R_2$ and $R_3$.}\label{choose}
\end{table}

\subsection{Initialization of the Algorithm}\label{local}

As discussed in Section~\ref{sec:proj}, properly setting the initial values of the alternating algorithms (for both LSE and MLE) is critical for them to find the global minimum. We have suggested to use the projection estimators as initializers for both $\hA_k^{(ir)}$ and $\Sigma_k$ (for MLE). In this section we attest this suggestion by comparing the performances of the following initializers, for the TenAR(1) model of dimensions $d_1=d_2=d_3=5$.
\begin{enumerate}
    \item [(i)] SCAL. Each $\hA_{k}^{(r)}$ is initialized as a scalr matrix with diagonal elements 0.5. The speific value 0.5 guarantees that the model satisfies the causality condition at the initial values of the parameters.
    \item [(ii)] RAND. The elements of $\hA_{k}^{(r)}$ are IID $N(0,1/80)$. Again, the choice of the variance $1/80$ is to guaratee the fulfilment of the causality condition at initial values.
    \item [(iii)] PROJ. The initialization by the projection estimators as introduced in Section~\ref{sec:proj}.
    \item [(iv)] PROJ-$m$. We try $m$ random initializers within a small neighborhood of the PROJ estimators, and pick the one that minimizes the least squares (for LSE) or maximizes the likelihood (for MLE).
    \item [(v)] TRUE. For each case, the sum of squares or likelihood based on the true model is minimized/maximized when the true parameters are used to initialize the algorithms.
\end{enumerate}

For all the simulations in this subsection, the true model is TenAR(1) with dimensions $d_1=d_2=d_3=5$, $R=2$, $T=1000$, and covariance matrix $\Sigma_e$ generated according to Setting III, the most general setting among three settings introduced in Section~\ref{sec:simulations}. We put similar results (Figure~\ref{fig: likelihood_mle_iid}) for setting I in Appendix~\ref{Appendix:sim}. We plot the log ratios of the maximum likelihood and the likelihood under the true parameters in Figure~\ref{likelihood_mle}, based on 100 repetitions. The results for the RAND initializer are omitted, as they are much more inferior compared with others. It is seen from the middle figure that when the true model ($R=2$) is fitted, TRUE initilizer leads to the largest likelihood, but the other three initializers work just as well, with PROJ and PROJ-10 being slightly better than SCAL. When $R=1$, the model is under-fitting, and the three initializers PROJ, PROJ-10 and SCAL lead again to about the same likelihood, while the former two are slightly better. When $R=3$ the model is over-fitting, and the SCAL results in much smaller likelihood than PROJ and PROJ-10. For all cases, while PROJ-10 is expected to lead to a larger likelihood than PROJ, we find almost no difference between these two. Therefore, we confirm our suggestion of using PROJ as the initializer in practice.

\begin{figure}[!ht]
    \centering
    \includegraphics[width = 16cm]{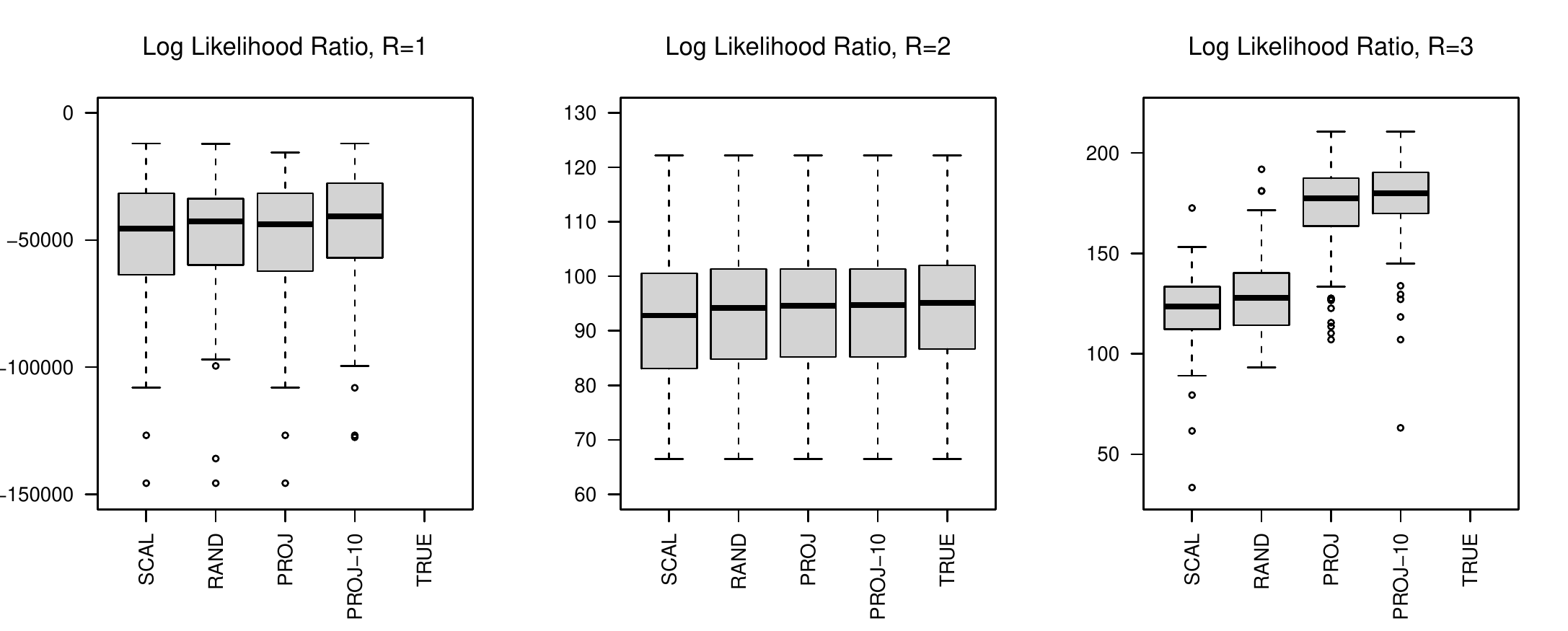}
    \caption{Log likelihood of MLE iterative estimation with different initial values, scaler matrices (SCAL), projection estimator (PROJ), best of $10$ initial values in a neighborhood of PROJ (PROJ-10) and true values (TRUE). True model is two-term ($R=2$) TenAR(1) model, $(d_1,d_2,d_3) = (5,5,5)$, $T=1000$, under Setting III.}\label{likelihood_mle}
\end{figure}

\subsection{Fama-French Research Portfolios}
\label{sec:fama-french}

As the first example, we apply the tensor autoregressive model to analyze the Fama-French research portfolios. We choose the monthly three-way-sorts data based on Size (small and large), Book-to-Market Ratio (four levels from low to high), and Operating Profitability (four levels from low to high), i.e. a $2\times4\times4$ tensor is observed at each month. The data is publicly available at the data library maintained by Prof. Kennth R. French. We consider the range from July 1963 to December 2018.



Both criteria (\ref{ic1}) and (\ref{ic2}) identifies the one-term TenAR(1) model. The estimated coefficients by MLE and the corresponding standard errors of $\hA_1$, $\hA_2$ and $\hA_3$ are reported in Table \ref{A1}, \ref{A2} and \ref{A3}, respectively. $\hA_1$ and $\hA_2$ are rescaled to have Frobenius norm one. We observe that along mode-2 (Book-to-Market Ratio, Table~{A2}), $\cX_t$ depends on $\cX_{t-1}$ through the low and high Book-to-Market Ratio groups.

We report the out-sample mean squared rolling forecast errors of the TenAR(p) models ($p=1,2,3$) in Table~\ref{rolling}. A few other models are also tried as benchmarks for comparison, including:
\begin{itemize}
    \item [--] iAR. Fit an univariate AR model to each individual time series.
    \item [--] VAR. Fit an VAR model to $\{\vect(\cX_t)\}$.
    \item [--] MEAN. Predict $\cX_{t+1}$ by the rolling sample mean up until $t$.
\end{itemize}
The mean squared error $\sum_{t=t_0-1}^{T-1}\|\hat{\mathcal{X}}_{t+1} - \mathcal{X}_{t+1}\|^2_F/\left(d(T-t_0+1)\right)$ over the period from 2001/01 ($t_0=451$) to 2018/12 ($T=666$) are reported in Table~\ref{rolling}, where $d=2\times4\times4=32$ is the size of each tensor $\mathcal{X}_{t}$. We see that although improvement of the 1-term TenAR(1) model over other methods are marginal, reflecting the difficulty of predicting returns, but it does results in the smallest prediction error (with the estimation done by MLE).

\begin{table}[ht]
\centering{
\begin{tabular}{@{}ccccccc@{}}
\toprule
      & LSE   & MLE   & iAR   & VAR   & MEAN & TOTAL \\ \cmidrule(l){2-7}
$p=1$ & 36.74 & 36.57 & 36.91 & 38.76 & 36.96 & 36.94\\
      & 37.96 & 37.73 &       &       & &\\
      & 37.74 & 37.93 &       &       & &\\ 
$p=2$ & 38.52 & 37.93 & 37.06 & 41.18 & &\\
      & 38.82 & 38.79 &       &       & &\\
      & 39.21 & 38.56 &       &       & &\\ 
$p=3$ & 38.72 & 38.10 & 37.17 & 43.96 & &\\
      & 39.62 & 38.74 &       &       & &\\
      & 41.51 & 40.02 &       &       & &\\\bottomrule
\end{tabular}
}
\caption{Mean  squared  rolling  forecast  errors of various models, for the Fama-French data, where $p$ stands for the autoregressive order $p$ of the TenAR, iAR and VAR models. In the first two columns, for each $p$, there are three rows corresponding to TenAR($p$) models with 1, 2 and 3 terms respectively. TOTAL denotes the averaged total sum of squares.}\label{rolling}
\end{table}

\begin{table}[ht]
\centering
\begin{tabular}{c|cc||cc}
      & Small & Big & Small & Big \\ \hline
Small & \begin{tabular}[c]{@{}c@{}}0.233  \end{tabular} & \begin{tabular}[c]{@{}c@{}}0.955  \end{tabular} & +     & +   \\
      & \begin{tabular}[c]{@{}c@{}}(0.104)\end{tabular} & \begin{tabular}[c]{@{}c@{}}(0.161)\end{tabular} &       &    \\
Big   & \begin{tabular}[c]{@{}c@{}}-0.193 \end{tabular} & \begin{tabular}[c]{@{}c@{}}0.795  \end{tabular} & -     & +  \\
      & \begin{tabular}[c]{@{}c@{}}(0.074)\end{tabular} & \begin{tabular}[c]{@{}c@{}}(0.131)\end{tabular} &       &    
\end{tabular}
\caption{The MLE $\hat\hA_1$ of the one-term TenAR(1) model, with standard errors shown in the parentheses. The mode-1 of the tensor corresponds to the Size. The right panel indicates whether the corresponding coefficient is significant and positive (+), significant and negative (-), or insignificant (0) at $5\%$ level.}\label{A1}
\end{table}

\begin{table}[ht]
\centering
\begin{tabular}{c|cccc||cccc}
 & LoBM & 50\%BM & 75\%BM & HiBM & LoBM & 50\%BM & 75\%BM & HiBM \\ \hline
LoBM   & \begin{tabular}[c]{@{}c@{}}-0.413  \end{tabular} & \begin{tabular}[c]{@{}c@{}}0.053 \end{tabular} & \begin{tabular}[c]{@{}c@{}}0.113 \end{tabular} & \begin{tabular}[c]{@{}c@{}}-0.329  \end{tabular} 
& +    & 0      & 0      & +    \\
       & \begin{tabular}[c]{@{}c@{}}(0.062)\end{tabular} & \begin{tabular}[c]{@{}c@{}}(0.089)\end{tabular} & \begin{tabular}[c]{@{}c@{}}(0.098)\end{tabular} & \begin{tabular}[c]{@{}c@{}}(0.054)\end{tabular} 
&      &        &        &      \\
50\%BM & \begin{tabular}[c]{@{}c@{}}-0.253  \end{tabular} & \begin{tabular}[c]{@{}c@{}}0.052 \end{tabular} & \begin{tabular}[c]{@{}c@{}}-0.091  \end{tabular} & \begin{tabular}[c]{@{}c@{}}-0.350  \end{tabular} 
& +    & 0      & 0      & +    \\
       & \begin{tabular}[c]{@{}c@{}}(0.047)\end{tabular} & \begin{tabular}[c]{@{}c@{}}(0.070)\end{tabular} & \begin{tabular}[c]{@{}c@{}}(0.079)\end{tabular} & \begin{tabular}[c]{@{}c@{}}(0.040)\end{tabular} 
&      &        &        &      \\
75\%BM & \begin{tabular}[c]{@{}c@{}}-0.130  \end{tabular} & \begin{tabular}[c]{@{}c@{}}0.023 \end{tabular} & \begin{tabular}[c]{@{}c@{}}-0.111  \end{tabular} & \begin{tabular}[c]{@{}c@{}}-0.359  \end{tabular} 
& +    & 0      & 0      & +    \\
       & \begin{tabular}[c]{@{}c@{}}(0.045)\end{tabular} & \begin{tabular}[c]{@{}c@{}}(0.065)\end{tabular} & \begin{tabular}[c]{@{}c@{}}(0.074)\end{tabular} & \begin{tabular}[c]{@{}c@{}}(0.039)\end{tabular} 
&      &        &        &      \\
HiBM   & \begin{tabular}[c]{@{}c@{}}-0.165  \end{tabular} & \begin{tabular}[c]{@{}c@{}}0.122 \end{tabular} & \begin{tabular}[c]{@{}c@{}}-0.241  \end{tabular} & \begin{tabular}[c]{@{}c@{}}-0.499  \end{tabular} 
& +    & 0      & +      & +    \\
       & \begin{tabular}[c]{@{}c@{}}(0.045)\end{tabular} & \begin{tabular}[c]{@{}c@{}}(0.066)\end{tabular} & \begin{tabular}[c]{@{}c@{}}(0.082)\end{tabular} & \begin{tabular}[c]{@{}c@{}}(0.047)\end{tabular} 
&      &        &        &    
\end{tabular}
\caption{The MLE $\hat\hA_2$ of the one-term TenAR(1) model, corresponding to the Book-to-Market Ratio.}\label{A2}
\end{table}

\begin{table}[ht]
\centering
\begin{tabular}{c|cccc||cccc}
 & LoOP & 50\%OP & 75\%OP & HiOP & LoOP & 50\%OP & 75\%OP & HiOP \\ \hline
LoOP   & \begin{tabular}[c]{@{}c@{}}-0.422  \end{tabular} & \begin{tabular}[c]{@{}c@{}}-0.206  \end{tabular} & \begin{tabular}[c]{@{}c@{}}0.459 \end{tabular} & \begin{tabular}[c]{@{}c@{}}-0.076  \end{tabular} 
& +    & +      & -      & +    \\
       & \begin{tabular}[c]{@{}c@{}}(0.065)\end{tabular} & \begin{tabular}[c]{@{}c@{}}(0.060)\end{tabular} & \begin{tabular}[c]{@{}c@{}}(0.030)\end{tabular} & \begin{tabular}[c]{@{}c@{}}(0.033)\end{tabular} 
&      &        &        &      \\
50\%OP & \begin{tabular}[c]{@{}c@{}}-0.242  \end{tabular} & \begin{tabular}[c]{@{}c@{}}-0.180  \end{tabular} & \begin{tabular}[c]{@{}c@{}}0.339 \end{tabular} & \begin{tabular}[c]{@{}c@{}}-0.114  \end{tabular} 
& +    & +      & -      & +    \\
       & \begin{tabular}[c]{@{}c@{}}(0.047)\end{tabular} & \begin{tabular}[c]{@{}c@{}}(0.052)\end{tabular} & \begin{tabular}[c]{@{}c@{}}(0.027)\end{tabular} & \begin{tabular}[c]{@{}c@{}}(0.032)\end{tabular} 
&      &        &        &      \\
75\%OP & \begin{tabular}[c]{@{}c@{}}-0.177  \end{tabular} & \begin{tabular}[c]{@{}c@{}}-0.153  \end{tabular} & \begin{tabular}[c]{@{}c@{}}0.259 \end{tabular} & \begin{tabular}[c]{@{}c@{}}-0.128  \end{tabular} 
& +    & +      & -      & +    \\
       & \begin{tabular}[c]{@{}c@{}}(0.042)\end{tabular} & \begin{tabular}[c]{@{}c@{}}(0.050)\end{tabular} & \begin{tabular}[c]{@{}c@{}}(0.034)\end{tabular} & \begin{tabular}[c]{@{}c@{}}(0.033)\end{tabular} 
&      &        &        &      \\
HiOP   & \begin{tabular}[c]{@{}c@{}}-0.123  \end{tabular} & \begin{tabular}[c]{@{}c@{}}-0.289 \end{tabular} & \begin{tabular}[c]{@{}c@{}}0.297  \end{tabular} & \begin{tabular}[c]{@{}c@{}}-0.137  \end{tabular} 
& +    & +      & -      & +    \\
       & \begin{tabular}[c]{@{}c@{}}(0.045)\end{tabular} & \begin{tabular}[c]{@{}c@{}}(0.056)\end{tabular} & \begin{tabular}[c]{@{}c@{}}(0.041)\end{tabular} & \begin{tabular}[c]{@{}c@{}}(0.038)\end{tabular} 
&      &        &        &    
\end{tabular}
\caption{The MLE $\hat\hA_3$ of the one-term TenAR(1) model, corresponding the Operating Profitability.}\label{A3}
\end{table}

\subsection{Taxi traffic in New York city}

In this section, we apply the tensor autoregressive model to analyze the New York taxi traffic data. The data includes information of individual taxi rides operated by Yellow Taxi within New York City, from January 1, 2009 to December 31, 2019. It is maintained by the Taxi \& Limousine Commission of New York City and published at \url{https://www1.nyc.gov/site/tlc/about/tlc-trip-record-data.page}. 
We consider the pick-up, drop-off locations and the hour of the pick-up for each ride, and count the number of rides from region $i$ to region $j$, during hour $k$. As a result, an order three-tensor $\cX_t=\{\cX_{t,ijk}\}$ is observed for each day $t$.

\cite{Chen2019FactorMF} introduced the factor model for tensor time series and applied it to analyze the taxi data. One of their findings is that the underlying factors are heavily loaded for certain areas and hours. The traffic among regions around Midtown and Times Square during 8am to 3pm on business days is such an example. Following this identified pattern, here we consider the pick-up and drop-off locations in four heavily loaded regions: Midtown Center, Midtown East, Midtown North and Times Square, and include rides between 8am to 3pm of business days. The observations from business days between January 1, 2009 and December 31, 2019 lead to a tensor time series of length $T=2768$. Each observation $\cX_t$ is an order-3 $4\times 4\times 7$ tensor. 

The data exhibit some strong and persistent trend, especially due to the impact of Uber and Lyft after 2015. We estimate the trend of each individual series by an exponential smoothing, which is then removed from the original data. Specifically, the trend of the series $\{\cX_{t,ijk}\}$ is estimated as $\{\mathcal{S}_{t, ijk}\}$, given by:
\begin{align*}
   \mathcal{S}_{1,ijk} &= \cX_{1,ijk}, \\
   \mathcal{S}_{t,ijk} &= \alpha \cX_{t,ijk} + (1-\alpha) \mathcal{S}_{t-1,ijk}, \quad t>1,
\end{align*}
where $\alpha$ is the smoothing factor. After some exploratory analysis, we decide to set $\alpha = 2/(63+1)$, which corresponds roughly to the $63$-day one sided moving average. We then apply the tensor autoregressive and other models to the de-trended series  $\cY_t=\cX_t-\mathcal S_t$. 

Similar to Section~\ref{sec:fama-french}, we compare the performance of the following prediction methods in terms of the mean squared rolling forecast errors.
\begin{itemize}
    \item [--] LSE and MLE: TenAR models estimated by LSE or MLE.
    \item [--] iAR. Fit an univariate AR model to each individual time series.
    \item [--] VAR. Fit an VAR model to $\{\vect(\cX_t)\}$.
    \item [--] ES. Predict $\cX_{t+1}$ by the trend $\mathcal{S}_{t+1}$ estimated by exponential smoothing.
    \item [--] RW. Random walk prediction: predict $\cX_{t+1}$ by $\cX_t$.
\end{itemize}

For the model based forecasts (LSE, MLE, iAR, and VAR), we first generate the prediction $\hat\cY_{t+1}$ by $\{\cY_1,\ldots,\cY_t\}$, based on the fitted model, then predict $\cX_{t+1}$ as $\hat\cX_{t+1}=\mathcal S_{t} + \hat\cY_{t+1}$.
The mean squared rolling forecast errors over the last two years $\sum_{t=t_0-1}^{T-1}\|{\mathcal{X}}_{t+1} - \hat{\mathcal{X}}_{t+1}\|^2_F/\left(d(T-t_0+1)\right)$ are reported in Table~\ref{taxi_roling}, where $d=4\times4\times7=112$, and $t_0$ and $T$ correspond to 01/01/2018 and 12/31/2019 respectively. 
We see that all TenAR($p$) model outperform the individual AR($p$) models and VAR($p$) models. Because of the overfitting, VAR models are the worst especially after $p\ge3$. For TenAR($p$) and individual AR($p$) models, it seems that the mean square errors are decreasing as $p$ increases for $p \le 5$. The best model is found by MLE using seven-day lags $p=7$. It implies that due to the cyclical nature of business-day data, the seasonal model may better fit the taxi data. Also, increasing the K-ranks may not always result in a better performance, for example, where TenAR(2) MLE with K-rank $\h{R}=(1,1)$ performs better than $\h{R}=(3,3)$. In this example, there is no obvious difference between the LSE and MLE, while for best model $p=5$, TenAR(p) with MLE performs slightly better than LSE. Overall, the TenAR model outperforms other autoregressive models for tensor-valued time series.

\begin{table}[ht]
\centering{
\scalebox{0.9}{
\begin{tabular}{@{}ccccccccc@{}}
\toprule
      & LSE   & MLE   & iAR   & VAR   & ES  & RW  & TOTAL \\ \cmidrule(l){2-8} 
$p=1$ & 50.72 & 50.33 & 51.64 & 52.46 & 50.59 & 83.47 & 56.48\\
      & 50.54 & 50.49 &       &       &       &       & \\
      & 50.17 & 50.81 &       &       &       &       & \\ 
$p=2$ & 48.43 & 48.46 & 51.42 & 51.33 &       &       & \\
      & 47.69 & 48.09 &       &       &       &       & \\
      & 48.06 & 48.52 &       &       &       &       & \\ 
$p=3$ & 48.02 & 47.94 & 51.31 & 53.58 &       &       & \\
      & 47.76 & 48.00 &       &       &       &       & \\
      & 47.78 & 47.94 &       &       &       &       & \\ 
$p=4$ & 47.71 & 47.83 & 50.96 & 55.90 &       &       & \\ 
$p=5$ & 47.67 & 47.18 & 49.94 & 58.32 &       &       & \\ 
$p=6$ & 47.73 & 47.02 & 49.98 & 61.57 &       &       & \\ 
$p=7$ & 47.56 & 46.57 & 49.46 & 64.91 &       &       & \\ 
$p=8$ & 47.84 & 46.74 & 49.36 & 70.41 &       &       & \\ \bottomrule
\end{tabular}
}
}
\caption{Mean squared rolling forecast errors of the taxi data. $p$ stands for the autoregressive order of the TenAR, iAR and VAR models. For each $p=1,2,3$, the three rows for LSE and MLE correspond to the TenAR($p$) model with 1, 2, 3 terms at each lag, respectively. When the autoregressive order is larger ($p=4,5,6$), the TenAR($p$) model only has one term for each lag.
The averaged total sum of squares, still denoted by TOTAL, is also reported for comparison.} \label{taxi_roling}
\end{table}


\section{Conclusion} 

We proposed an multi-linear autoregressive model for tensor-valued time series (TenAR), and consider the extension of including multiple terms and multiple lags in the autoregression. 
Both the LSE and the MLE (under a separable covariance tensor) are introduced and theoretically studied: including the asymptotic normality for the fixed-dimensional case, and convergence rates for the high-dimensional setup. We emphasize the importance of initialization in the alternating algorithms for both LSE and MLE, and suggest using the projection estimator as the initializer. 
We also propose to use the information criteria to select the autoregressive order and the number of terms for each lag, and establish the model selection consistency. 


There are a number of directions which are worth further investigations. First, we have discussed the high dimensional TenAR models in a very general setting, with minimal assumptions. If additional structures are imposed, e.g. low rankness, sparsity of $\hA_i$, it is natural to anticipate faster convergence rates. It will also be possible to make inferences for the high dimensional models as well. Second, from a practical point of view, it is interesting and important to consider the seasonal TenAR models. During the analysis of the Taxi data, we notice that the TenAR(6) and TenAR(7) model have the best predictive performance (see Table~\ref{taxi_roling} for details). This is suggesting that the data may have the cyclic behavior with period 5, which is very reasonable since we are considering the data over business days. Following the development of the univariate and vector autoregressive models, it is natural to include both the regular and seasonal autoregression in multiplicative forms, which will create challenges for both the algorithmic and theoretic aspects under the TenAR models introduced in this paper.

\clearpage

\bibliography{mybibfile}

\clearpage

\begin{appendices}
\section{Some Notations and Basics}\label{appendix:basic}
\nomenclature[R]{$T$}{Time length or sample size in the time series model}
\nomenclature[R]{$d_k$}{The dimension of coefficient matrix in $k$-th mode, i.e. $\hA_k \in \mathbb{R}^{d_k \times d_k}$}
\nomenclature[R]{$N$}{The dimension of coefficient matrix $\Phi$, $N=d_1\cdots d_K$}
\nomenclature[R]{$K$}{The number of mode of tensor $\mathcal{X}_t$}
\nomenclature[R]{$p$}{The number of order in the autoregressive model}
\nomenclature[R]{$R_i$}{The number of terms in $i$-th order. Denotes as $R$ in TenAR(1)}

\nomenclature[T]{$\mathcal{R}_t$}{$\mathcal{R}_t = \mathcal{X}_{t} - \sum_{i=1}^{p} \sum_{r=1}^{R_i} \mathcal{X}_{t-1} \times_{1} \hA_{1}^{(ir)} \times_{2} \cdots \times_{K} \hA_{K}^{(ir)}$}

\nomenclature[T]{$\mathcal{E}_t$}{Error term in tensor form at time $t$}

\nomenclature[T]{$\mathcal{X}_t$}{Tensor-valued time series at time $t$}

\nomenclature[M]{$\hA_{k}^{(ir)}$}{Coefficient matrix for lag-order $i$ and term $r$ of $k$-th mode in multi-term TenAR($p$) model}
\nomenclature[M]{$\hA_{k}^{(r)}$}{Coefficient matrix for term $r$ of $k$-th mode in multi-term TenAR($1$) model}
\nomenclature[M]{$\hA_{k}$}{Coefficient matrix of $k$-th mode in one-term TenAR($1$) model}
\nomenclature[M]{$\ha_{k}^{(ir)}$}{Vectorized coefficient matrix $\vect(\hA_{k}^{(ir)})$}
\nomenclature[M]{$\ha_{k}^{(r)}$}{Vectorized coefficient matrix $\vect(\hA_{k}^{(r)})$}
\nomenclature[M]{$\ha_{k}$}{Vectorized coefficient matrix $\vect(\hA_{k})$}

\nomenclature[M]{$\Phi$}{Coefficient matrix of $\vect(\mathcal{X}_{t-1})$ in multi-term or one-term TenAR($1$)}

\nomenclature[M]{$\Phi^{(i)}$}{$\Phi^{(i)} = \sum_{r=1}^{R_i} \Phi^{(ir)}$ denotes coefficient matrix of $i$-th order $\vect(\mathcal{X}_{t-i})$ in multi-term TenAR($p$)}

\nomenclature[M]{$\Phi^{(ir)}$}{$\Phi^{(i r)} = \hA^{(i r)}_{K} \otimes \cdots \otimes \hA^{(ir)}_{1}$ denotes the $r$-th term in $\Phi^{(i)}$}

\nomenclature[M]{$\Phi_{k}^{(r)}$}{$\Phi_{k}^{(r)} = \hA^{(r)}_{K} \otimes \cdots \otimes \hA^{(r)}_{k+1} \otimes \hA^{(r)}_{k-1} \otimes \cdots \otimes \hA^{(r)}_{1}$ denotes the $\Phi^{(1r)}$ without the $k$-th mode coefficient matrix in multi-term TenAR(1) model}

\nomenclature[M]{$\Phi_{k}$}{$\Phi_{k} = \hA_{K} \otimes \cdots \otimes \hA_{k+1} \otimes \hA_{k-1} \otimes \cdots \otimes \hA_{1}$ denotes the $\Phi$ without the $k$-th mode coefficient matrix in one-term TenAR(1) model}

\nomenclature[M]{$\Phi_{k}^{(ir)}$}{$\Phi_{k}^{(i r)} = \hA^{(i r)}_{K} \otimes \cdots \otimes \hA^{(ir)}_{k+1} \otimes \hA^{(ir)}_{k-1} \otimes \cdots \otimes \hA^{(ir)}_{1}$ denotes the $\Phi^{(i r)}$ without the $k$-th mode coefficient matrix in multi-term TenAR($p$) model}

\nomenclature[M]{$\boldsymbol{X}_{(k)}$}{$k$-th mode matricization of tensor $\mathcal{X}$}

\printnomenclature


In this section we introduce some permutation matrices and basic properties, which are frequently used in our proof. For notations about tensors, we follow the paper \cite{kolda2009tensor}. The details of notations are introduced in section 2. In the follow, we introduce two types of permutation matrices and their properties.

\begin{definition}
Permutation Matrix $\boldsymbol{P}_{m,n}$ is defined as,
\begin{equation}\label{P}
    \boldsymbol{P}_{m,n} = \sum\limits_{i=1}^{n}\sum\limits_{j=1}^{m} (\boldsymbol{U}_{ij} \otimes \boldsymbol{U}_{ij}^{\prime})
\end{equation}
where $\boldsymbol{U}_{ij}$ is the $n\times m$ matrix with one in $(i,j)$ element and zeros elsewhere. Vector Permutation Matrix  $\boldsymbol{Q}_{k}$ is defined as,
\begin{equation}\label{Q}
    \boldsymbol{Q}_{k}=\boldsymbol{I}_{d_K d_{K-1} \cdots d_{k+1}} \otimes \boldsymbol{P}_{d_k, d_{k-1} d_{k-2}\cdots d_1}
\end{equation}
where $\boldsymbol{I}$ is identity matrix and $\boldsymbol{P}$ is permutation matrix, $1 < k < K$. Especially, when $k=1$, $\boldsymbol{Q}_1 = \boldsymbol{I}_{d_K d_{K-1} \cdots d_{1}}$; when $k=K$, $\boldsymbol{Q}_K = \boldsymbol{P}_{d_K, d_{K-1} d_{K-2}\cdots d_1}$
\end{definition}


\begin{proposition}\label{permutation_properties}
For tensor $\mathcal{X} \in \mathbb{R}^{d_1 \times \cdots \times d_K}$ and its unfolding $\boldsymbol{X}_{(k)}$, $1 \le k \le K$, 
\begin{enumerate}
    \item[(i)] $\vect (\boldsymbol{X}_{(k)}) = \boldsymbol{Q}_{k} \vect(\mathcal{X}),$.
    \item[(ii)] Permutation Matrix $\boldsymbol{P}_{m,n}$ is nonsingular, and $(\boldsymbol{P}_{m,n})^{-1}=(\boldsymbol{P}_{m,n})^{\prime}=\boldsymbol{P}_{n,m}$.
    \item[(iii)] Vector Permutation Matrix $\boldsymbol{Q}$, is nonsingular, and $\boldsymbol{Q}^{-1} = \boldsymbol{Q}^{\prime}$.
\end{enumerate}

\end{proposition}






Next we introduce some properties relate to outer product, Kronecker product and permutation matrices.

\begin{proposition}\label{basic_properties}
For vectors $\ha_i \in \mathbb{R}^{d_i}$, $1\le i\le K$, we have
\begin{enumerate}
    \item[(i)] $\ha_i \circ \ha_j = \ha_i \ha^{\prime}_j = \ha_i \otimes \ha^{\prime}_j$, and $\ha_i \otimes \ha_j = \vect(\ha_j \ha_i^{\prime})$.
    \item[(ii)] $(\ha_1 \circ  \cdots \circ \ha_K)_{(i)} = \ha_{i} \vect^{\prime}(\ha_1 \circ \cdots \ha_{i-1}\circ \ha_{i+1} \cdots \circ \ha_K)$.
    \item[(iii)] $\vect(\ha_1 \circ \cdots \circ \ha_K) = \ha_K \otimes \cdots \otimes \ha_1$.
\end{enumerate}

For matrix $\hA \in \mathbb{R}^{d_{11} \times d_{12}}$, matrix $\boldsymbol{B} \in \mathbb{R}^{d_{21} \times d_{22}}$ and matrix $\boldsymbol{C} \in \mathbb{R}^{d_{31} \times d_{32}}$, we have
\begin{enumerate}
    \item[(iv)] $\boldsymbol{B} \otimes \hA = \boldsymbol{P}_{d_{11},d_{21}} (\hA \otimes \boldsymbol{B}) \boldsymbol{P}_{d_{22},d_{12}}$
    \item[(v)] $\boldsymbol{C}\otimes \hA \otimes \boldsymbol{B} = \boldsymbol{P}_{d_{11}d_{21},d_{31}} [\hA \otimes \boldsymbol{B} \otimes \boldsymbol{C}] \boldsymbol{P}_{d_{12},d_{22}d_{32}}$
    \item[(vi)] $\vect(\hA) = \boldsymbol{P}_{d_{11},d_{12}} \vect(\hA^{\prime})$
\end{enumerate}

\end{proposition}

\begin{proposition}\label{pur}
Let $\hA_k$ be $d_k \times d_k$ matrix, $1 \le k \le K$. Then
$$\vect(\hA_1 \otimes  \cdots \otimes \hA_K) = \boldsymbol{T}_K \vect(\ha_K \circ \cdots \circ \ha_1)$$
where $\boldsymbol{T}_K= \boldsymbol{M}_K\left(\cdots \boldsymbol{M}_4\left(\boldsymbol{M}_3(\boldsymbol{M}_2 \otimes \boldsymbol{I}_{d_3^2})\otimes \boldsymbol{I}_{d_4^2}\right)\cdots \right) \otimes \boldsymbol{I}_{d_{K}^2}$, $\boldsymbol{M}_k = \boldsymbol{I}_{d_1 d_2\cdots d_{k-1}} \otimes \boldsymbol{P}_{d_1 d_2\cdots d_{k-1}, d_k} \otimes \boldsymbol{I}_{d_k}$ for $2 \le k \le K$, $\boldsymbol{I}$ is identity matrix and $\boldsymbol{P}$ is permutation matrix defined in appendix.
\end{proposition}

\begin{remark}
Above expression may be intricate but will have much simpler expression in lower dimensions. We can specify two important cases here.
\begin{itemize}
\item When $K=2$, $\vect(\hA_1 \otimes \hA_2) = \boldsymbol{M}_2(\ha_1 \otimes \ha_2) = \boldsymbol{M}_2 (\vect(\ha_2 \circ \ha_1)) = \boldsymbol{M}_2 (\vect(\ha_2  \ha_1^{\prime}))$.

\item When $K=3$, $\vect(\hA_1 \otimes \hA_2 \otimes \hA_3) = \boldsymbol{M}_3(\boldsymbol{M}_2 \otimes \boldsymbol{I}_{d_3^2}) \vect(\ha_3 \circ \ha_{2} \circ \ha_1)$.
\end{itemize}
\end{remark}

\begin{remark}
The closed form of re-arrangement operator $\mathcal{R}$ discussedin Section~{\ref{sec:tensor}} is 
$$\mathcal{R}(\cdot)=\vect^{-1}\boldsymbol{T}_K^{-1}\vect(\cdot).$$
it's easy to check such $\boldsymbol{T}_K$ is invertible.
\end{remark}

\section{Additional Theorems}\label{appendix:theorems}
\subsection{Asymptotics for projection estimators in one-term TenAR(1) model}

Since TenAR(1) model also has the form (\ref{eq:multi_tenAR_vec}), following standard theory of multivariate ARMA models \citep{dunsmuir1976vector, hannan2009multiple}, $\hat{\Phi}$ converges to a multivariate normal distribution:
\begin{equation*}
    \sqrt{T}\vect(\hat{\Phi} - \hA_K \otimes \cdots \otimes \hA_1) \Rightarrow \mathcal{N}(0,\Gamma_0^{-1} \otimes \Sigma),
\end{equation*}
where $\Sigma$ is the covariance matrix of $\vect(\cE_t)$ and $\Gamma_0 = \textup{Cov}(\vect(\mathcal{X}_t),\vect(\mathcal{X}_{t}))$. Thus, we have

\begin{equation*}
    \sqrt{T}\vect[(\mathcal{R}({\Phi}) - \ha_1 \circ \cdots \circ \ha_K)_{(1)}] \Rightarrow \mathcal{N}(0,\Xi_1),
\end{equation*}
where $(\cdot)_{(1)}$ denote mode-1 matricization of the tensor. The matrix $\Xi_1$ is obtained by rearranging the entries of $\Gamma_0^{-1} \otimes \Sigma$, such rearrangement operator is discussed before and the explicit form is given in Appendix.

\begin{theorem}\label{proj_arp_multi}
Consider one-term TenAR(1) model. Set $\ha_k = \vect(\hA_k)$,  $\boldsymbol{\beta}_k = \vect(\ha_1 \circ \cdots \circ \ha_{k-1} \circ \ha_{k+1} \circ \cdots  \circ \ha_K)$, $1 \le k \le K$, $\boldsymbol{P}$ and $\boldsymbol{Q}$ are permutation matrices defined in Appendix.
\begin{align*}
\boldsymbol{V}_0 =& \begin{pmatrix} 
(\boldsymbol{\beta}_1^{\prime} \otimes \boldsymbol{I}_{d_1})(\boldsymbol{Q}_{1}) - (\boldsymbol{\beta}_1^{\prime} \otimes \ha_1)(\boldsymbol{I}_{d_K} \otimes \boldsymbol{\beta}_1 \boldsymbol{\beta}_K^{\prime})\boldsymbol{P}_{d_K, d_1\cdots d_{K-1}} \boldsymbol{Q}_{K} \\
\cdots\\
(\boldsymbol{\beta}_k^{\prime} \otimes \boldsymbol{I}_{d_k})(\boldsymbol{Q}_{k}) - (\boldsymbol{\beta}_k^{\prime} \otimes \ha_k)(\boldsymbol{I}_{d_K} \otimes \boldsymbol{\beta}_k \boldsymbol{\beta}_K^{\prime}) \boldsymbol{P}_{d_K, d_1\cdots d_{K-1}} \boldsymbol{Q}_{K}
\\
\cdots \\
(\boldsymbol{\beta}_K^{\prime} \otimes \boldsymbol{I}_{d_K}) \boldsymbol{Q}_{K}
\end{pmatrix} 
\end{align*}
Assume that $\mathcal{E}_1, \cdots, \mathcal{E}_{T}$ are IID with mean zero and finite second moments. Also assume the causality condition $\prod_{k=1}^{K} \rho(\hA_k) < 1$, and coefficient matrices $\hA_k$,  $1 \le k \le K$, and $\Sigma$ are nonsingular. Then it holds that
\begin{equation*}
 \sqrt{T}\begin{pmatrix} 
\vect(\bar{\hA}_{1} - \hA_1) \\
\cdots \\
\vect(\bar{\hA}_K - \hA_K)
\end{pmatrix} \to \mathcal{N}(0,\h{V}_0 {\Xi}_1 \h{V}_0^{\prime})
\end{equation*}
\end{theorem}
The proof of the theorem is presented in Appendix. When $K=2$, above Theorem degenerates to the form given by \cite{chen2020autoregressive}.

\subsection{Corollaries of Theorem~\ref{lse_arp_multi}}

Corollary~\ref{lse_ar1_one} is for LSE estimators in one-term TenAR(1) model with $p=1$ and $R_1=1$.
\begin{corollary}\label{lse_ar1_one}
Define $\boldsymbol{H} := \mathbb{E}(\boldsymbol{W}_t \boldsymbol{W}_t^{\prime}) + \sum_{k=1}^{K-1} \h{\gamma}_k \h{\gamma}_k^{\prime} $, $\Xi_2 =: \boldsymbol{H}^{-1} \mathbb{E}(\boldsymbol{W}_t \Sigma \boldsymbol{W}_t^{\prime}) \boldsymbol{H}^{-1}$ where
\begin{align*}
\boldsymbol{W}_t =& \begin{pmatrix} 
((\boldsymbol{X}_{t(1)}\Phi_1^{\prime}) \otimes \boldsymbol{I}_{d_1}) \boldsymbol{Q}_1 \\
\cdots \\
((\boldsymbol{X}_{t(K)}\Phi_K^{\prime}) \otimes \boldsymbol{I}_{d_K}) \boldsymbol{Q}_K
\end{pmatrix} 
\end{align*}
where $\Phi_k=\hA_K \otimes  \cdots  \hA_{k-1}\otimes \hA_{k+1}  \cdots \otimes \hA_1$ and $\h{\gamma}_{k} := (0^{\prime},\ha_k^{\prime},0^{\prime})^{\prime}$ be a vector in $\mathbb{R}^{d_1^{2} + \cdots + d_K^{2}}$, where the first $\h{0} \in \mathbb{R}^{d_1^2 + \cdots + d_{k-1}^2}$ and the later $\h{0} \in \mathbb{R}^{d_{k+1}^2 + \cdots + d_K^2}$ for $1 \le k \le K$. Under same condition of Theorem~\ref{lse_arp_multi}. It holds that
\begin{equation*}
 \sqrt{T}\begin{pmatrix} 
\vect(\hat{\hA}_{1} - \hA_1) \\
\cdots \\
\vect(\hat{\hA}_K - \hA_K)
\end{pmatrix} \to \mathcal{N}(0,{\Xi}_2)
\end{equation*}
\end{corollary}
Next consider multi-term TenAR(1) model with $p=1$ and $R_1=R$.
\begin{corollary}\label{lse_ar1_multi}
Define $\boldsymbol{H} := \mathbb{E}(\boldsymbol{W}_t \boldsymbol{W}_t^{\prime}) + \sum_{r=1}^{R}\sum_{k=1}^{K-1} \h{\gamma}_{k}^{(r)} \h{\gamma}_{k}^{(r)^{\prime}}$, $\Xi_2 =: \boldsymbol{H}^{-1} \mathbb{E}(\boldsymbol{W}_t {\Sigma} \boldsymbol{W}_t^{\prime}) \boldsymbol{H}^{-1}$,
\begin{align*}
\boldsymbol{W}_t = \begin{pmatrix} 
\boldsymbol{W}^{(1)}_t \\
\cdots\\
\boldsymbol{W}^{(R)}_t\\
\end{pmatrix}, \
\boldsymbol{W}^{(r)}_t = \begin{pmatrix} 
((\boldsymbol{X}_{t(1)}{\Phi^{(r)}_1}^{\prime}) \otimes \boldsymbol{I}_{d_1}) \boldsymbol{Q}_1 \\
\cdots\\
((\boldsymbol{X}_{t(K)}{\Phi^{(r)}_K}^{\prime}) \otimes \boldsymbol{I}_{d_K})  \boldsymbol{Q}_K\\
\end{pmatrix}, \ 1 \le r \le R
\end{align*}
where $\Phi^{(r)}_k= \hA^{(r)}_K \otimes \cdots \hA^{(r)}_{k-1} \otimes \hA^{(r)}_{k+1} \cdots \otimes \hA^{(r)}_1$ and $\h{\gamma}_{k}^{(r)} := (0^{\prime},\ha_k^{ (r)\prime},0^{\prime})^{\prime}$ be a vector in $\mathbb{R}^{R(d_1^{2} + \cdots + d_K^{2})}$, where the first $\h{0} \in \mathbb{R}^{(r-1)(d_1^2 + \cdots + d_K^2) + d_1^2 + \cdots + d_{k-1}^2}$ and the later $\h{0} \in \mathbb{R}^{(R-r)(d_1^2 + \cdots + d_K^2) + d_{k+1}^2 + \cdots + d_K^2}$ for $1 \le k \le K$, $1 \le r \le R$. Under same condition of Theorem~\ref{lse_arp_multi}. It holds that
\begin{equation*}
 \sqrt{T}\begin{pmatrix} 
vec(\hat{\hA}^{(1)}_{1} - \hA_1) \\
\cdots \\
vec(\hat{\hA}^{(R)}_{K} - \hA_K)
\end{pmatrix} \to \mathcal{N}(0,\Xi_2)
\end{equation*}
\end{corollary}

\subsection{Corollaries of Theorem~\ref{mle_arp_multi}}

For MLE estimators in one-term TenAR(1) model with $p=1$ and $R_1=1$, we have Corollary~\ref{mle_ar1_one}.
\begin{corollary}\label{mle_ar1_one}
Define $\boldsymbol{H} := \mathbb{E}(\boldsymbol{W}_t {\Sigma}^{-1} \boldsymbol{W}_t^{\prime}) + \sum_{k=1}^{K-1} \h{\gamma}_{k} \h{\gamma}_{k}^{ \prime}$, where $\boldsymbol{W}_t$ and $\h{\gamma}_{k}$ are defined in Corollaries~\ref{lse_ar1_one}. Let $\Xi_3 =: \boldsymbol{H}^{-1} \mathbb{E}(\boldsymbol{W}_t {\Sigma}^{-1} \boldsymbol{W}_t^{\prime}) \boldsymbol{H}^{-1}$. Assume the same conditions as Theorem~\ref{mle_arp_multi}. It holds that
\begin{equation*}
 \sqrt{T}\begin{pmatrix} 
vec(\tilde{\hA}_{1}^{(11)} - \hA_{1}^{(11)}) \\
\cdots \\
vec(\tilde{\hA}^{(pR_p)}_{K} - \hA^{(pR_p)}_K)
\end{pmatrix} \to \mathcal{N}(0,\Xi_3).
\end{equation*}
\end{corollary}

Corollary~\ref{mle_ar1_multi} is for MLE estimators in multi-term TenAR(1) model with $p=1$ and $R_1=R$.
\begin{corollary}\label{mle_ar1_multi}
Define $\boldsymbol{H} := \mathbb{E}(\boldsymbol{W}_t {\Sigma}^{-1} \boldsymbol{W}_t^{\prime}) + \sum_{r=1}^{R}\sum_{k=1}^{K-1} \h{\gamma}_{k}^{(r)} \h{\gamma}_{k}^{(r)\prime}$, where $\boldsymbol{W}_t$ and $\h{\gamma}_{k}$ are defined in Corollaries~\ref{lse_ar1_multi}. Let $\Xi_3 =: \boldsymbol{H}^{-1} \mathbb{E}(\boldsymbol{W}_t {\Sigma}^{-1} \boldsymbol{W}_t^{\prime}) \boldsymbol{H}^{-1}$. Assume the same conditions as Theorem~\ref{mle_arp_multi}. It holds that
\begin{equation*}
 \sqrt{T}\begin{pmatrix} 
vec(\tilde{\hA}_{1}^{(11)} - \hA_{1}^{(11)}) \\
\cdots \\
vec(\tilde{\hA}^{(pR_p)}_{K} - \hA^{(pR_p)}_K)
\end{pmatrix} \to \mathcal{N}(0,\Xi_3).
\end{equation*}
\end{corollary}

\section{Proof of the Theorems}\label{appendix:proofs}
\subsection{Proof of Proposition~\ref{kruskal}}
\begin{proof}[Proof of Proposition~\ref{kruskal}]
When $K=2$, \eqref{TAR1} becomes a multi-term MAR model,
and \eqref{eq:Phi_CP} corresponds to the singular value decomposition
of $\mathcal R(\Phi)$. The identifiability conditions are implied by the uniqueness conditions of the singular value decomposition. When $K \ge 3$, the classical results on the uniqueness of the tensor CP decomposition suggest that the identifiability of $\hA^{(r)}_{k}$ is granted under the Kruskal’s condition for $K$-mode tensors \citep{sidiropoulos2000uniqueness}.

\end{proof}

\subsection{Proof of Proposition~\ref{symmetric_prosig}}
\begin{proof}[Proof of Proposition~\ref{symmetric_prosig}]
Without loss of generality, we consider the case $K=3$. Claim (i) is simply the direct result of \cite{van1993approximation}. To prove Claim (ii), We have form (\ref{kronecker cov}) that $\Sigma = \Sigma_3 \otimes \Sigma_2 \otimes \Sigma_1$ and after rearrangement $\mathcal{R}(\Sigma) = \vect(\Sigma_1) \vect(\Sigma_3 \otimes \Sigma_2)^{\prime}$. Since $\vect(\hat{\Sigma}_1)$ is the first left singular vector of $\mathcal{R}(\hat{\Sigma})$, by the sin $\Theta$ theorems \citep{davis1970rotation, wedin1972perturbation}, we have
$$\sqrt{1-\left(\vect({\Sigma_1})^{\prime} \vect(\hat{\Sigma}_1)\right)^2} \le \frac{\max\left\{\|\h{Z} \hat{\h{V}} \|_F, \|\hat{\h{U}}^{\prime} \h{Z} \|_F\right\}}{\delta},$$
where $\hat{\h{U}} \in \mathbb{R}^{d_1^2}$ is the first left singular vector and $\hat{\h{V}} \in \mathbb{R}^{d_2^2 d_3^2}$ is the first right singular vector of $\mathcal{R}(\hat{\Sigma})$, $\delta = \|\mathcal{R}(\hat{\Sigma})\|_s$ and $\h{Z} = \hat{\Sigma} - \Sigma$. By Proposition~\ref{sig_convergence} we know that $\|\h{Z}\|_s = \|\hat{\Sigma} - \Sigma\|_s = o_p(1)$ for fixed dimension. Also, $\|\delta\|_s = \|\mathcal{R}({\Sigma}) + \mathcal{R}({\h{Z}})\|_s \ge \|\mathcal{R}({\Sigma})\|_s - \|\mathcal{R}(\h{Z})\|_s = \|\Sigma\|_F - \|\mathcal{R}(\h{Z})\|_s > 0$ with probability one. So $\sqrt{1-\left(\vect({\Sigma_1})^{\prime} \vect(\hat{\Sigma}_1)\right)^2} = o_p(1)$ and therefore
$$\|\hat{\Sigma}_1 - \Sigma_1\|_F = \|\vect(\hat{\Sigma}_1) - \vect(\Sigma_1)\|_F = \sqrt{2\left(1-\vect(\Sigma_1)^{\prime} \vect(\hat{\Sigma}_1)\right)} = o_p(1).$$
Similarly, we have $\|\hat{\Sigma}_3 \otimes \hat{\Sigma}_2 - \Sigma_3 \otimes \Sigma_2\|_F = o_p(1)$. By the conditions that $\|\Sigma_2\|_F = 1$ and $\|\hat{\Sigma}_3\|_F = 1$, it can be shown that $\|\hat{\Sigma}_2 - \Sigma_2\|_F =o_p(1)$ and $\|\hat{\Sigma}_3 - \Sigma_3\|_F =o_p(1)$. 

\end{proof}

\subsection{Proof of Theorem~\ref{lse_arp_multi}}
To prove Theorem~\ref{lse_arp_multi}, \ref{mle_arp_multi} and Theorem~\ref{proj_arp_multi}, we will use the results of following Lemma~\ref{hessian}. Recall in Section~2 we introduced the factor matrices $\mathbb{A}_{k} = [\ha_k^{(1)},\cdots,\ha_k^{(R)}] \in \mathbb{R}^{d_k^2 \times R}$, $1 \le k \le K$. We assume the minimum singular value of these factor matrices does not vanish as dimension $d_k$ goes to infinity, i.e. $s_{\min}(\mathbb{A}_{k}) > 0$. 

\begin{remark}
This assumption is reasonable and quite general. First, it satisfied with some simplest cases, for example, in matrix case $K=2$ it meets the assumption since the $\ha_k^{(r)}$, $r=1,\cdots,R$, are orthogonal to each other. Second, in random setting, say if coefficient matrices are generated by IID Gaussian, then it satisfied with probability one.
\end{remark}

\begin{lemma}\label{hessian}
Let $\bar{\hA}_k^{(ir)}$ be $d_k \times d_k$ matrices that $\|\bar{\hA}_K^{(ir)}\|_F=1$, $i=1,\cdots,p$, $r=1,\cdots,R_i$, $k=1,\cdots,K-1$, and $\|\bar{\Phi}^{(i)} -\Phi^{(i)}\|_F \neq 0$. Then we have $\|\bar{\hA}_k^{(ir)} - \hA_k^{(ir)}\|_F \le C \sum_{i=1}^{p}\|\bar{\Phi}^{(i)} -\Phi^{(i)}\|_F$ where $C$ is a constant.
\end{lemma}
\begin{proof}
It is sufficient to show that for each $i$ we have $\|\bar{\hA}_k^{(ir)} - \hA_k^{(ir)}\|_F \le C \|\bar{\Phi}^{(i)} -\Phi^{(i)}\|_F$. We fix $i$, we denote $R_i$ as $R$ for simplicity, and let $\ha_r = \vect(\hA_{K-1}^{(ir)} \otimes \cdots \otimes \hA_{1}^{(ir)})$ and $\boldsymbol{b}_r = \vect(\hA_{K}^{(ir)})$, $r=1,\cdots,R$. $\h{B}=[\boldsymbol{b}_1, \cdots, \boldsymbol{b}_{R}]$ and $\h{A} =[\ha_1, \cdots, \ha_{R}]$. $\h{B}$ is mode-$K$ factor matrix that meets our assumption $s_{\min}(\h{B}) \ge c_2 > 0$. Observe that $\h{A}$ is column-wise Khatri-Rao product of factor matrices such that $\hA = \mathbb{A}_{K-1} \odot \cdots \odot \mathbb{A}_1$, which is column sub-matrix of $\mathbb{A}_{K-1} \otimes \cdots \otimes \mathbb{A}_1$. By singular value interlacing theorem [\cite{horn2012}] and $s_{\min}(\mathbb{A}_{k}) > c_1 > 0$, we know,
$$s_{\min}(\hA) \ge s_{\min}(\mathbb{A}_{K-1} \otimes \cdots \otimes \mathbb{A}_1) = \Pi_{k=1}^{K-1} s_{\min}(\mathbb{A}_{k}) > 0$$
Let $f(\bar{\ha}_1,\cdots,\bar{\ha}_R,\bar{\boldsymbol{b}}_1,\cdots,\bar{\boldsymbol{b}}_R) = \|\sum_{r=1}^{R} (\bar{\ha}_r \bar{\boldsymbol{b}}_r^{\prime} -\ha_r \boldsymbol{b}_r^{\prime} )\|_F^2$. We have $\|\bar{\Phi}^{(i)} - \Phi^{(i)} \|_F^2=f(\bar{\ha}_1,\cdots,\bar{\boldsymbol{b}}_R)$. Denote $\boldsymbol{H}(\boldsymbol{x})$ be the Hessian matrix of $f(\boldsymbol{x})$. Consider $\boldsymbol{x}=[\bar{\ha}_1,\cdots,\bar{\boldsymbol{b}}_R]$ as some local parameters around $\boldsymbol{x}_0=[\ha_1,\cdots,\boldsymbol{b}_R]$, $\Delta \boldsymbol{x} = \boldsymbol{x} - \boldsymbol{x}_0$. Note that $f(\boldsymbol{x}_0)$ and $\Delta f(\boldsymbol{x}_0)$ is zero and by the Tyler expansion of $f(\boldsymbol{x})$ at $\boldsymbol{x}_0$, we only need to show that 
\begin{equation}\label{H_claim_1}
    \Delta \boldsymbol{x}^{\prime}\boldsymbol{H}(\boldsymbol{x}_0)\Delta \boldsymbol{x} \ge c \Delta \boldsymbol{x}^{\prime}\Delta \boldsymbol{x}
\end{equation}
for some constant $c$. Let $\beta_{ij} = \boldsymbol{b}_i^{\prime} \boldsymbol{b}_j$ and $\alpha_{ij} = \ha_i^{\prime} \ha_j$. Taking second derivatives of $f(\boldsymbol{x})$ we have
\begin{align}\label{H}
\frac{1}{2}\boldsymbol{H}(\boldsymbol{x}_0) =& \begin{pmatrix} 
\beta_{11} \boldsymbol{I}_{1} &\cdots& \beta_{R1} \boldsymbol{I}_{1} & \ha_1 \boldsymbol{b}_1^{\prime}&\cdots&\ha_R \boldsymbol{b}_1^{\prime} \\
\vdots &&\vdots&\vdots&&\vdots& \\
\beta_{1R} \boldsymbol{I}_{1} &\cdots& \beta_{RR} \boldsymbol{I}_{1} & \ha_1 \boldsymbol{b}_R^{\prime}&\cdots&\ha_R \boldsymbol{b}_R^{\prime} \\
\boldsymbol{b}_1 \ha_1^{\prime}&\cdots&\boldsymbol{b}_R \ha_1^{\prime} & \alpha_{11} \boldsymbol{I}_{2} &\cdots& \alpha_{R1} \boldsymbol{I}_{2} \\
\vdots &&\vdots&\vdots&&\vdots& \\
\boldsymbol{b}_1 \ha_R^{\prime}&\cdots&\boldsymbol{b}_R \ha_R^{\prime} & \alpha_{1R} \boldsymbol{I}_{2} &\cdots& \alpha_{RR} \boldsymbol{I}_{2} 
\end{pmatrix}.
\end{align}    
Before we move on, we define some notations. First, as shown in (\ref{H}) we view $\boldsymbol{H}(\boldsymbol{x}_0)$ as $2R \times 2R$ block matrix, and denote the dimension of $\boldsymbol{H}(\boldsymbol{x}_0)$ as $n$. Let $\boldsymbol{v}_i^j$ be vectors such that the $i$-th block equals $\ha_j$, the $(R+j)$-th block equals $-\boldsymbol{b}_i$ and other places are zero. $\boldsymbol{V} = \spn\{\boldsymbol{v}_i^j,\, i,j=1,\cdots,R\}$. Second, we also view $\boldsymbol{H}(\boldsymbol{x}_0)$ as $2 \times 2$ block matrix
\begin{align}
\frac{1}{2}\boldsymbol{H}(\boldsymbol{x}_0) =& \begin{pmatrix} 
\boldsymbol{H}_1 & \boldsymbol{H}_2 \\
\boldsymbol{H}_3 & \boldsymbol{H}_4
\end{pmatrix}, 
\end{align}
where $\boldsymbol{H}_1$ is top-left $R \times R$ blocks in (\ref{H}). Now we claim that
\begin{equation}\label{H_claim_2}
    \min_{\boldsymbol{\mu} \neq 0,\, \boldsymbol{\mu} \notin \boldsymbol{V}} |\boldsymbol{\mu}^{\prime} \boldsymbol{H}(\boldsymbol{x}_0) \boldsymbol{\mu}| > c > 0.
\end{equation}
Next we prove this claim. By block matrix diagonalization, we have
\begin{equation}
    \rank(\boldsymbol{H}(\boldsymbol{x}_0)) = \rank
        \begin{pmatrix}
        \boldsymbol{H}_1-\boldsymbol{H}_2 \boldsymbol{H}_4^{-1} \boldsymbol{H}_3 & \boldsymbol{0} \\
        \boldsymbol{0} & \boldsymbol{H}_4
        \end{pmatrix}.
\end{equation}
It is sufficient to show that $\rank{\boldsymbol{H}(\boldsymbol{x}_0)} \ge n - R^2$. Since $\boldsymbol{H}_1 = \boldsymbol{B}^{\prime} \boldsymbol{B} \otimes \boldsymbol{I}_{1}$, $\boldsymbol{H}_4 = \hA^{\prime} \hA \otimes \boldsymbol{I}_{2}$ so that
\begin{equation}
    \lambda_{\min}(\boldsymbol{H}_1)=\lambda_{\min}(\boldsymbol{B}^{\prime} \boldsymbol{B}) \ge c_2^2 > 0, \, \lambda_{\min}(\boldsymbol{H}_4)=\lambda_{\min}(\hA^{\prime} \hA) \ge c_1^2 > 0.
\end{equation}
By further simplification we can rewrite the $ij$-th block of $\boldsymbol{H}_2 \boldsymbol{H}_4^{-1} \boldsymbol{H}_3$ in the form
\begin{equation}
    [\boldsymbol{H}_2 \boldsymbol{H}_4^{-1} \boldsymbol{H}_3]_{ij} = \sum_{r=1}^{R} \ha_r \boldsymbol{c}_r^{\prime}
\end{equation}
for some vectors $\boldsymbol{c}_r$, $r=1,\cdots,R$ and $i,j= 1,\cdots,R$. This form is sum of $R$ outer product of two vectors, so $\rank([\boldsymbol{H}_2 \boldsymbol{H}_4^{-1} \boldsymbol{H}_3]_{ij}) \le R$, which implies that $\rank(\boldsymbol{H}_2 \boldsymbol{H}_4^{-1} \boldsymbol{H}_3) \le R^2$. It follows that,
\begin{equation}
\begin{split}
    \rank(\boldsymbol{H}(\boldsymbol{x}_0)) =& \rank(\boldsymbol{H}_4) + \rank(\boldsymbol{H}_1 - \boldsymbol{H}_2 \boldsymbol{H}_4^{-1} \boldsymbol{H}_3))\\
    \ge& \rank(\boldsymbol{H}_4) + \rank(\boldsymbol{H}_1) - \rank(\boldsymbol{H}_2 \boldsymbol{H}_4^{-1} \boldsymbol{H}_3)) \\
    \ge& n- R^2
\end{split}
\end{equation}
It implies the claim (\ref{H_claim_2}). Next we discuss (\ref{H_claim_1}). Observe that for any $\delta > 0$ we have 
\begin{equation}
\sum_{r=1}^{R} \ha_r \boldsymbol{b}_r^{\prime} = \sum_{r \neq i,j} \ha_r \boldsymbol{b}_r^{\prime} +  (\ha_i + \delta \ha_j) \boldsymbol{b}_i^{\prime} + \ha_j (\boldsymbol{b}_j - \delta \boldsymbol{b}_i)^{\prime}.  
\end{equation}
This implies, if $\boldsymbol{x} = \boldsymbol{x}_0 + \delta \boldsymbol{v}_i^{j}$ then $f(\boldsymbol{x}) = f(\boldsymbol{x}_0)$. These directions have to be excluded from $\Delta \boldsymbol{x}$ by the requirement that $\|\bar{\Phi}^{(l)} - \Phi^{(l)}\|_F \neq 0$ for $l=1,\cdots,p$. Thus the (\ref{H_claim_1}) is implied by (\ref{H_claim_2}), which completes our proof.
\end{proof}
We need one more lemma. First state and prove following lemma in the fixed dimension.
\begin{lemma}\label{lemma2}
Consider the VAR($p$) representation of (\ref{multiarp}), and $\Phi^{(i)}$ is defined in Appendix A. Assume the conditions of Theorem~\ref{lse_arp_multi}. Then for any sequence $\{c_T\}$ such that $c_T \rightarrow \infty$,
$$P\left[ \inf_{\sqrt{T} \sum_{i=1}^{p}\|\bar{\Phi}^{(i)} - \Phi^{(i)}\|_F \ge c_T} \sum_{t=p+1}^{T} \|\vect( \mathcal{X}_t) - \sum_{i=1}^{p} \bar{\Phi}^{(i)} \vect( \mathcal{X}_{t-1} )\|_F^2 \le \sum_{t=p+1}^{T} \|\vect(\mathcal{E}_t)\|^2   \right] \rightarrow 0.$$
\end{lemma}
\begin{proof}
By the ergodic theorem
$$\frac{1}{T} \sum_{t=p+1}^{T} \vect(\mathcal{X}_{t-1}) \vect(\mathcal{X}_{t-1})^{\prime} \to \Gamma_0 \quad \as$$
So we have for any constant $c > 0$,
\begin{equation}\label{lemma_2.1}
\begin{split}
       \sup_{\sqrt{T} \sum_{i=1}^{p}\|\bar{\Phi}^{(i)} - \Phi^{(i)}\|_F \ge c}
       \bigg|  &\sum_{t=p+1}^{T} \tr\left[(\bar{\Phi}^{(i)} - \Phi^{(i)}) \vect(\mathcal{X}_{t-i})  \vect(\mathcal{X}_{t-j})^{\prime} (\bar{\Phi}^{(j)} - \Phi^{(j)})^{\prime}\right] \\
      &- T\tr\left[(\bar{\Phi}^{(i)} - \Phi^{(i)}) \Gamma_h (\bar{\Phi}^{(j)} - \Phi^{(j)})^{\prime}\right]  \bigg|  \quad
      \to 0 \quad \as  
\end{split}
\end{equation}
It follows that there exits a sequence $\{c_T^{\prime}\}$ such that $c_T^{\prime} \to \infty$, $c_T^{\prime} \le c_T$ and
\begin{equation}\label{lemma_2.2}
\begin{split}
       \sup_{\sqrt{T} \sum_{i=1}^{p}\|\bar{\Phi}^{(i)} - \Phi^{(i)}\|_F \ge c_T^{\prime}}
       \bigg|  &\sum_{t=p+1}^{T} \tr\left[(\bar{\Phi}^{(i)} - \Phi^{(i)}) \vect(\mathcal{X}_{t-i})  \vect(\mathcal{X}_{t-j})^{\prime} (\bar{\Phi}^{(j)} - \Phi^{(j)})^{\prime}\right] \\
      &- T\tr\left[(\bar{\Phi}^{(i)} - \Phi^{(i)}) \Gamma_h (\bar{\Phi}^{(j)} - \Phi^{(j)})^{\prime}\right]  \bigg|  \quad
      \to 0 \quad \text{in probability.} 
\end{split}
\end{equation}
Now we write
\begin{equation}\label{lemma_2.3}
\begin{split}
      &\sum_{t=p+1}^{T} \|\vect(\mathcal{X}_t) - \sum_{i=1}^{p}\bar{\Phi}^{(i)} \vect(\mathcal{X}_{t-i}) \|^2_{F} - \sum_{t=p+1}^{T} \|\vect(\mathcal{E}_t)\|^2 \\
      =& \sum_{t=p+1}^{T} \sum_{i,j=1}^{p} \tr\left[(\bar{\Phi}^{(i)} - \Phi^{(i)}) \vect(\mathcal{X}_{t-i}) \vect(\mathcal{X}_{t-j})^{\prime} (\bar{\Phi}^{(j)} - \Phi^{(j)})^{\prime}\right] \\
      & - 2 \sum_{t=p+1}^{T} \sum_{i=1}^{p} \tr\left[(\bar{\Phi}^{(i)} - \Phi^{(i)}) \vect(\mathcal{X}_{t-i}) \vect(\mathcal{E}_{t})^{\prime}\right]
\end{split}
\end{equation}
On the boundary set $\sqrt{T} \sum_{i=1}^{p}\| \bar{\Phi}^{(i)} - \Phi^{(i)} \|_{F} = c_{T}^{\prime}$, by calculating the variance, we know that
\begin{equation}\label{lemma_2.4}
      \sum_{t=p+1}^{T} \tr\left[(\bar{\Phi}^{(i)} - \Phi^{(i)}) \vect(\mathcal{X}_{t-i}) \vect(\mathcal{E}_{t})^{\prime}\right] =  O_p(  c_{T}^{\prime})
\end{equation}
On the other hand,
\begin{equation}\label{lemma_2.5}
\begin{split}
      T \tr\left[(\bar{\Phi}^{(i)} - \Phi^{(i)}) \Gamma_{h} (\bar{\Phi}^{(j)} - \Phi^{(j)})^{\prime}\right] \ge  \lambda_{\min}(\Gamma_h) (c_{T}^{\prime})^2 
\end{split}
\end{equation}
where $\lambda_{\min}(\Gamma_h)$ is the minimum eigenvalue of $\Gamma_h$, which is strictly positive under our assumptions. Follows from (\ref{lemma_2.2}) to (\ref{lemma_2.5}), and the fact that $c_{T}^{\prime} \to \infty$, we have
\begin{equation}\label{lemma_2.6}
    P \left[ \inf_{\sqrt{T} \sum_{i=1}^{p} \|\bar{\Phi}^{(i)} - \Phi^{(i)} \|_{F} = c_{T}^{\prime}} \sum_{t=p+1}^{T} \|\vect(\mathcal{X}_t) - \sum_{i=1}^{p}\bar{\Phi}^{(i)} \vect(\mathcal{X}_{t-i}) \|^2_{F} \le \sum_{t=p+1}^{T} \|\vect(\mathcal{E}_t)\|^2  \right] \to 0.
\end{equation}
Observe that $\sum_{t=p+1}^{T} \|\vect(\mathcal{X}_t) - \sum_{i=1}^{p}\bar{\Phi}^{(i)} \vect(\mathcal{X}_{t-i}) \|^2_{F}$ is a convex function of $\bar{\Phi}^{(i)}$, so Lemma~\ref{lemma2} is implied by (\ref{lemma_2.6}).
\end{proof}

\begin{proof}[Proof of Theorem~\ref{lse_arp_multi}]

Without loss of generality, consider multi-term TenAR(1) model where $p=1$. The proof can be extended to TenAR($p$) under same idea. By taking partial derivatives of (\ref{lsearp}) respect to the $k$-mode, $r$-term coefficient matrix $\hA_{k}^{(r)}$, $1 \le r \le R$ and $1 \le k \le K$, we obtain the gradient condition, 
\begin{equation}\label{arpgd}
\sum_{t} \left( \boldsymbol{X}_{t(k)} -  \sum_{j = 1}^{R} \hat{\hA}_{k}^{(j)} \boldsymbol{X}_{t-1,(k)} {\hat{\Phi}^{(j){\prime}}_k} \right) \left( {\hat{\Phi}^{(r)}_k} {\boldsymbol{X}_{t-l,(k)}}^{\prime}  \right) = 0
\end{equation}
  Let $\Phi_k = \sum_{j=1}^{R} \hA_K^{(j)} \otimes \cdots \hA_{k+1}^{(j)} \otimes \hA_{k-1}^{(j)} \cdots \otimes \hA_1^{(j)}$, $k=1,\cdots,K$. By Lemma~\ref{hessian} and Lemma~\ref{lemma2}, we know that $\hat{\hA}_k^{(j)} = \hA_k^{(j)} + O_p(T^{-1/2})$. By gradient condition (\ref{arpgd}) and replacing $\hX_{t(k)} = \sum_{j=1}^{R}\hA_k^{(j)} \hX_{t-1,(k)} \Phi_k^{(j) \prime} + \h{E}_{t(k)}$, we have
\begin{equation} \label{2.1}
\begin{split}
    &\sum_{t} \left[ \sum_{j=1}^{R} (\hat{\hA}_{k}^{(j)} - \hA_{k}^{(j)}) \boldsymbol{X}_{t-1,(k)} \Phi_k^{(j)\prime} \Phi_k^{(r)} \boldsymbol{X}_{t-1,(k)}^{\prime} 
    + \sum_{j=1}^{R} \hA_{k}^{(j)} \boldsymbol{X}_{t-1, (k)} (\hat{\Phi}_k^{(j)} - \Phi_k^{(j)})^{\prime} \Phi_k^{(r)} \boldsymbol{X}_{t-1, (k)}^{\prime} \right] \\
    &= \sum_t \boldsymbol{E}_{t(k)} \Phi_k^{(r)} \boldsymbol{X}_{t-1, (k)}^{\prime} + o_p(\sqrt{T})
    \end{split}
\end{equation}
Let $\underline{M}_{ik}^{(j)} = \hA_K^{(j)} \otimes \cdots \otimes (\hat{\hA}_i^{(j)} - \hA_i^{(j)}) \otimes \cdots \otimes \hA_{k+1}^{(j)} \otimes  \hA_{k-1}^{(j)} \otimes \cdots \otimes \hA_1^{(j)}$ denotes Kronecker product without $\hA_k^{(j)}$ and the difference is taking with $\hat{\hA}_i^{(j)} - \hA_i^{(j)}$. By Proposition~\ref{basic_properties} and Proposition~\ref{pur}, we can take vectorization and re-write it as 
\begin{equation}\label{2.2}
\begin{split}
    \vect(\underline{M}_{ik}^{(j)}) &= \boldsymbol{T}_k \vect(\ha_1^{(j)} \circ \cdots \circ \ha_{k-1}^{(j)} \circ  \ha_{k+1}^{(j)} \circ \cdots \circ (\hat{\ha}_i^{(j)} - \ha_i^{(j)}) \circ  \cdots \ha_K^{(j)})\\
    &= \boldsymbol{T}_k \left(\ha_K^{(j)} \otimes \cdots \otimes (\hat{\ha}_i^{(j)} - \ha_i^{(j)}) \otimes \cdots \otimes \ha_{k+1}^{(j)} \otimes  \ha_{k-1}^{(j)} \otimes \cdots \ha_1^{(j)}\right)\\
    &= \boldsymbol{T}_k \boldsymbol{S}_{(i,k)}\vect(\hat{\hA}_i^{(j)} - \hA_i^{(j)})
\end{split}
\end{equation}
where $\boldsymbol{S}_{(i,k)} = \ha_K^{(j)} \otimes \cdots \ha_{i+1}^{(j)} \otimes \boldsymbol{I}_i \otimes  \ha_{i-1}^{(j)} \otimes \cdots \ha_{k+1}^{(j)} \otimes \ha_{k-1}^{(j)} \otimes \cdots \otimes \ha_1^{(j)}$ such that the term with subscript $i$ is $\boldsymbol{I}_i$ and no term with subscript $k$, and $\boldsymbol{T}_k$ is defined in Proposition~\ref{pur}. Using (\ref{2.2}) we can take vectorization of the second term of LHS in (\ref{2.1}) and re-write it as
\begin{equation} \label{2.3}
\begin{split}
    &\vect\left(\sum_t \sum_{j=1}^{R} \hA_{k}^{(j)} \boldsymbol{X}_{t-1, (k)} (\hat{\Phi}_k^{(j)} - \Phi_k^{(j)})^{\prime} \Phi_k^{(r)} \boldsymbol{X}_{t-1, (k)}^{\prime}\right) \\
    =& \sum_t \sum_{j=1}^{R} \left[\left(\boldsymbol{X}_{t-1, (k)} \Phi_k^{(r)\prime} \otimes  \hA_{k}^{(j)} \boldsymbol{X}_{t-1, (k)} \right) \vect(\sum_{i=1}^{K} \underline{M}_{ik}^{(j)\prime} ) \right] +o_p(\sqrt{T}) \\
    =& \sum_t \sum_{j=1}^{R} \left[ \left(\boldsymbol{X}_{t-1, (k)} \Phi_k^{(r)\prime} \otimes  \hA_{k}^{(j)} \boldsymbol{X}_{t-1, (k)} \right) \left(\sum_{i=1}^{K} \boldsymbol{T}_k \boldsymbol{S}_{(i,k)} \boldsymbol{P}_{d_i,d_i} \vect(\hat{\hA}_i^{(j)} - \hA_i^{(j)})\right) \right] +o_p(\sqrt{T}) 
\end{split}
\end{equation}
where $\boldsymbol{P}_{d_i, d_i}$ is permutation matrix such that has property Proposition~\ref{basic_properties} (vi). Then taking (\ref{2.3}) into (\ref{2.1}) and taking vectorization on both sides of (\ref{2.1}) for $k = 1,\cdots,K$, we have
\begin{equation} \label{2.6}
    \boldsymbol{U} \begin{pmatrix}
\vect(\hat{\hA}_{1}^{(1)} - \hA_1^{(1)}) \\
\vect(\hat{\hA}_{2}^{(1)} - \hA_2^{(1)}) \\
\cdots \\
\vect(\hat{\hA}_K^{(R)} - \hA_K^{(R)}) 
\end{pmatrix}  = \sum_{t} 
\boldsymbol{W}_t  \vect(\mathcal{E}_t) + o_p(\sqrt{T})
\end{equation}
where we can view $\boldsymbol{U}$ as a hierarchy block-matrix such that we first view it as a $R \times R$ block-matrix, the $m,n$-th block denoted as $\boldsymbol{U}^{(mn)}$, $1 \le m,n \le R$. For each block, we view it as a $K \times K$ block-matrix that we denote $i,j$-th block, $1 \le i,j \le K$, as $\boldsymbol{U}^{(mn)}_{(ij)}$, such that 
\begin{align*}
    \text{if} \ i=j, \ \boldsymbol{U}^{(mn)}_{(ii)} &= \left( \sum_t \boldsymbol{X}_{t-1,(i)} \Phi_k^{(m)\prime} \Phi_k^{(n)} \boldsymbol{X}_{t-1,(i)}^{\prime}  \right) \otimes \boldsymbol{I}_{d_i}
    \\
    \text{if} \ i\neq j,\  \boldsymbol{U}^{(mn)}_{(ij)} &= \sum_t \left(\boldsymbol{X}_{t-1, (i)} \Phi_i^{(m)\prime} \otimes  \hA_{i}^{(n)} \boldsymbol{X}_{t-1, (i)} \right) \boldsymbol{T}_K \boldsymbol{S}_{(j,i)} \boldsymbol{P}_{d_j,d_j}
\end{align*}
We claim that (\ref{2.6}) can be re-written as,
\begin{equation}\label{2.8}
     \boldsymbol{U} \begin{pmatrix}
\vect(\hat{\hA}_{1}^{(1)} - \hA_1^{(1)}) \\
\vect(\hat{\hA}_{2}^{(1)} - \hA_2^{(1)}) \\
\cdots \\
\vect(\hat{\hA}_{K}^{(R)} - \hA_K^{(R)}) 
\end{pmatrix}=
 \sum_t \boldsymbol{W}_{t-1} \boldsymbol{W}_{t-1}^{\prime} \begin{pmatrix}
\vect(\hat{\hA}_{1}^{(1)} - \hA_1^{(1)}) \\
\vect(\hat{\hA}_{2}^{(1)} - \hA_2^{(1)}) \\
\cdots \\
\vect(\hat{\hA}_{K}^{(R)} - \hA_K^{(R)}) 
\end{pmatrix}.
\end{equation}
Consider the same block partition as $\boldsymbol{U}$, $i,j$-th block in the $m,n$-th block of $\boldsymbol{W}_{t-1} \boldsymbol{W}_{t-1}^{\prime}$ would be denoted as $(\boldsymbol{W}_{t-1} \boldsymbol{W}_{t-1}^{\prime})_{(ij)}^{(mn)}$. First, note that for $i=j, 1\le i,j \le K$, by Proposition~\ref{basic_properties} in the Appendix, we have $\boldsymbol{Q}_i \boldsymbol{Q}_i^{\prime} = \boldsymbol{I}_{d_1\cdots d_K}$, so
\begin{equation*}
\begin{split}
    \sum_t (\boldsymbol{W}_{t-1} \boldsymbol{W}_{t-1}^{\prime})_{(ii)}^{(mn)} &= \sum_t \left((\boldsymbol{X}_{t-1,(i)} \Phi_i^{(m) \prime}) \otimes \boldsymbol{I}_{d_i} \right) \boldsymbol{Q}_i \boldsymbol{Q}_i^{\prime}
    \left((\boldsymbol{X}_{t-1,(i)} \Phi_i^{(n)\prime}) \otimes \boldsymbol{I}_{d_i}  \right)^{\prime} \\
    &= \left( \sum_t \boldsymbol{X}_{t-1,(i)} \Phi_i^{(m) \prime} \Phi_i^{(n)} \boldsymbol{X}_{t-1,(i)} ^{\prime} \right) \otimes \boldsymbol{I}_{d_i} = \boldsymbol{U}_{(ii)}^{(mn)}
\end{split}
\end{equation*}
Next we verify that (\ref{2.8}) holds for $i\neq j$,
\begin{align} \label{2.9}
\begin{split}
& \sum_t (\boldsymbol{W}_{t-1} \boldsymbol{W}_{t-1}^{\prime})_{(ij)}^{(mn)}  \vect(\hat{\hA}_{j}^{(n)} - \hA_{j}^{(n)}) \\
=& \sum_t \left((\boldsymbol{X}_{t-1,(i)} \Phi_i^{(m) \prime}) \otimes \boldsymbol{I}_{d_i} \right) \boldsymbol{Q}_i \boldsymbol{Q}_j^{\prime}
    \left(  ( \Phi_j^{(n)} \boldsymbol{X}_{t-1,(j)}^{\prime}) \otimes \boldsymbol{I}_{d_j} \right) \vect(\hat{\hA}_{j}^{(n)} - \hA_{j}^{(n)}) \\
=& \sum_t \left((\boldsymbol{X}_{t-1,(i)} \Phi_i^{(m) \prime}) \otimes \boldsymbol{I}_{d_i} \right) \boldsymbol{Q}_i \boldsymbol{Q}_j^{\prime} \vect\left((\hat{\hA}_{j}^{(n)} - \hA_{j}^{(n)}) \hX_{t-1,(j)} \Phi_j^{(n)\prime}\right)\\
=& \sum_t \left((\boldsymbol{X}_{t-1,(i)} \Phi_i^{(m) \prime}) \otimes \boldsymbol{I}_{d_i} \right)  \vect\left( \hA_{i}^{(n)} \boldsymbol{X}_{t-1,(i)} \underline{M}_{ji}^{(n) \prime} \right)
\end{split}
\end{align}
On the other hand,
\begin{align}\label{2.10}
\begin{split}
    & \boldsymbol{U}_{(ij)}^{(mn)} \vect(\hat{\hA}_{j}^{(n)} - \hA_{j}^{(n)})\\
    =& \sum_t \left(\boldsymbol{X}_{t-1, (i)} \Phi_i^{(m)\prime} \otimes \hA_i^{(n)} \boldsymbol{X}_{t-1, (i)} \right) \boldsymbol{T}_k \boldsymbol{S}_{(j,i)} \boldsymbol{P}_{d_j,d_j}  \vect(\hat{\hA}_{j}^{(n)} - \hA_{j}^{(n)}) \\
    =& \sum_t \left((\boldsymbol{X}_{t-1,(i)} \Phi_i^{(m) \prime}) \otimes \boldsymbol{I}_{d_i} \right)  \vect\left( \hA_{i}^{(n)} \boldsymbol{X}_{t-1,(i)} \underline{M}_{ji}^{(n) \prime} \right)
\end{split}
\end{align}
Thus, by (\ref{2.9}) and (\ref{2.10}) we verified that (\ref{2.8}) holds. From (\ref{2.6}) and (\ref{2.8}), we have
\begin{equation} \label{2.11}
    \sum_t \boldsymbol{W}_{t-1} \boldsymbol{W}_{t-1}^{\prime} \begin{pmatrix}
\vect(\hat{\hA}_{1}^{(1)} - \hA_1^{(1)}) \\
\vect(\hat{\hA}_{2}^{(1)} - \hA_2^{(1)}) \\
\cdots \\
\vect(\hat{\hA}_K^{(R)} - \hA_K^{(R)}) 
\end{pmatrix}  = \sum_{t} 
\boldsymbol{W}_{t}  \vect(\mathcal{E}_t) + o_p(\sqrt{T})
\end{equation}
By the ergodic theorem as $\cX_t$ is strictly stationary with IID. innovations under the conditions, we have
\begin{equation*}
    \frac{1}{T} \sum_{t}  \boldsymbol{W}_{t-1} \boldsymbol{W}_{t-1}^{\prime} \to \mathbb{E}(\boldsymbol{W}_t \boldsymbol{W}_t^{\prime}),\quad a.s.
\end{equation*}
Observe that $\mathbb{E}(\boldsymbol{W}_t \boldsymbol{W}_t^{\prime}) (\h{\gamma}_k^{(j)} - \h{\gamma}_K^{(j)}) = 0$ for $1 \le k \le K-1$ and $1 \le j \le R$. We construct a full rank matrix $\boldsymbol{H}$ such that $\boldsymbol{H} := \mathbb{E}(\boldsymbol{W}_t \boldsymbol{W}_t^{\prime}) + \sum_{j=1}^{R} \sum_{i=1}^{K-1} \h{\gamma}_i^{(j)} \h{\gamma}_i^{(j) \prime}$. Since $\|\hA_k^{(j)}\|=\|\hat{\hA}_k^{(j)}\|=1$, it holds that $\ha_k^{(j)\prime} (\vect(\hat{\hA}_k^{(j)} - \hA_k^{(j)})) = o_p(T^{-1/2})$, consequently from (\ref{2.11}) we have
\begin{equation*}
    \boldsymbol{H} \begin{pmatrix}
\vect(\hat{\hA}_{1}^{(1)} - \hA_1^{(1)}) \\
\vect(\hat{\hA}_{2}^{(1)} - \hA_2^{(1)}) \\
\cdots \\
\vect(\hat{\hA}_K^{(R)} - \hA_K^{(R)}) 
\end{pmatrix}  = \frac{1}{T} \sum_{t} \boldsymbol{W}_{t-1} \vect(\mathcal{E}_t) + o_p(T^{-1/2})
\end{equation*}
By martingale central limit theorem
\begin{equation*}
    \frac{1}{T} \sum_{t} \boldsymbol{W}_{t-1} \vect(\mathcal{E}_t) \Rightarrow \mathcal{N}(0, \mathbb{E}(\boldsymbol{W}_t \Sigma \boldsymbol{W}_t^{\prime})).
\end{equation*}
Thus, it holds that
\begin{equation*}
 \sqrt{T}\begin{pmatrix}
\vect(\hat{\hA}_{1}^{(1)} - \hA_1^{(1)}) \\
\vect(\hat{\hA}_{2}^{(1)} - \hA_2^{(1)}) \\
\cdots \\
\vect(\hat{\hA}_K^{(R)} - \hA_K^{(R)}) 
\end{pmatrix} \to \mathcal{N}(0,\Xi_2)
\end{equation*}
where $\Xi_2 =: \boldsymbol{H}^{-1} \mathbb{E}(\boldsymbol{W}_t \Sigma \boldsymbol{W}_t^{\prime}) \boldsymbol{H}^{-1}$.
\end{proof}

\subsection{Proof of Theorem~\ref{mle_arp_multi}}
The proof is almost the same as Theorem 4 in \cite{chen2020autoregressive}. Since the argument is based on the vectorized model, it can be applied to our Theorem~\ref{mle_arp_multi} without difficulty, we omit the proof. The idea is that first prove the consistency that $\hA_{k}^{(ir)} = \tilde{\hA}_{k}^{(ir)} + O_p(T^{-1/2})$, $\hat{\Sigma}_k =\Sigma_k + o_p(1)$, for $i=1,\cdots,p$, $r=1,\cdots,R_i$, $k=1,\cdots,K$. Then we prove CLT by the argument similar with that of Theorem~\ref{lse_arp_multi}.

\subsection{Proof of Proposition~\ref{phi_convergence}}

To prove Proposition~\ref{phi_convergence}, we first state the following Lemmas from \cite{vershynin2018high}.

\begin{lemma}[Theorem 4.7.1 in \cite{vershynin2018high}]\label{HDL1} 
Let $\hX_i$ be independent sub-gaussian random vectors in $\mathbb{R}^N$, $i = 1,\cdots, T$, $\hat{\Sigma} = \frac{1}{T} \sum_{i=1}^{T} \hX_i \hX_i^{\prime}$ and $\Sigma = \mathbb{E}(\hat{\Sigma})$. We have 
$$\mathbb{E}\|\hat{\Sigma} - \Sigma\|_{s} \le C (\sqrt{\frac{N}{T}} + \frac{N}{T}) \|\Sigma\|_{s}$$
where $C$ is an absolute constant.
\end{lemma}

\begin{corollary}[Exercise 4.7.3 in \cite{vershynin2018high}]
Let $\hX_i$ be independent sub-gaussian random vectors in $\mathbb{R}^N$, $i = 1,\cdots, T$. For any $u \ge 0$. We have 
$$\|\hat{\Sigma} - \Sigma\|_{s} \le C (\sqrt{\frac{N +u}{T}} + \frac{N +u}{T}) \|\Sigma\|_{s}$$
with probability at least $1- 2e^{-u}$, where $C$ is an absolute constant.
\end{corollary}
We extend above results to the autoregressive model, which is the following Proposition.

\begin{proposition}\label{sig_convergence}
Let $\hX_i$ be random vectors in $\mathbb{R}^N$ with stationary AR(1) structures, $i = 1,\cdots, T$. More precisely, $\hX_{i+1} = \Phi \hX_{i} + \xi_i$, assume that $\xi_i$ are IID sub-gaussian random vectors in $\mathbb{R}^{N}$. Let $\hat{\Sigma} = \frac{1}{T} \sum_{i=1}^{T} \hX_i \hX_i^{\prime}$ and $\Sigma = \mathbb{E}(\hat{\Sigma})$. We have 
$$\mathbb{E}\|\hat{\Sigma} - \Sigma\|_{s} \le C (\sqrt{\frac{N\log{N}}{T}} + \frac{N\log{N}}{T}) \|\Sigma\|_{s}$$
where $C$ is a constant.
\end{proposition}
\begin{proof}
We can write $\hX_i$ as $\hX_i = \sum_{j=0}^{\infty} \Phi^{j} \xi_{i-j}$ and denote its truncation form $\tilde{\hX}_i = \sum_{j=0}^{m} \Phi^{j} \xi_{i-j}$, where $\|\Phi\| = r < 1$ since it's a stationary process. Let $\tilde{\Sigma} = \frac{1}{T} \sum_{i=1}^{T} \tilde{\hX}_i \tilde{\hX}_i^{\prime}$. First we claim that as long as $m \ge \log{N}$, we have
$$\mathbb{E} \|\hat{\Sigma} - \tilde{\Sigma}\|_{s} \to 0$$
To prove this claim, let $\tilde{\h{Y}}_i = \hX_i - \tilde{\hX}_i$, $i = 1,\cdots,T$, which is $\tilde{\h{Y}}_i = \sum_{j=m+1}^{\infty} \Phi^{j} \xi_{i-j}$. So that
$\|\tilde{\h{Y}}_i\|_s = \|\sum_{j=m+1}^{\infty} \Phi^{j} \xi_{i-j}\|_s \le r^{m} \|\sum_{j=0}^{\infty} \Phi^{j} \xi_{i-j}\|_s \le r^{m} \sqrt{N} c_1$ with high probability since $\xi_i$ are sub-gaussian random vectors, where $c_1$ is an absolute constant. Similarly $\|\tilde{\hX}_i\|_s \le c_1 \sqrt{N}$ with high probability. Thus, we have
\begin{align}
    \begin{split}
       \mathbb{E} \|\hat{\Sigma} - \tilde{\Sigma}\|_{s} 
        &= \mathbb{E} \|\frac{1}{T}\sum_{i=1}^{T} (\tilde{\hX}_i \tilde{\hX}_i^{\prime} - \hX_i \hX_i^{\prime})\|_{s} \\
        &=  \mathbb{E} \|\frac{1}{T}\sum_{i=1}^{T} (\tilde{\hX}_i \tilde{\h{Y}}_i^{\prime} + \tilde{\h{Y}}_i \tilde{\hX}_i^{\prime} + \tilde{\h{Y}}_i \tilde{\h{Y}}_i^{\prime})\|_{s}\\
        &\le \frac{1}{T}\sum_{i=1}^{T} \left(\mathbb{E}\|\tilde{\hX}_i \tilde{\h{Y}}_i^{\prime}\|_s + \mathbb{E}\|\tilde{\h{Y}}_i \tilde{\hX}_i^{\prime}\|_s + \mathbb{E}\|\tilde{\h{Y}}_i \tilde{\h{Y}}_i^{\prime}\|_s\right) \\
        &\le c_2 r^m N
    \end{split}
\end{align}
since $0 < r < 1$, it goes to zero as long as $m \ge \log{N}$. Thus the claim has been proved.

Next, let $\tilde{\Sigma}_j = \frac{1}{h} \sum_{k=1}^{h} \tilde{\hX}_{j + km} \tilde{\hX}_{j + km}^{\prime}$ where $h$ is the integer part of $T/m$, $j = 1,\cdots,m$.
$$\mathbb{E}\|\hat{\Sigma} - \Sigma\|_s \le \mathbb{E}\left(\|\hat{\Sigma} - \tilde{\Sigma}\|_s + \|\tilde{\Sigma} - \mathbb{E}(\tilde{\Sigma})\|_s + \|\mathbb{E}(\tilde{\Sigma}) - \mathbb{E}(\hat{\Sigma})\|_s \right)$$
The first and third part on the right side of above formula goes to zero by the claim we have just proved. Thus we consider the second term,
\begin{align}
    \begin{split}
        \mathbb{E} \|\tilde{\Sigma} - \mathbb{E}(\tilde{\Sigma})\|_s
        &\le \mathbb{E} \left(\frac{1}{m}\|\sum_{j=1}^{m} (\tilde{\Sigma}_j - \mathbb{E}(\tilde{\Sigma}_j))\|_s \right)\\
        &\le  \mathbb{E} \left(\frac{1}{m}\sum_{j=1}^{m} \|\tilde{\Sigma}_j - \mathbb{E}(\tilde{\Sigma}_j)\|_s \right)\\
        &\le C \left(\sqrt{\frac{N\log{N}}{T}} + \frac{N\log{N}}{T}\right) \left(\frac{1}{m} \sum_{j=1}^{m}\|\mathbb{E} (\tilde{\Sigma}_j)\|_s\right)
    \end{split}
\end{align}
The last inequality follows from the Lemma~\ref{HDL1} since $\{\tilde{\hX}_i\}$ are independent sub-gaussian random vectors. Since $\mathbb{E} (\tilde{\Sigma}_j) = \mathbb{E} (\tilde{\Sigma}_i)$, $i \neq j$, we finish our proof by
$$\frac{1}{m} \sum_{j=1}^{m}\|\mathbb{E} (\tilde{\Sigma}_j)\|_s = \|  \mathbb{E} (\frac{1}{m} \sum_{j=1}^{m} \tilde{\Sigma}_j)\|_s = \|\mathbb{E} (\hat{\Sigma})\|_s = \|\Sigma\|_s$$
\end{proof}

\begin{proposition}\label{arp_sig}
Let $\hX_i$ be random vectors in $\boldsymbol{R}^N$ with stationary AR(p) structures, $i = 1,\cdots, T$. More precisely, $\hX_{i} = \sum_{l=1}^{p} \Phi^{(l)} \hX_{i-l} + \xi_{i}$, assume that $\xi_i$ are IID sub-gaussian random vectors in $\boldsymbol{R}^{N}$. Let $\hat{\Gamma}_s = \frac{1}{T} \sum_{i=1}^{T} \hX_{i+s} \hX_i^{\prime}$ and $\Gamma_s = \mathbb{E}(\hat{\Gamma_s})$. for $s=1,\cdots, p$ we have 
$$\mathbb{E}\|\hat{\Gamma}_s - \Gamma_s\|_{s} \le C (\sqrt{\frac{N\log{N}}{T}} + \frac{N\log{N}}{T}) \|\Sigma\|$$
where $C$ is a constant.
\end{proposition}
\begin{proof}
By causality of the process, we have the representation $\vect(\mathcal{X}_{t}) = \sum_{j=0}^{\infty} \boldsymbol{C}_{j} \vect(\ten{E}_{t-j})$, where $\{\boldsymbol{C}_j\}$ is a sequence of matrices whose components are absolutely summable. Then the extension of the proof in Proposition~\ref{sig_convergence} to AR($p$) case and $\hat{\Gamma}_s = \frac{1}{T} \sum_{i=1}^{T} \hX_{i+s} \hX_i^{\prime}$ is almost straightforward so we omit the details.
\end{proof}

\begin{proof}[Proof of Proposition~\ref{phi_convergence}]
Under the conditions of Theorem~\ref{lse_arp_multi} and $d\log(d)/T \to 0$. Consider the VAR($p$) representation. Then for any sequence $\{C_{NT}\}$ such that $C_{NT} \to \infty$ as $N,T \to \infty$,
\begin{equation}\label{arp_lemma_0}
    P\bigg[ \inf_{\sqrt{T/N} \sum_{i=1}^{p} \|\bar{\Phi}^{(i)} - \Phi^{(i)} \|_{s} \ge C_{NT}} \sum_{t=p+1}^{T} \|\vect(\mathcal{X}_t) - \sum_{i=1}^{p}\bar{\Phi}^{(i)} \vect(\mathcal{X}_{t-i}) \|^2_{F} \le \sum_{t=p+1}^{T} \|\vect(\mathcal{E}_t)\|^2  \bigg] \to 0
\end{equation}

Let $\hat{\Gamma}_s = \frac{1}{T} \sum_{t=s+2}^{T} \vect(\mathcal{X}_{t-1}) \vect(\mathcal{X}_{t-s-1})^{\prime}$ and $\Gamma_s = \mathbb{E}(\hat{\Gamma}_s)$, $s=0,\cdots,p-1$. By Proposition~\ref{arp_sig} we know that $\|\hat{\Gamma}_s - \Gamma_s\|_{s} \rightarrow 0$ in probability. Then we can find subsequence $N_k, T_k$ such that $\|\hat{\Gamma}_s - \Gamma_s\|_{s} \rightarrow 0$ almost surely. It follows that for any constant $c>0$, $i,j=1,\cdots,p$, and $h=|i-j|$,
\begin{equation}\label{arp_lemma_1}
\begin{split}
       \sup_{\sqrt{\frac{T}{N}}  \sum_{i=1}^{p} \|\bar{\Phi}^{(i)} - \Phi^{(i)} \|_{s} \le c} \frac{1}{N_k T_k} \bigg|  &\sum_{t=p+1}^{T_k} \tr\left[(\bar{\Phi}^{(i)} - \Phi^{(i)}) \vect(\mathcal{X}_{t-i})  \vect(\mathcal{X}_{t-j})^{\prime} (\bar{\Phi}^{(j)} - \Phi^{(j)})^{\prime}\right] \\
      &- T_k\tr\left[(\bar{\Phi}^{(i)} - \Phi^{(i)}) \Gamma_h (\bar{\Phi}^{(j)} - \Phi^{(j)})^{\prime}\right]  \bigg|  \quad
      \to 0 \quad \as  
\end{split}
\end{equation}
The convergence follows from Proposition~\ref{arp_sig} and $d\log(d) /T \to 0$. As a consequence of (\ref{arp_lemma_1}), there exists a sequence $C_{N_k T_k}$ such that $C_{N_k T_k} \to \infty$, $C_{N_k T_k} \le C_{NT}$ and
\begin{equation}\label{arp_lemma_2}
\begin{split}
      \sup_{\sqrt{\frac{T_k}{N_k}}  \sum_{i=1}^{p} \|\bar{\Phi}^{(i)} - \Phi^{(i)} \|_{s} \le C_{N_k T_k}} \frac{1}{N_k T_k} \bigg| &\sum_{t=p+1}^{T_k} \tr\left[(\bar{\Phi}^{(i)} - \Phi^{(i)}) \vect(\mathcal{X}_{t-i}) \vect(\mathcal{X}_{t-j})^{\prime} (\bar{\Phi}^{(j)} - \Phi^{(j)})^{\prime}\right] \\
      &- T_k\tr\left[(\bar{\Phi}^{(i)} - \Phi^{(i)}) \Gamma_h (\bar{\Phi}^{(j)} - \Phi^{(j)})^{\prime}\right]  \bigg|  \quad
      \to 0 \quad \text{in probability.}
\end{split}
\end{equation}
Now we write
\begin{equation}\label{arp_lemma_3}
\begin{split}
      &\sum_{t=p+1}^{T} \|\vect(\mathcal{X}_t) - \sum_{i=1}^{p}\bar{\Phi}^{(i)} \vect(\mathcal{X}_{t-i}) \|^2_{F} - \sum_{t=p+1}^{T} \|\vect(\mathcal{E}_t)\|^2 \\
      =& \sum_{t=p+1}^{T} \sum_{i,j=1}^{p} \tr\left[(\bar{\Phi}^{(i)} - \Phi^{(i)}) \vect(\mathcal{X}_{t-i}) \vect(\mathcal{X}_{t-j})^{\prime} (\bar{\Phi}^{(j)} - \Phi^{(j)})^{\prime}\right] \\
      & - 2 \sum_{t=p+1}^{T} \sum_{i=1}^{p} \tr\left[(\bar{\Phi}^{(i)} - \Phi^{(i)}) \vect(\mathcal{X}_{t-i}) \vect(\mathcal{E}_{t})^{\prime}\right]
\end{split}
\end{equation}
On the boundary set $\sqrt{\frac{T_k}{N_k}} \sum_{i=1}^{p}\|\bar{\Phi}^{(i)} - \Phi^{(i)} \|_{s} = C_{N_k T_k}$, 
\begin{equation}\label{arp_lemma_4}
\begin{split}
      \sum_{t=p+1}^{T_k} \tr\left[(\bar{\Phi}^{(i)} - \Phi^{(i)}) \vect(\mathcal{X}_{t-i}) \vect(\mathcal{E}_{t})^{\prime}\right] 
      &\le  \|\bar{\Phi}^{(i)} - \Phi^{(i)}\|_{F} \|\sum_{t=p+1}^{T_k} \vect(\mathcal{X}_{t-i}) \vect(\mathcal{E}_t)^{\prime}\|_{F} \\
      &\le O_p( N_k^2 C_{N_k T_k})
\end{split}
\end{equation}
On the other hand,
\begin{equation}\label{arp_lemma_5}
\begin{split}
      T_k \sum_{t=p+1}^{T_k}  \tr\left[(\bar{\Phi}^{(i)} - \Phi^{(i)}) \Gamma_{h} (\bar{\Phi}^{(j)} - \Phi^{(j)})^{\prime}\right] \ge N_k^2  C_{N_k T_k}^2 \lambda_{\min}(\Gamma_h)
\end{split}
\end{equation}
where $\lambda_{\min}(\Gamma_h)$ is the minimum eigenvalue of $\Gamma_h$, which is strictly positive under our assumptions. Follows from (\ref{arp_lemma_2}) to (\ref{arp_lemma_5}), and the fact that order $p$ is fixed, $C_{N_k T_k} \to \infty$, we have
\begin{equation}\label{arp_lemma_6}
    P\left[ \inf_{\sqrt{T_k/N_k} \sum_{i=1}^{p} \|\bar{\Phi}^{(i)} - \Phi^{(i)} \|_{s} = C_{N_k T_k}} \sum_{t=p+1}^{T_k} \|\vect(\mathcal{X}_t) - \sum_{i=1}^{p}\bar{\Phi}^{(i)} \vect(\mathcal{X}_{t-i}) \|^2_{F} \le \sum_{t=p+1}^{T_k} \|\vect(\mathcal{E}_t)\|^2  \right] \to 0
\end{equation}
Since $\sum_{t=p+1}^{T_k} \|\vect(\mathcal{X}_t) - \sum_{i=1}^{p}\bar{\Phi}^{(i)} \vect(\mathcal{X}_{t-i}) \|^2_{F}$ is a convex function of $\bar{\Phi}^{(i)}$, so (\ref{arp_lemma_0}) is implied by (\ref{arp_lemma_6}).
And (\ref{arp_lemma_0}) implies Proposition~\ref{phi_convergence}, which completes the proof.

\end{proof}

\subsection{Proof of Theorem~\ref{select}}

We need one more Lemma to begin our proof.
\begin{lemma}\label{XE}
Under same condition of Theorem~\ref{select}, we have $\frac{1}{T}\|\sum_{t=2}^{T} \vect(\mathcal{X}_{t-1}) \vect(\mathcal{E}_t) \|_s = O_p (\sqrt{\frac{N}{T}})$.
\end{lemma}
\begin{proof}
By causality of the process, we have the representation $\vect(\mathcal{X}_{t-1}) = \sum_{j=0}^{\infty} \boldsymbol{C}_{j} \vect(\ten{E}_{t-1-j})$, where $\{\boldsymbol{C}_j\}$ is a sequence of matrices whose components are absolutely summable. Then we have,
$$\sum_{t=2}^{T} \vect(\mathcal{X}_{t-1}) \vect(\mathcal{E}_{t})^{\prime} = \sum_{t=2}^{T} \left[\sum_{j=0}^{\infty} \boldsymbol{C}_{j} \vect(\ten{E}_{t-1-j})\right] \vect(\mathcal{E}_{t})^{\prime}=  \sum_{j=0}^{\infty} \left[ \boldsymbol{C}_j \left(\sum_{t=2}^{T} \vect(\ten{E}_{t-1-j}) \vect(\mathcal{E}_{t})^{\prime}\right)\right].$$
By the results in \cite{L2015} and \cite{wang2015}, we have 
$$\|\sqrt{\frac{T}{N}} \sum_{t=2}^{T} \frac{\vect(\ten{E}_{t-1-j}) \vect(\mathcal{E}_{t})^{\prime})}{T}\|_s \to c \, \as$$
where $c$ is a universal constant. Thus,
$$\frac{1}{T}\|\sum_{t=2}^{T} \vect(\mathcal{X}_{t-1}) \vect(\mathcal{E}_t) \|_s  \le \sum_{j=0}^{\infty} \boldsymbol{C}_j \|\frac{\sum_{t=2}^{T} \vect(\mathcal{E}_{t-1-j}) \vect(\mathcal{E}_t)^{\prime}}{T} \|_s \le O_p(\sqrt{\frac{N}{T}})$$

\end{proof}

\begin{proof}[Proof of Theorem~\ref{select}]

It is sufficient to show that if $\hat{\h{R}}_{\hat p} \neq \h{R}_{p}$, we have
\begin{equation}\label{4.1}
    \lim_{N,T \to \infty} P(\ic(\hat{\h{R}}_{\hat p}) < \ic(\h{R}_{p})) = 0,
\end{equation}
Denote ${\h{X}}_{i} \in \mathbb{R}^{(T-p-1) \times d}$ as ${\h{X}}_{i}^{\prime} = \left[ \vect(\mathcal{X}_{p-i+1}),\cdots,\vect(\mathcal{X}_{T-i}) \right]$ for $1 \le i \le p$, ${\h{Y}} \in \mathbb{R}^{(T-p-1) \times d}$ as ${\h{Y}}^{\prime} = \left[ \vect(\mathcal{X}_{p+1}), \cdots, \vect(\mathcal{X}_{T+1}) \right]$, and ${\h{E}} \in \mathbb{R}^{(T-p-1) \times d}$ as ${\h{E}}^{\prime} = \left[ \vect(\mathcal{E}_{p+1}), \cdots, \vect(\mathcal{E}_{T+1}) \right]$. Let $V_0 = \frac{1}{dT} \|\h{Y}^{\prime} - \Phi^{(1)} \h{X}_{1}^{\prime} - \cdots - \Phi^{(p)} \h{X}_{ p}^{\prime}\|_{F}^2$ and $V(\hat{\h{R}}_{\hat p}) = \frac{1}{dT} \|\h{Y}^{\prime} - \hat{\Phi}^{(1)} \h{X}_{1}^{\prime} - \cdots - \hat{\Phi}^{(\hat p)} \h{X}_{\hat p}^{\prime}\|_{F}^2$, where $\hat{\Phi}^{(i)} = \sum_{j=1}^{\hat{R}_i} \hat{\hA}_{K}^{(ij)} \otimes \cdots \otimes \hat{\hA}_1^{(ij)}$ are estimated $\Phi^{(i)}$ under given K-rank $\hat{R}_i$, $1 \le i \le \hat{p}$. (\ref{4.1}) is equivalent to
\begin{equation}\label{4.2}
    P\left( \log\left( V(\hat{\h{R}}_{\hat p}) / V(\h{R}_{p}) \right) < g(d,T) (\sum_{i=1}^{p} R_i - \sum_{i=1}^{\hat{p}} \hat{R}_{i})  \right) \to 0.
\end{equation}
Since we allow $R_p = 0$, we can assume $p = P_{\max} \ge \hat{p}$. Replace $\h{Y}^{\prime} = \Phi^{(1)} \h{X}_{1}^{\prime} + \cdots + \Phi^{(p)} \h{X}_{p}^{\prime} + \h{E}^{\prime}$. Denote $\bar{\h{X}} = \{ \hX_1, \cdots, \hX_p \} \in \mathbb{R}^{(T-p-1) \times dp}$, $\bar{\h{E}} = \left[ \h{E}, \h{0} \right] \in \mathbb{R}^{(T-p-1) \times dp}$ and
\begin{equation*}
\Phi - \hat{\Phi} = \begin{pmatrix}
\Phi^{(1)} - \hat{\Phi}^{(1)} & & \\
& \cdots & \\
& & \Phi^{(p)} - \hat{\Phi}^{(p)} 
\end{pmatrix}.
\end{equation*}
For given $\hat{\h{R}}_{\hat p}$, we have 
\begin{align}\label{4.3}
\begin{split}
V(\hat{\h{R}}_{\hat p}) 
=& \frac{1}{dT} \|(\Phi^{(1)} - \hat{\Phi}^{(1)})\h{X}_{1}^{\prime} + \cdots + (\Phi^{(p)} - \hat{\Phi}^{(p)})\h{X}_{p}^{\prime} + \h{E}^{\prime}\|_F^2\\
=& \frac{1}{dT} \| (\Phi - \hat{\Phi}) \bar{\hX}^{\prime} + \bar{\h{E}}^{\prime} \|_F^2\\
=& \frac{1}{dT} \left(\| (\Phi - \hat{\Phi}) \bar{\hX}^{\prime}\|_F^2 + 2 \tr\left( (\Phi - \hat{\Phi}) \bar{\hX}^{\prime} \bar{\h{E}}\right) + \|\bar{\h{E}} \|_F^2 \right)
\end{split}
\end{align}

\textbf{Case 1}: If $\hat{R}_i \ge R_i$ for all $1 \le i \le p$, since
\begin{equation}\label{4.7}
    |V(\hat{\h{R}}_{\hat p}) - V(\h{R}_{p})| \le |V(\hat{\h{R}}_{\hat p}) - V_0| + |V_0 - V(\h{R}_{p})|
\end{equation}
By Theorem~\ref{phi_convergence}, we know that when $\hat{R}_i \ge R_i$, we have $\|\Phi^{(i)}- \hat{\Phi}^{(i)}\|_F \le  = O_p(\frac{N}{\sqrt{T}})$. Let $\Sigma = \mathbb{E}(\frac{1}{T} \bar{\hX}^{\prime} \bar{\hX}) = \mathbb{E}(\frac{1}{T} \sum_{i,j=1}^{p}\hX_i^{\prime} \hX_j)$. We have,
\begin{align}\label{4.8}
    \begin{split}
        |V(\hat{\h{R}}_{\hat p}) - V_0| &= \frac{1}{dT} \|(\Phi - \hat{\Phi})\bar{\h{X}}^{\prime}\|_F^2  +  \frac{2}{dT} \text{tr}((\Phi - \hat{\Phi})\bar{\h{X}}^{\prime} \bar{\h{E})} \\
        &\le \frac{1}{d} \lambda_{max}( \frac{1}{T} \bar{\hX}^{\prime} \bar{\hX})  \|\Phi - \hat{\Phi}\|_F^2 + \frac{2}{dT} \|\Phi - \hat{\Phi}\|_{F} \|\bar{\h{X}}^{\prime} \bar{\h{E}}\|_{F}  \\
        &\le \lambda_{max}(\Sigma) O_p(\frac{d}{T}) + O_p (\frac{d}{T})
    \end{split}
\end{align}
The last inequality follows from Proposition~\ref{arp_sig}, Theorem~\ref{phi_convergence} and Lemma~\ref{XE}. Similarly, we have
$|V(\h{R}_{p}) - V_0| \le O_p (\frac{d}{T})$. Thus, taking (\ref{4.8}) into (\ref{4.7}), it follows that,
\begin{equation}\label{4.10}
    |V(\hat{\h{R}}_{\hat p}) - V(\h{R}_{p})| \le O_p (\frac{d}{T})
\end{equation}
This implies that $V(\hat{\h{R}}_{\hat p}) / V(\h{R}_{p}) \le 1 + O_p(\frac{d}{T})$. Thus $\log\left(V(\hat{\h{R}}_{\hat p}) / V(\h{R}_{p}) \right) \le O_p(\frac{d}{T})$. However, we have that $g(N,T) > O_p(\frac{d}{T})$. Thus, the (\ref{4.2}) holds, which implies (\ref{4.1}) holds.

\textbf{Case 2}: If exists $\hat{R}_i < R_i$ for some $1 \le i \le p$. Under the Assumption \ref{fullrank}, we can lower bound the gap between wrong $\hat{\Phi}^{(i)}$ and true $\Phi^{i}$. $\|\Phi^{(i)} - \hat{\Phi}^{(i)}\|_{F}^2 = \|\sum_{j=\hat{R}_{i}+1}^{R_i} (\hA_{K}^{(ij)} \otimes \cdots \otimes \hA_{1}^{(ij)})\|_{F}^2 \ge (R_i-\hat{R}_i)\eta^2 d$. By Proposition~\ref{arp_sig},
\begin{align}
\begin{split}
    \frac{1}{dT} \|(\Phi^{(i)} - \hat{\Phi}^{(i)}) \h{X}_i^{\prime}\|_F^2 &\overset{\text{P}}{\longrightarrow} \frac{1}{d} \text{tr}\left((\Phi^{(i)} - \hat{\Phi}^{(i)}) \mathbb{E}(\frac{1}{T} \hX_i^{\prime} \hX_i) (\Phi^{(i)} - \hat{\Phi}^{(i)})^{\prime} \right) \\
    &\ge \frac{1}{d} \lambda_{\min}\left(\mathbb{E}(\frac{1}{T} \hX_i^{\prime} \hX_i)\right) \|\Phi^{(i)} - \hat{\Phi}^{(i)}\|_{F}^2 \\
    &\ge \lambda_{\min}\left(\mathbb{E}(\frac{1}{T} \hX_i^{\prime} \hX_i)\right) (R_i-\hat{R}_i) \eta^2 > 0.
\end{split}
\end{align}
Note for $i$ such that $\hat{R}_i < R_i$, we have $\|\Phi^{(i)} - \hat{\Phi}^{(i)}\|_{F}^2 \ge (R_i-\hat{R}_i) \eta^2 d$; for $j$ such that $\hat{R}_j \ge R_j$, by Theorem~\ref{phi_convergence} we have $\|\Phi^{(j)} - \hat{\Phi}^{(j)}\|_{F}^2 \le O_p (\frac{d^2}{T})$. Thus,
\begin{align}\label{4.4}
\begin{split}
    \frac{1}{dT} \|(\Phi - \hat{\Phi}) \bar{\h{X}}^{\prime}\|_F^2 &\overset{\text{P}}{\longrightarrow} \frac{1}{d} \text{tr}\left((\Phi - \hat{\Phi}) \Sigma (\Phi - \hat{\Phi})^{\prime}\right) \\
    &\ge \frac{1}{d} \lambda_{\min}(\Sigma) \|\Phi - \hat{\Phi}\|_{F}^2 \\
    &\ge \lambda_{\min}(\Sigma) (R_i-\hat{R}_i) \eta^2 + O_p(\frac{d}{T})> 0.
\end{split}
\end{align}
Again by Lemma~{\ref{XE}}, we have
\begin{equation}\label{4.5}
    \frac{1}{dT} \text{tr}\left( (\Phi - \hat{\Phi})\bar{\h{X}}^{\prime} \bar{\h{E}} \right) \le \frac{1}{dT} \|\Phi - \hat{\Phi}\|_{F} \|\bar{\h{X}}^{\prime} \bar{\h{E}} \|_{F} \le  O_p (\sqrt{\frac{d}{T}}) \to 0.
\end{equation}
The last inequality comes from Lemma~\ref{XE}. Also for true K-ranks $\h{R}_{p}$, we have $V(\h{R}_{p}) \le  \frac{1}{dT} \| E\|_{F}^2$ by the definition of estimators. Thus, taking (\ref{4.4}), (\ref{4.5}) in (\ref{4.3}) we have
\begin{equation}\label{4.6}
    V(\hat{\h{R}}_{\hat p}) - V(\h{R}_{p}) \ge c_1 > 0,
\end{equation}
where $c_1$ is a constant determined by $\lambda_{\min}(\Sigma)$, $R_i-\hat{R}_i$, $p$ and $\eta$. This implies $\log\left(V(\hat{\h{R}}_{\hat{p}}) / V(\h{R}_{p}) \right) > c_2 > 0$, where $c_2$ is a constant determined by $c_1$ and $\frac{1}{dT} \| E\|_{F}^2$. However, we have that $g(d,T) (\sum_{i=1}^{p} R_i - \sum_{i=1}^{\hat{p}} \hat{R}_{i})  \to 0$. So it follows that (\ref{4.2}) holds, which implies (\ref{4.1}) holds.

\end{proof}

\subsection{Proof of Theorem~\ref{proj_arp_multi}}
\begin{proof}[Proof of Theorem~\ref{proj_arp_multi}] Lemma~\ref{hessian} reveals that if $\|\bar{\Phi} -\Phi\|_F = O_p(\frac{1}{\sqrt{T}})$ then we have $\|\bar{\ha}_k - \ha_k\|_F = O_p(\frac{1}{\sqrt{T}})$ in one term TenAR(1) model, where $\ha_k = \vect(\hA_k)$, $k=1,\cdots,K$. Now we are ready to prove the Theorem~\ref{proj_arp_multi}. We let $\Psi = \mathcal{R}(\Phi)= \ha_1 \circ \cdots \circ \ha_K$ and $\mat{\beta}_k = \vect(\ha_1 \circ \cdots \circ \ha_{k-1} \circ \ha_{k+1} \circ \cdots \circ \ha_K )$. Note that $\|\hA_k\|_F = \|\ha_k\|_F = 1$ and recall we also require $\|\bar{\hA}_k\|_F = 1$, for $k < K$. For any $k=1,\cdots,K$, the gradient condition of (\ref{multiprovec}) is given by
\begin{equation}{\label{th1.1}}
    \bar{\mat{a}}_k \bar{\mat{\beta}}_k^{\prime} \bar{\mat{\beta}}_k - \bar{\Psi}_{(k)}\bar{\mat{\beta}}_k = 0.
\end{equation}
Replacing $\bar{\Psi}_{(k)}$ by $\ha_k \mat{\beta}_k^{\prime} + \bar{\Psi}_{(k)} - \ha_k \mat{\beta}_k^{\prime}$ in (\ref{th1.1}), we have
\begin{equation}\label{P1.2}
    (\bar{\ha}_k - \ha_k) \mat{\beta}_k^{\prime}\mat{\beta}_k + \ha_k (\bar{\mat{\beta}}_k - \mat{\beta}_k)^{\prime} \mat{\beta}_k = (\bar{\Psi}_{(k)} - \ha_{k} \mat{\beta}_{k}^{\prime})\mat{\beta}_{k} + o_p (T^{-1/2}).
\end{equation}
Let $\underline{\Psi}_k = \ha_1 \circ \cdots \circ (\bar{\ha}_k - \ha_k) \circ \cdots \circ \ha_K$ and $\bar{\Psi}_{k} = \bar{\ha}_1 \circ \cdots \circ \ha_k \circ \cdots \circ \bar{\ha}_K$, then by Proposition~\ref{basic_properties} (ii) we can rewrite (\ref{P1.2}) as,
\begin{equation*}
    (\underline{\Psi}_k)_{(k)} \beta_k + (\bar{\Psi}_k - \Psi)_{(k)} \mat{\beta}_{k}  = (\bar{\Psi} - \Psi)_{(k)}\mat{\beta}_{k} + o_p (T^{-1/2}) 
\end{equation*}
Since $\|\ha_k\|_F = \|\bar{\ha}_k\|_F = 1$, it follows that $(\bar{\ha}_k - \ha_k)^{\prime} \ha_k = o_p(T^{-1/2})$ and $(\bar{\boldsymbol{\beta}}_K - \boldsymbol{\beta}_K)^{\prime} \boldsymbol{\beta}_K = o_p(T^{-1/2})$, $1 \le k \le K-1$. Then above equations can be  further simplified to,
\begin{equation}\label{P1.3}
\begin{split}
    (\underline{\Psi}_k)_{(k)} \mat{\beta}_k + (\underline{\Psi}_K)_{(k)} \mat{\beta}_{k}  &= (\bar{\Psi} - \Psi)_{(k)}\mat{\beta}_{k} + o_p (T^{-1/2}), \, 1 \le k \le K-1\\
    \bar{\ha}_K - \ha_K &= (\bar{\Psi} - \Psi)_{(K)}\mat{\beta}_{K} + o_p (T^{-1/2})
\end{split}
\end{equation}
Taking the last formula that $k=K$ into the other equation that $k<K$ in (\ref{P1.3}), we have
\begin{equation}\label{P1.4}
    (\bar{\ha}_k - \ha_k) \mat{\beta}_k^{\prime} \mat{\beta}_k = (\bar{\Psi} - \Psi)_{(k)}\mat{\beta}_{k} - \left(\ha_1\circ \cdots \circ \ha_{K-1} \circ \left((\bar{\Psi} - \Psi)_{(K)}\mat{\beta}_{K}\right) \right)_{(k)} \mat{\beta}_{k} + o_p (T^{-1/2}).
\end{equation}
Let $I =  \vect\left((\bar{\Psi} - \Psi)_{(k)}\mat{\beta}_{k}\right)$ and $II = (\mat{\beta}_k^{\prime} \otimes \boldsymbol{I}_{d_k}) \vect\left(\left(\ha_1\circ \cdots \circ \ha_{K-1} \circ ((\bar{\Psi} - \Psi)_{(K)}\mat{\beta}_{K}) \right)_{(k)} \right)$. Taking vectorization on both sides of (\ref{P1.4}), we have
\begin{align}\label{P1.5}
\vect(\hat{\ha}_k - \ha_k) \|\mat{\beta}_k\|^2_{F} 
= I - II + o_p (T^{-1/2}).
\end{align}
Since 
\begin{equation*}
\sqrt{T} \vect( (\bar{\Psi} - \Psi)_{(k)} ) \overset{\text{d}}{\longrightarrow}  N(0, \h{Q}_{k} \Xi_{1} \h{Q}_{k}^{\prime}), 1 \le k \le K,
\end{equation*}
where $\h{Q}_k$ are vector permutation matrices defined in Appendix~{\ref{appendix:basic}} with property in Proposition~{\ref{permutation_properties}}~(i). In (\ref{P1.5}), the first term $I$ is asymptotic normal distributed,
\begin{equation}\label{P1.6}
    \sqrt{T} \cdot I \overset{\text{d}}{\longrightarrow}  (\h{\beta}_k^{\prime} \otimes \boldsymbol{I}_{d_k}) N(0, \h{Q}_{k} \Xi_{1} \h{Q}_{k}^{\prime}).
\end{equation}
Using properties in Proposition~\ref{basic_properties}, after some algebra, we can further simplify II as, 
\begin{equation*}
    II = (\boldsymbol{\beta}_{k}^{\prime} \otimes \ha_k) (\boldsymbol{I}_{d_K} \otimes \h{\beta}_k \boldsymbol{\beta}_K^{\prime})  \boldsymbol{P}_{d_K, d_1\cdots d_{K-1}} \vect\left((\bar{\Psi} - \Psi)_{(K)}\right)
\end{equation*}
Thus, the second term $II$ is also asymptotic normal,
\begin{equation}\label{P1.7}
    \sqrt{T}\cdot II \overset{\text{d}}{\longrightarrow} (\boldsymbol{\beta}_{k}^{\prime} \otimes \ha_k) (\boldsymbol{I}_{d_K} \otimes \boldsymbol{\beta}_k \boldsymbol{\beta}_K^{\prime})  \boldsymbol{P}_{d_K,d_1\cdots d_{K-1}} N(0, \h{Q}_{K} \Xi_{1} \h{Q}_{K}^{\prime})
\end{equation}
By (\ref{P1.4}) to (\ref{P1.7}), we have for $k < K$,
\begin{equation}\label{P1.8}
    \sqrt{T} \vect(\bar{\ha}_k - \ha_k) \overset{\text{d}}{\longrightarrow} \|\boldsymbol{\beta}_k\|_{F}^{-2} \left( (\boldsymbol{\beta}_k^{\prime} \otimes \boldsymbol{I}_{d_k}) \boldsymbol{Q}_k - (\boldsymbol{\beta}_{k}^{\prime} \otimes \ha_k) (I_{d_K} \otimes \boldsymbol{\beta}_k \boldsymbol{\beta}_K^{\prime})  \boldsymbol{P}_{d_K,d_1\cdots d_{K-1}} \boldsymbol{Q}_K \right) N(0,\Xi_1) 
\end{equation}
For $k=K$, by the last formula in (\ref{P1.3}), we have
\begin{equation}\label{P1.9}
    \sqrt{T} \vect(\bar{\ha}_K - \ha_K) \overset{\text{d}}{\longrightarrow} (\boldsymbol{\beta}_K^{\prime} \otimes \boldsymbol{I}_{d_K}) \boldsymbol{Q}_K N(0,\Xi_1) 
\end{equation}
The theorem follows from (\ref{P1.8}) and (\ref{P1.9}).

\end{proof}
\section{Additional Simulations}\label{Appendix:sim}

Following are additional simulation results of Simulation I mentioned in Section~\ref{sec:simulations}. From Figure~\ref{R1_TenAR(1)_IID} to \ref{R1_TenAR(1)_MLE}, the true model is one-term TenAR(1) under different settings and we compare four estimators, one-term PROJ, one-term LSE, one-term MLE and VAR. From Figure \ref{R1_TenAR(2)_IID} to \ref{R1_TenAR(2)_MLE}, the true model is one-term TenAR(2) under different settings. For each figure, the order of tensor we simulate increases from top to bottom, taking values in $(d_1, d_2, d_3) = (3,3,3)$, $(4,4,4)$, and $(5,5,5)$. The length of time series increases from left to right as $T=300$, $500$, $800$, $5000$. We can easily see that the PROJ, LSE, MLE estimators outperform the VAR estimators in all figures and the have similar patterns discussed in Section~\ref{sec:simulations}.

\begin{figure}[!ht]
    \centering
    \includegraphics[width = 16cm]{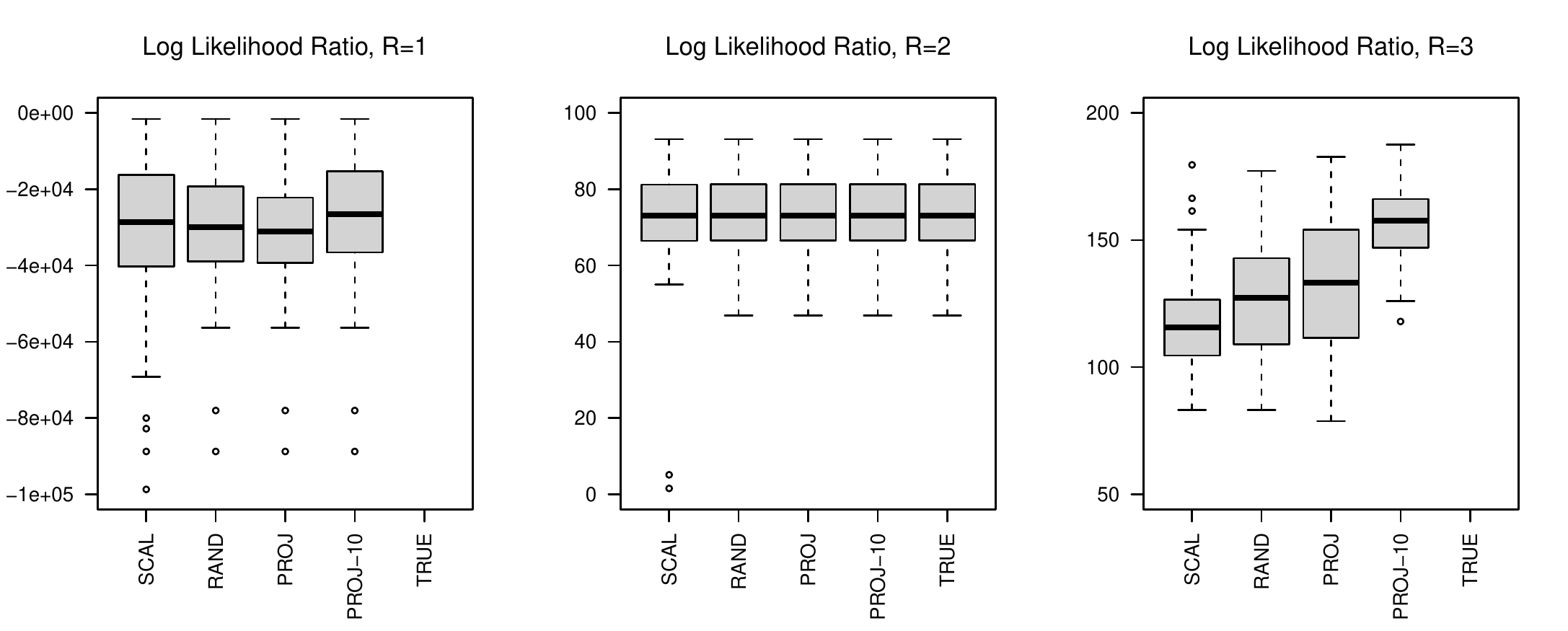}
    \caption{Log likelihood of LSE iterative estimation with different initial values, scaler matrices (SCAL), projection estimator (PROJ), best of $10$ initial values in a neighborhood of PROJ (PROJ-10) and true values (TRUE). True model is two-term ($R=2$) TenAR(1) model, $(d_1,d_2,d_3) = (5,5,5)$, $T=1000$, under Setting I.}\label{fig: likelihood_mle_iid}
\end{figure}

\begin{figure}[!ht]
    \centering
    \includegraphics[width = 16cm]{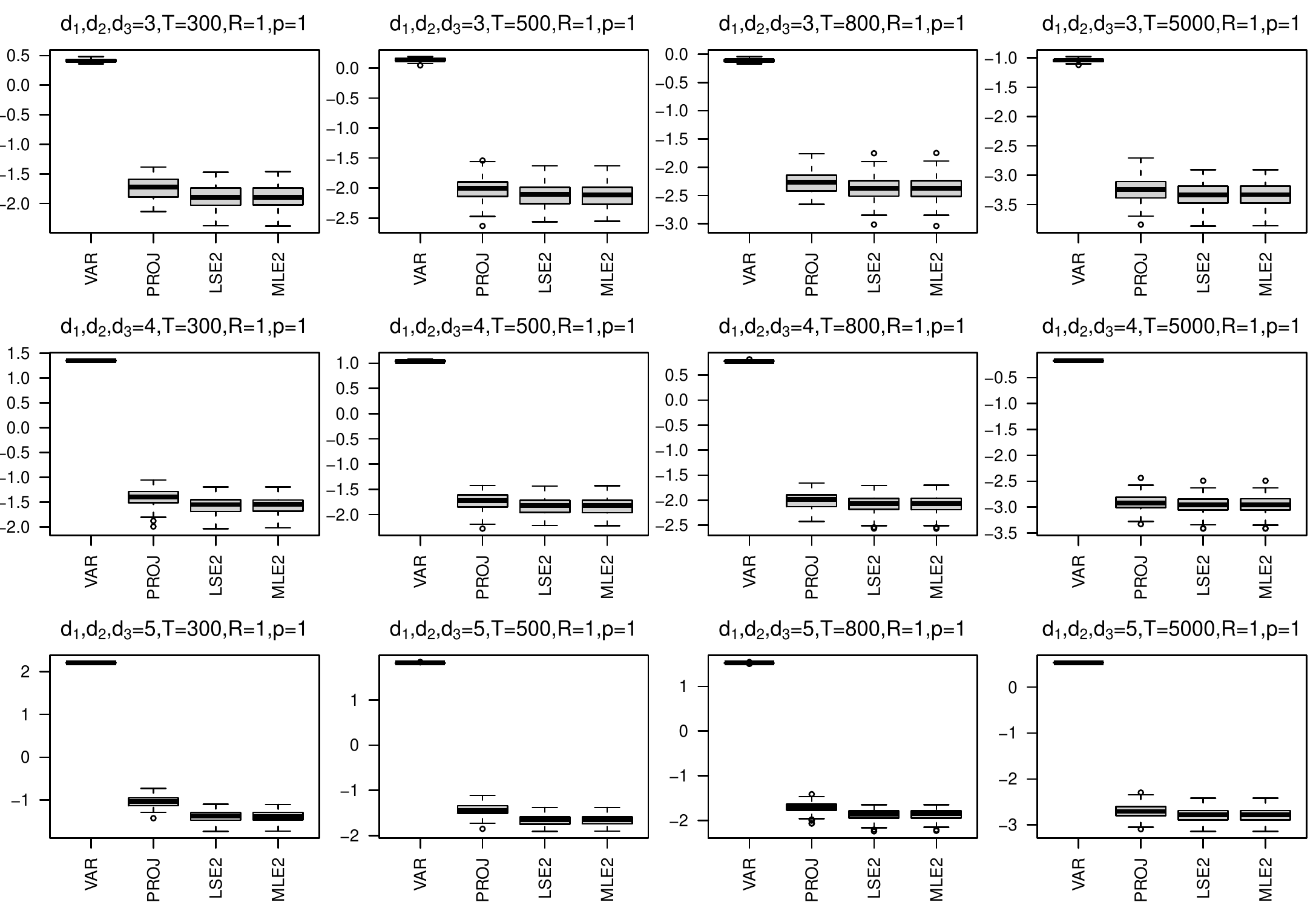}
    \caption{Estimation errors in the log scale. True model is one-term TenAR(1) under setting~I. Comparison of VAR, PROJ, LSE, MLE. We repeat the simulation 100 times. For each row, we fixed the dimension while let $T=300, 500, 800, 5000$. For each column, $T$ is fixed while $(d_1, d_2, d_3) = (3,3,3), (4,4,4), (5,5,5)$.}
    \label{R1_TenAR(1)_IID}
\end{figure}

\begin{figure}[!ht]
    \centering
    \includegraphics[width = 16cm]{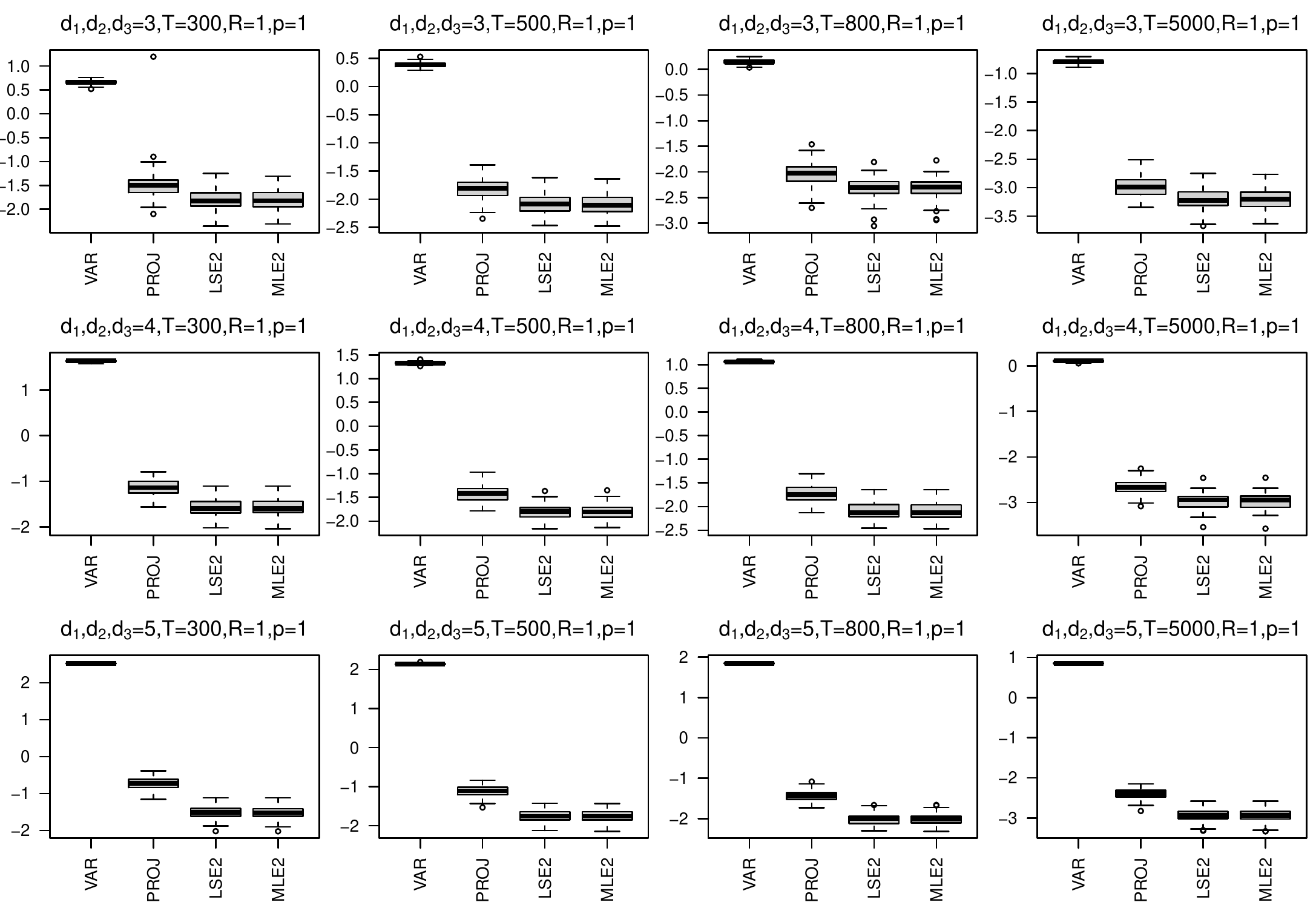}
    \caption{Estimation errors in the log scale. True model is one-term TenAR(1) under setting~II. Comparison of VAR, PROJ, LSE, MLE. We repeat the simulation 100 times. For each row, we fixed the dimension while let $T=300, 500, 800, 5000$. For each column, $T$ is fixed while $(d_1, d_2, d_3) = (3,3,3), (4,4,4), (5,5,5)$.}
    \label{R1_TenAR(1)_SVD}
\end{figure}

\begin{figure}[!ht]
    \centering
    \includegraphics[width = 16cm]{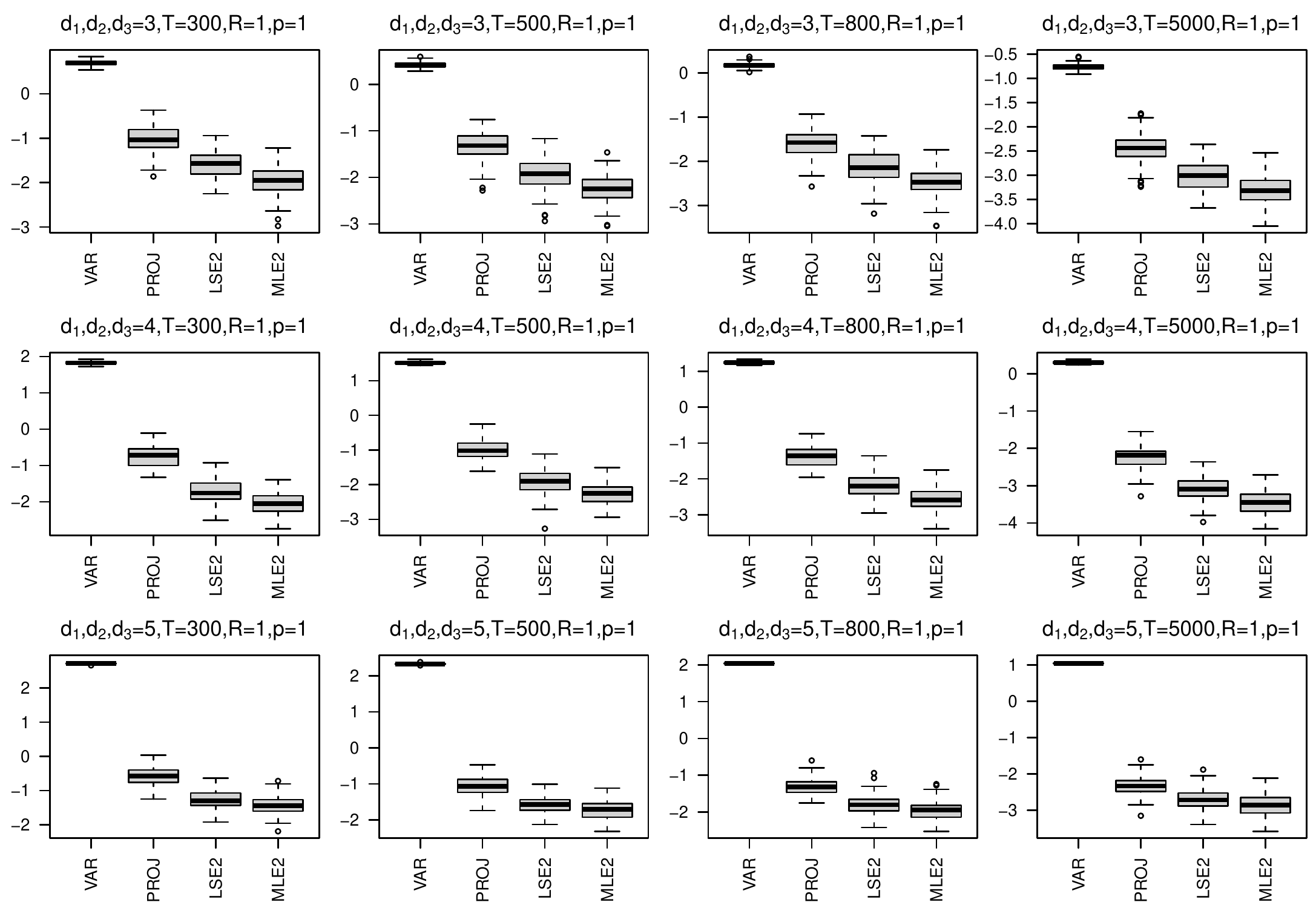}
    \caption{Estimation errors in the log scale. True model is one-term TenAR(1) under setting~III. Comparison of VAR, PROJ, LSE, MLE. We repeat the simulation 100 times. For each row, we fixed the dimension while let $T=300, 500, 800, 5000$. For each column, $T$ is fixed while $(d_1, d_2, d_3) = (3,3,3), (4,4,4), (5,5,5)$.}
    \label{R1_TenAR(1)_MLE}
\end{figure}

\begin{figure}[!ht]
    \centering
    \includegraphics[width = 16cm]{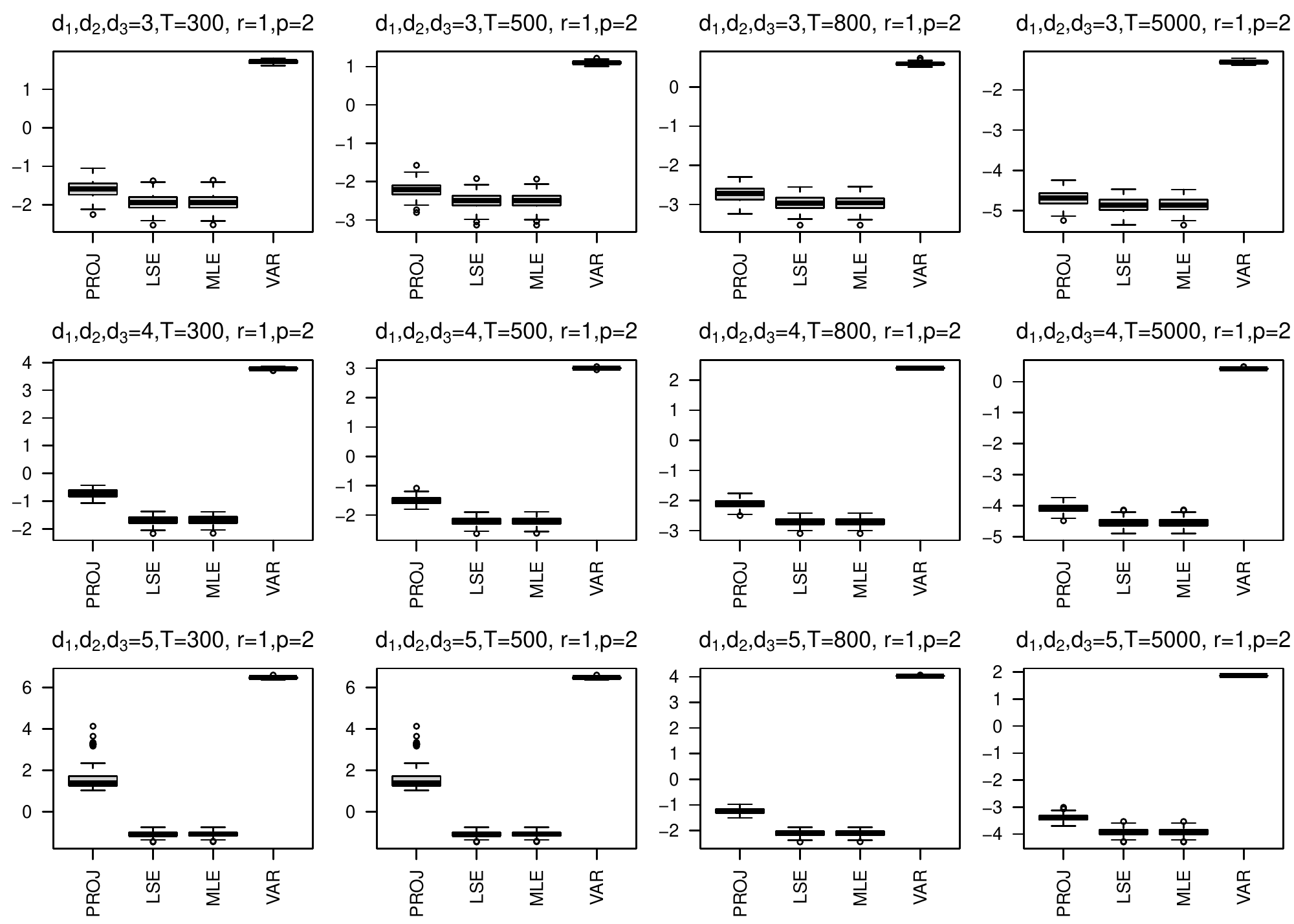}
    \caption{Estimation errors in the log scale. True model is TenAR(2) with $\h{R} = (1,1)$ under setting~I. Comparison of VAR, PROJ, LSE, MLE. We repeat the simulation 100 times. For each row, we fixed the dimension while let $T=300, 500, 800, 5000$. For each column, $T$ is fixed while $(d_1, d_2, d_3) = (3,3,3), (4,4,4), (5,5,5)$.}
    \label{R1_TenAR(2)_IID}
\end{figure}

\begin{figure}[!ht]
    \centering
    \includegraphics[width = 16cm]{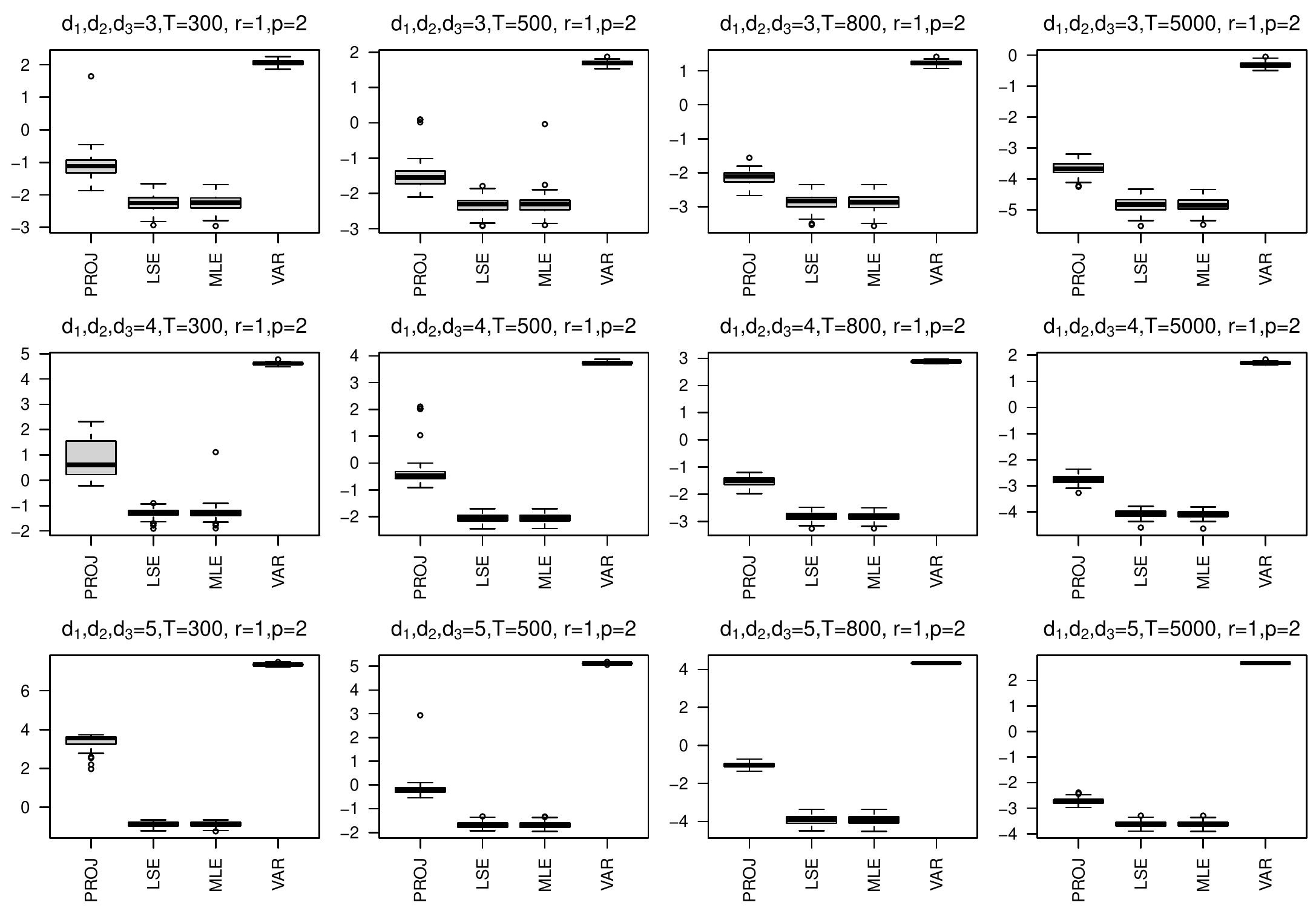}
    \caption{Estimation errors in the log scale. True model is TenAR(2) with $\h{R} = (1,1)$ under setting~II. Comparison of VAR, PROJ, LSE, MLE. We repeat the simulation 100 times. For each row, we fixed the dimension while let $T=300, 500, 800, 5000$. For each column, $T$ is fixed while $(d_1, d_2, d_3) = (3,3,3), (4,4,4), (5,5,5)$.}
    \label{R1_TenAR(2)_SVD}
\end{figure}

\begin{figure}[!ht]
    \centering
    \includegraphics[width = 16cm]{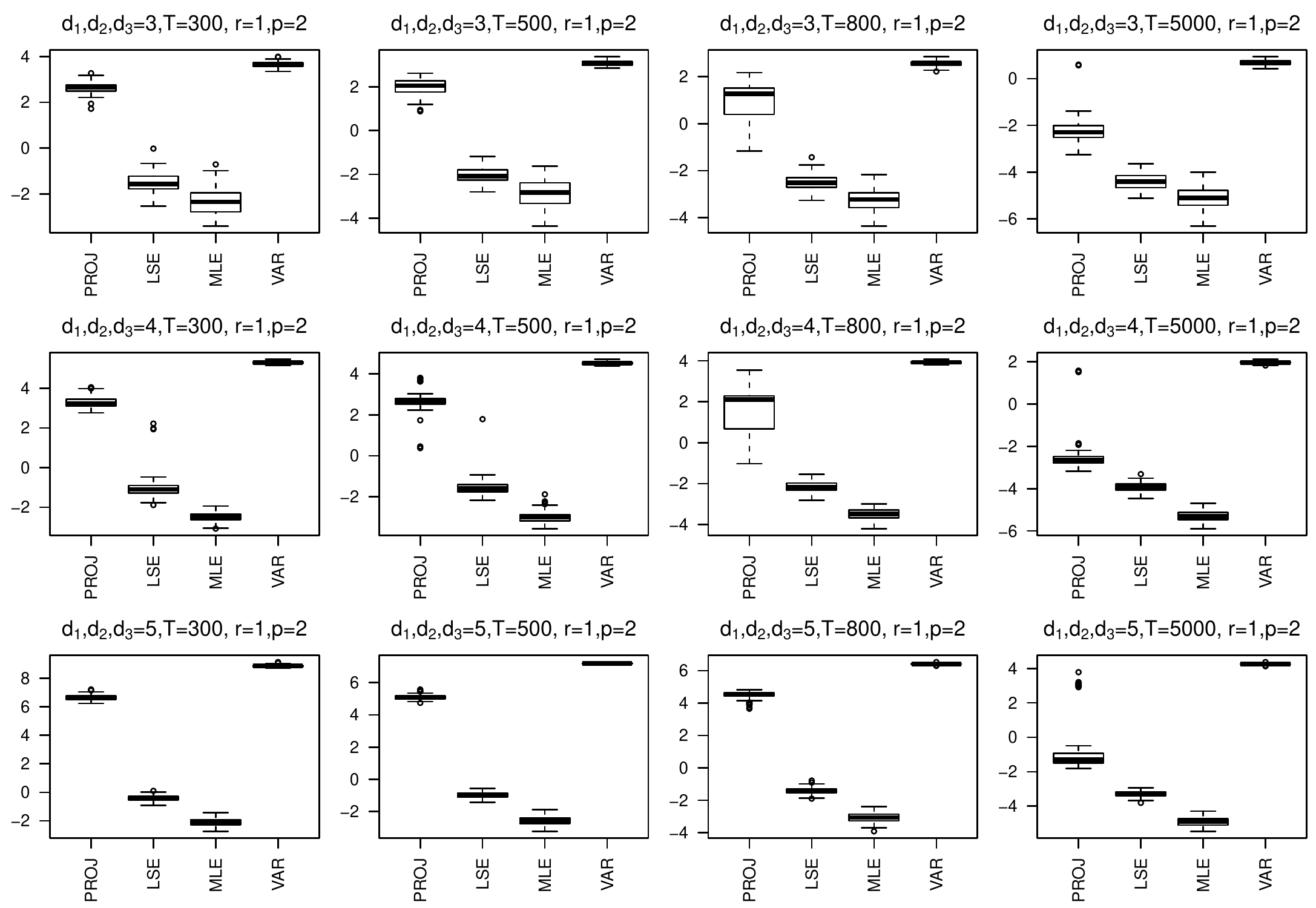}
    \caption{Estimation errors in the log scale. True model is TenAR(2) with $\h{R} = (1,1)$ under setting~III. Comparison of VAR, PROJ, LSE, MLE. We repeat the simulation 100 times. For each row, we fixed the dimension while let $T=300, 500, 800, 5000$. For each column, $T$ is fixed while $(d_1, d_2, d_3) = (3,3,3), (4,4,4), (5,5,5)$.} 
    \label{R1_TenAR(2)_MLE}
\end{figure}


\begin{table}[ht]
\centering
\begin{tabular}{@{}rrrrrcrrr@{}}
\toprule
                       &       & \multicolumn{3}{c}{R=2} & \phantom{abc}& \multicolumn{3}{c}{R=3} \\ \cmidrule(l){3-5} \cmidrule(l){7-9}
                       &       &  T=200  &  T=500  & T=1000 &&  T=200  &  T=500  & T=1000 \\ \midrule
\multirow{3}{*}{II-1}   &(3,3,3)&0.89  &0.94   &0.99  &&0.63   &0.87   &0.92  \\
                       &(4,4,4)&0.97   &1      &1     &&0.84   &0.96   &1     \\
                       &(5,5,5)&0.98   &1      &1     &&0.92   &0.99   &1     \\ \midrule
\multirow{3}{*}{II-2}   &(3,3,3)&0.89  &0.92   &0.99  &&0.61   &0.87   &0.92  \\
                       &(4,4,4)&0.98   &1      &1     &&0.91   &0.98   &1     \\
                       &(5,5,5)&0.98   &1      &1     &&0.99   &0.99   &1     \\ \midrule
\multirow{3}{*}{III-1} &(3,3,3)&0.75   &0.90   &0.96  &&0.56   &0.82   &0.91  \\
                       &(4,4,4)&0.75   &0.88   &0.95  &&0.62   &0.86   &0.99  \\
                       &(5,5,5)&0.79   &0.97   &0.99  &&0.64   &0.93   &0.97  \\ \midrule
\multirow{3}{*}{III-2} &(3,3,3)&0.79   &0.90   &0.96  &&0.50   &0.80   &0.91  \\
                       &(4,4,4)&0.51   &0.76   &0.88  &&0.48   &0.79   &0.93  \\
                       &(5,5,5)&0.28   &0.58   &0.89  &&0.24   &0.61   &0.78  \\ \bottomrule
\end{tabular}
\caption{The empirical frequencies that the true number of terms is selected by the information criteria ${\ic}_1$ in (\ref{ic1}), and ${\ic}_2$ in (\ref{ic2}), out of 100 repetitions. The signal strength $\rho = 0.5$. II-1 and II-2 represent results of ${\ic}_1$ and ${\ic}_2$ under setting II. III-1 and III-2 represent results of ${\ic}_1$ and ${\ic}_2$ under setting III. }\label{choose_2}
\end{table}

\end{appendices}

\end{document}